\documentclass{lmcs} %%% last changed 2014-08-20
\pdfoutput=1

\usepackage{lastpage}
\lmcsdoi{16}{3}{13}
\lmcsheading{}{\pageref{LastPage}}{}{}%
{Aug.~01,~2019}{Aug.~28,~2020}{}

\usepackage[utf8]{inputenc}

\pdfoutput=1

\keywords{call-by-need, classical logic, duality, polarity, sequent calculus, type isomorphism}

\usepackage{microtype}
\usepackage{scalerel}
\usepackage{xcolor}
\usepackage{graphicx}
\usepackage{url}
\usepackage{amsmath}
\usepackage{amsthm}
\usepackage{amssymb}
\usepackage{stmaryrd}
\usepackage{thm-restate}
\usepackage{mathtools}
\usepackage{cmll}
\usepackage{esvect}
\usepackage{proof}
\usepackage[inline,shortlabels]{enumitem}
\usepackage{hyperref}
\usepackage[noabbrev,capitalize]{cleveref}
\usepackage{lmcspreamble}

\newcommand{\C}{\mathcal{C}}
\newcommand{\D}{\mathcal{D}}

\mathtoolsset{centercolon}

\newcommand{\citep}{\cite}
\newcommand{\citet}{\cite}

\newcommand{\paul}[1]{}%{{\bf PD:} {\color{red} #1} {\bf End of PD}}
\newcommand{\zena}[1]{}%{{\bf ZA:} {\color{red} #1} {\bf End of ZA}}

\newcommand{\hi}[2][black!15]{{\setlength{\fboxsep}{1pt}\kern-1pt{\colorbox{#1}{$#2$}}}}

\begin{document}

\title{Compiling With Classical Connectives}
\titlecomment{A previous version of this work appeared in \cite{DownenAriola18BP}.}

\author[P.~Downen]{Paul Downen}	%required
\address{University of Oregon}	%required
\email{\{pdownen,ariola\}@cs.uoregon.edu}  %optional
%\thanks{thanks 1, optional.}	%optional

\author[Z.M.~Ariola]{Zena M. Ariola}	%optional
%\address{University of Oregon}	%optional
%\email{ariola@cs.uoregon.edu}  %optional
%\thanks{thanks 2, optional.}	%optional

% \author[C.~Name3]{Carla Name3}	%optional
% \address{address 3}	%optional
% \urladdr{name3@url3\quad\rm{(optionally, a web-page can be specified)}}  %optional
% \thanks{thanks 3, optional.}	%optional

%% etc.

%% required for running head on odd and even pages, use suitable
%% abbreviations in case of long titles and many authors:

%%%%%%%%%%%%%%%%%%%%%%%%%%%%%%%%%%%%%%%%%%%%%%%%%%%%%%%%%%%%%%%%%%%%%%%%%%%

%% the abstract has to PRECEDE the command \maketitle:
%% be sure not to issue the \maketitle command twice!

\begin{abstract}
  The study of polarity in computation has revealed that an ``ideal''
  programming language combines both call-by-value and call-by-name evaluation;
  the two calling conventions are each ideal for half the types in a programming
  language.  But this binary choice leaves out call-by-need which is used in
  practice to implement lazy-by-default languages like Haskell.  We show how the
  notion of polarity can be extended beyond the value/name dichotomy to include
  call-by-need by adding a mechanism for sharing which is enough to compile a
  Haskell-like functional language with user-defined types.  The key to
  capturing sharing in this mixed-evaluation setting is to generalize the usual
  notion of polarity ``shifts:'' rather than just two shifts (between positive
  and negative) we have a family of four dual shifts.

  We expand on this idea of logical duality---``and'' is dual to ``or;'' proof
  is dual to refutation---for the purpose of compiling a variety of types.
  Based on a general notion of data and {\co}data, we show how classical
  connectives can be used to encode a wide range of built-in and user-defined
  types.  In contrast with an intuitionistic logic corresponding to pure
  functional programming, these classical connectives bring more of the pleasant
  symmetries of classical logic to the computationally-relevant, constructive
  setting.  In particular, an involutive pair of negations bridges the gulf
  between the wide-spread notions of parametric polymorphism and abstract data
  types in programming languages.  To complete the study of duality in
  compilation, we also consider the dual to call-by-need evaluation, which
  shares the computation within the control flow of a program instead of
  computation within the information flow.
\end{abstract}

\maketitle

\section{Introduction}
\label{sec:introduction}

Finding a universal intermediate language suitable for compiling and optimizing
both strict and lazy programs has been a long-sought holy grail for compiler
writers.  First, there was continuation-passing style (CPS)
\cite{Plotkin1975CBNCBVLC,Appel1992CWC}, which hard-codes the evaluation
strategy into the program itself.  In CPS, all the details of an evaluation
strategy can be understood by just looking at the syntax of the program.
Second, there were monadic languages \cite{Moggi1989CLM,peytonjones1997henk},
that abstract away from the concrete continuation-passing into a general monadic
sequencing operation. Besides moving away from continuations, making them an
optional rather than mandatory part of sequencing, they make it easier to
incorporate other computational effects by picking the appropriate monad for
those effects.  Third, there were adjunctive languages
\cite{Levy2001PhD,Zeilberger2009PhD,MunchMaccagnoni2013PhD}, as seen in
polarized logic and call-by-push-value $\lambda$-calculus, that mix both
call-by-name and -value evaluation inside a single program. Like the monadic
approach, adjunctive languages make evaluation order explicit within the terms
and types of a program, and can easily accommodate effects.  However, adjunctive
languages also enable more reasoning principles, by keeping the advantages of
inductive call-by-value data types, as seen in their denotational semantics.
For example, the denotation of a list is just a list of values, not a list of
values interspersed with computations that might diverge or cause side effects.

Each of these developments have focused only on call-by-value and -name
evaluation, but there are other evaluation strategies out there.  For example,
to efficiently implement laziness, the Glasgow Haskell Compiler (GHC) uses a
core intermediate language which is call-by-need
\cite{AMOFW1995CBNLC,AriolaFelleisen1997CBNLC} instead of call-by-name: the
computation of named expressions is shared throughout the lifetime of their
result, so that they need not be re-evaluated again.  This may be seen as merely
an optimization of call-by-name, but it is one that has a profound impact on the
other optimizations the compiler can do.  For example, full extensionality of
functions (\ie the $\eta$ law) does not apply in general, due to issues
involving divergence and evaluation order.  Furthermore, call-by-need is not
just a mere optimization but a full-fledged language choice when effects are
introduced \cite{AHS2011CCBND,MM2019ECBPV}: call-by-need and -name are
observationally different.  This difference may not matter for pure functional
programs, but even there, effects \emph{become} important during compilation.
For example, it is beneficial to use join points \cite{MDAPJ2017CWC}, which is a
limited form of jump or \emph{goto} statement, to optimize pure functional
programs.  So it seems like the quest for a common intermediate language for
strict and lazy languages that can accommodate effects is still ongoing.

To explore the space of intermediate languages, we use a dual sequent calculus
framework which provides a machine-like language suitable for representing
functional programs \cite{DMAPJ2016SCCIL}.  We begin by reviewing this general
framework; specifically:
\begin{itemize}
\item (\Cref{sec:intro-sequent}) We introduce the syntax and reduction theory of
  $\lmtm$---a term assignment for Gentzen's sequent calculus LK
  \cite{Gentzen1935UULS1}.  We discuss the loss of confluence and
  extensionality, and how to solve both issues by imposing a global
  \emph{evaluation strategy}.
\item (\Cref{sec:intro-polarity}) As an alternative to a global evaluation
  strategy, we discuss a local solution to restore both confluence and
  extensionality using \emph{polarity}.
\item (\Cref{sec:CBVCBNintoL}) We show how polarity \emph{subsumes} a global
  evaluation strategy through embeddings of call-by-name and call-by-value
  calculi into the polarized system, and discuss how a different formulation of
  the polarity shifting connectives can impact the directness and efficiency of
  the compiled code.
\item (\Cref{sec:intro-need}) To finish the review, we illustrate how sharing
  can be captured in the sequent calculus through \emph{classical call-by-need}
  evaluation and its dual.
\end{itemize}
We then proceed with the goal of integrating these different evaluation
strategies together into one cohesive, core calculus. We consider the following
questions involving evaluation order and types:
\begin{itemize}
\item How can all four evaluation strategies (call-by-value, -name, -need, and
  -{\co}need) be combined in the same calculus?  What \emph{shifts} are needed
  to convert between strategies?
\item What core connectives are needed to serve as a compile target?  Compiling
  data types is routine, but what is needed to compile \emph{{\co}data types}
  \cite{Hagino1987TLCWCTC}?
\item What is the general notion of user-defined types
  that can encompass the connectives of both intuitionistic and classical logic?
\item How can user-defined data and {\co}data types
  be compiled into the connectives of the core language, while making sure that
  the compiled program reflects the full semantics of the source program?
\item How can we apply the global compilation
  technique in a more local manner instead, as an encoding used for
  optimization?  What is the appropriate notion of correctness to ensure that
  the encoding has exactly the same semantics as the original program, even
  though the types have changed and may be abstracted and hidden from certain
  parts of the program?  How do we ensure that localized applications of the
  encoding is robust, even in the presence of computational effects?
\end{itemize}
This paper answers each of these questions with the following contributions:
\begin{itemize}
\item (\Cref{sec:polarizingCBNeed}) We introduce a family of four shifts,
  centered around the fundamental call-by-value and call-by-name evaluation
  strategies, which are capable of correctly compiling the sharing expressed by
  call-by-(\co)need.
\item (\Cref{sec:dual-core}) We give a dual core calculus, System $\DUAL$, that
  has a fixed, finite set of core connectives and can mix all of the above four
  evaluation strategies within one program.
\item (\cref{sec:dual-ext}) We then extend that core calculus with user-defined
  data and {\co}data types that is fully dual to one another.
\item (\cref{sec:dual-compile}) Our notion of user-defined types allows for both
  parametric polymorphism (in the case of {\co}data types) and a form of
  abstract data types (\ie existentially quantified data types), and we show how
  those defined types can be compiled into the core connectives.
\item (\cref{sec:encoding-type-iso}) Finally, we show that this encoding is
  correct using \emph{type isomorphisms}, which serve as the dual purpose of
  establishing a correspondence between types and their encodings, as well as
  illustrating the strong isomorphisms of the core $\DUAL$ connectives that
  still hold despite the possibility of side effects.
\end{itemize}
% We relate our intermediate language with polarity and call-by-push-value
% (\cref{sec:related-work,sec:core-calculi}).  Finally, full proofs of
% correctness and standard meta-theoretic properties are provided
% (\cref{sec:operational-semantics,sec:correspondence,sec:encoding-soundness,sec:rewriting-theory}).
%%% Functional calculus
% To illustrate how this framework might be applied more directly to functional
% programming, we give a $\lambda$-calculus-based intermediate language that
% corresponds to the functional sub-language of our sequent calculus in
% \cref{sec:fun-core,sec:fun-ext}.

\section{Computation in the Classical Sequent Calculus:
  \texorpdfstring{$\lmtm$}{lambda-mu-tilde-mu}}
\label{sec:intro-sequent}

In the realm of logic, Gentzen's sequent calculus \cite{Gentzen1935UULS1}
provides a clear window to view the symmetry and duality of classical deduction.
And through the Curry-Howard correspondence \cite{Wadler2015PAT}, the sequent
calculus can be seen as a programming language for bringing out many symmetries
that are otherwise hidden in computation \cite{DownenAriola2018CCLSC}.  For
example, there are dual calling conventions like call-by-value versus
call-by-name \cite{CurienHerbelin2000DC,Wadler2003CBVDCBN} and dual programming
constructs like functions versus structures
\cite{Zeilberger2008OUD,MunchMaccagnoni2009FCR}.

One of the ways this computational duality is achieved is by recognizing the two
diametrically opposed forces at the heart of every program: the production and
the consumption of information.  These two forces are formally given a name in
the syntax of the sequent calculus: \emph{terms} (denoted by $v$) produce an
output whereas \emph{{\co}terms} (denoted by $e$) consume an input.  For
example, a boolean value and a function (written as $\fn x v$ in the
$\lambda$-calculus) are both terms since they produce a result, whereas an
if-then-else construct or a call stack (written as $\app v e$) are both
{\co}terms since they are methods for using an input.

From this viewpoint, computation is modeled as the interaction between a
producer and a consumer in the composition $\cut{v}{e}$ called a \emph{command}
(and denoted by $c$).  Operationally, a command represents the current state of
an abstract machine, where the term $v$ acts as program code and the {\co}term
$e$ acts as a continuation.  In this way, the sequent calculus serves as a
high-level language for abstract machines \cite{Ariola2009SCAM}.  However, this language immediately
reveals the tension between producers and consumers in computation.

\subsection{The fundamental dilemma of computation}
\label{sec:dilemma}

\begin{figure}
\centering
\begin{align*}
  c &::= \cut{v}{e}
  \\
  v &::= x \Alt \outp \alpha c \Alt \fnc x \alpha c \Alt \inl v \Alt \inr v
  \\
  e &::= \alpha \Alt \inp x c \Alt \app v e \Alt \casesum{x}{c_1}{y}{c_2}
\end{align*}
\begin{align*}
  \rewriterule{\betamu}
  {
    \cut{\outp\alpha c}{e} &\red c\subst{\alpha}{e}
  }
  &
  \rewriterule{\betar[\to]}
  {
    \cut{\fnc x \alpha c}{\app v e} &\red c\subs{\asub{x}{v},\asub{\alpha}{e}}
  }
  \\
  \rewriterule{\betatmu}
  {
    \cut{v}{\inp x c} &\red c\subst{x}{v}
  }
  &
  \rewriterule{\betar[\oplus]}
  {
    \cut{\inj{i}v}{\casesum{x_1}{c_1}{x_2}{c_2}} &\red c_i\subst{x_i}{v}
  }
\end{align*}
\caption{The $\lmtm$-calculus with sums.}
\label{fig:basic-lmtm}
\end{figure}

To see why there is a conflict between producers and consumers, consider the
small sequent-based language $\lmtm$ in \cref{fig:basic-lmtm}.  There are
\emph{variables} $x$ representing an unknown term and symmetrically
\emph{{\co}variables} representing an unknown {\co}term.  These two identifiers
can be bound by $\tmu$- and $\mu$-abstractions, respectively.  The \emph{input
  abstraction} $\inp x c$ binds its input to the variable $x$ before running the
command $c$, and is analogous to the context $\Let x = \hole \In c$ where
$\hole$ denotes the hole in which the input is placed.  Dually, the \emph{output
  abstraction} $\outp \alpha c$ binds its partner {\co}term to $\alpha$ before
running $c$.  This can be viewed operationally as ``capturing'' the current
continuation as $\alpha$ analogous to a control operator or the $\mu$s of
Parigot's $\lambda\mu$-calculus \cite{Parigot1992LMC}.

We also have terms and {\co}terms for interacting with specific types of
information.  The \emph{call stack} $\app v e$ is a {\co}term describing a
function call with the argument $v$ such that the returned result will be
consumed by $e$.  Call stacks match with a \emph{function abstraction}
$\fnc x \alpha c$ which binds the argument and return continuation to $x$ and
$\alpha$, respectively, before running the command $c$.  Values of a sum type
can be introduced via one of the two \emph{injections} $\inl v$ or $\inr v$
denoting which side of the sum was chosen.  Injections match with a \emph{case
  abstraction} $\casesum{x}{c_1}{y}{c_2}$ which checks the tag of the injection
to determine which of $c_1$ or $c_2$ to run.

The operational behavior of this small sequent calculus can be given in terms of
substitution (similar to the $\lambda$-calculus) as shown in
\cref{fig:basic-lmtm}.  For now consider just the $\betamu$ and $\betatmu$
reduction rules.  These two rules are perfectly symmetric, but already they pose
a serious problem: they are completely non-deterministic!  To see why, consider
the following critical pair where both $\betamu$ and $\betatmu$ can fire:
\begin{align*}
  c_0\subst{\alpha}{\inp x c_1}
  \unred[\betamu]
  \cut{\outp\alpha c_0}{\inp x c_1}
  &\red[\betatmu]
  c_1\subst{x}{\outp \alpha c_0}
\end{align*}
There is no reason the two possible reducts must be related.  In fact, in the
special case where $\alpha$ and $x$ are not used, like in the command
$\cut{\outp\delta\cut{x}{\alpha}}{\inp x \cut{y}{\beta}}$, then we can step to
two arbitrarily different results:
\begin{align*}
  \cut{x}{\alpha}
  \unred[\betamu]
  \cut{\outp\delta \cut{x}{\alpha}}{\inp z \cut{y}{\beta}}
  &\red[\betatmu]
  \cut{y}{\beta}
\end{align*}
So this calculus is not just non-deterministic, but also non-confluent (in other
words, the starting command above reduces to both $\cut{x}{\alpha}$ and
$\cut{y}{\beta}$, but there is no common $c'$ such that
$\cut{x}{\alpha} \reds c' \unreds \cut{y}{\beta}$).

\subsection{Global confluence solution: Evaluation strategy}

The loss of confluence is a large departure from similar models of computation
like the $\lambda$-calculus, and may not be desirable.  In Curien and Herbelin's
$\lmtm$-calculus \cite{CurienHerbelin2000DC}, determinism was restored by
imposing a discipline onto reduction which prioritized one side over the other:
\begin{quote}
  Call-by-value consists in giving priority to the $\mu$ redexes,
  %(which serve to encode the terms, say, of the form $M~N$),
   while call-by-name gives priority
  to the $\tmu$ redexes.
\end{quote}
The difference between call-by-value and -name reduction can be seen in
$\lambda$-calculus terms like $f~(g~1)$: in call-by-value the call $g~1$ should
be evaluated first before passing the result to $f$, but in call-by-name the
call to $f$ should be evaluated first by passing $g~1$ unevaluated as its
argument.  This is analogous to the sequent calculus command
$\cut{\outp\beta\cut{g}{\app 1 \beta}}{\inp x \cut{f}{\app x \alpha}}$: in
call-by-value the $\mu$ should be given priority which causes $g$ to be called
first, whereas in call-by-name the $\tmu$ should be given priority which causes
$f$ to be called first.  Choosing a priority between $\mu$ and $\tmu$
effectively restricts the notion of binding and substitution for variables and
{\co}variables.  Prioritizing $\betamu$ reductions means that $\mu$-abstractions
are not substitutable, whereas prioritizing $\betatmu$ means that
$\tmu$-abstractions are not substitutable.

\begin{figure}
\centering

Sub-syntax and call-by-value reduction rules of $\lmtmQ$:
\begin{align*}
  c &::= \cut{v}{e}
  \\
  v &::= V \Alt \outp \alpha c
  &
  V &::= x
    \Alt \fnc x \alpha c
    \Alt \inl{V}
    \Alt \inr{V}
  \\
  e &::= \alpha \Alt \inp x c
    \Alt \app{V}{e}
    \Alt \casesum{x}{c_1}{y}{c_2}
\end{align*}
\begin{align*}
  \rewriterule{\betamu[\CBV]}
  {
    \cut{\outp\alpha c}{e} &\red c\subst{\alpha}{e}
  }
  &
  \rewriterule{\betar[\to][\CBV]}
  {
    \cut{\fnc x \alpha c}{\app{V}{e}}
    &\red
    c\subs{\asub{x}{V},\asub{\alpha}{e}}
  }
  \\
  \rewriterule{\betatmu[\CBV]}
  {
    \cut{V}{\inp x c} &\red c\subst{x}{V}
  }
  &
  \rewriterule{\betar[\oplus][\CBV]}
  {
    \cut{\inj{i}{V}}{\casesum{x_1}{c_1}{x_2}{c_2}}
    &\red
    c_i\subs{\asub{x_i}{V}}
  }
\end{align*}

Sub-syntax and call-by-name reduction rules of $\lmtmT$:
\begin{align*}
  c &::= \cut{v}{e}
  \\
  v &::= x
    \Alt \outp \alpha c
    \Alt \fnc x \alpha c
    \Alt \inl{v} \Alt \inr{v}
  \\
  e &::= E \Alt \inp x c
  &
  E &::= \alpha
    \Alt \app{v}{E}
    \Alt \casesum{x}{c_1}{y}{c_2}
\end{align*}
\begin{align*}
  \rewriterule{\betamu[\CBN]}
  {
    \cut{\outp\alpha c}{E} &\red c\subst{\alpha}{E}
  }
  &
  \rewriterule{\betar[\to][\CBN]}
  {
    \cut{\fnc x \alpha c}{\app{v}{E}}
    &\red
    c\subs{\asub{x}{v},\asub{\alpha}{E}}
  }
  \\
  \rewriterule{\betatmu[\CBN]}
  {
    \cut{v}{\inp x c} &\red c\subst{x}{v}
  }
  &
  \rewriterule{\betar[\oplus][\CBN]}
  {
    \cut{\inj{i}{v}}{\casesum{x_1}{c_1}{x_2}{c_2}}
    &\red
    c_i\subs{\asub{x_i}{v}}
  }
\end{align*}

\caption{The call-by-value $\lmtmQ$ and call-by-name $\lmtmT$ sub-calculi.}
\label{fig:lmtmQ}
\label{fig:lmtmT}
\end{figure}

In call-by-value, variables stand for a subset of terms called \emph{values}
(denoted by $V$), corresponding to the requirement that arguments are evaluated
before being bound to the parameter of a function.  The call-by-value evaluation
strategy is captured by the restriction of the $\lmtm$-calculus called $\lmtmQ$
\cite{CurienHerbelin2000DC}, which introduces a syntax for values as shown in
\cref{fig:lmtmQ}.  Notice how, in each reduction rule, variables are only ever
replaced with values instead of general terms.  In a command like
$\cut{\outp\alpha c_0}{\inp x c_1}$, only the $\betamu[\CBV]$ rule applies
because $\outp\alpha c_0$ is not a value.  To support this restriction on
reduction, $\lmtmQ$ uses a sub-syntax of $\lmtm$: only values may be pushed
onto a call stack ($\app V e$) and only values may be injected into a sum
($\inj{i} V$).  This restriction ensures that the $\betar[\to][\CBV]$ and
$\betar[\oplus][\CBV]$ rules do not get stuck; for example,
$\cut{\fnc x \alpha c}{(\app{\outp\beta c')}{e}}$ is outside the $\lmtmQ$
sub-syntax and is not a $\betar[\to][\CBV]$ redex.

Call-by-name evaluation in the sequent calculus follows a dual restriction to
call-by-value: {\co}variables stand for a subset of {\co}terms referred to as
\emph{{\co}values} (denoted by $E$).  This corresponds to the intuition that
terms are only ever evaluated (via $\betamu$ reduction) when their results are
needed.  The call-by-name evaluation strategy is captured by the restriction of
the $\lmtm$-calculus called $\lmtmT$ \cite{CurienHerbelin2000DC} with an
explicit syntax of {\co}values as shown in \cref{fig:lmtmT}.  Here, the
reduction rules only ever replace {\co}variables with {\co}values, which
analogously forces only the $\betatmu[\CBN]$ rule to apply in the command
$\cut{\outp\alpha c_0}{\inp x c_1}$.  Likewise, $\lmtmT$ uses a different
sub-syntax---the tail of all call stacks must be a {\co}value ($\app v E$)---to
prevent the $\betar[\to][\CBN]$ rule from getting stuck; for example,
$\cut{\fnc x \alpha c}{\app v {(\inp x c')}}$ is outside the $\lmtmT$
sub-syntax and is not a $\betar[\to][\CBN]$ redex.

\subsection{The impact of evaluation strategy on extensionality}

Similar to the $\lambda$-calculus, the above $\beta$ rules explain how to
evaluate commands in the sequent calculus (according to either call-by-value or
-name) to get the result of a program.  But to reason about programs, the
$\lambda$-calculus goes beyond just $\beta$.  The $\eta$ law $\fn x f ~ x = f$
captures a notion of extensionality for functions: two functions are the same if
they give the same output for every input.  This $\eta$ law can be symmetrically
rephrased as follows for both function and sum types in the sequent calculus:
\begin{align*}
  \rewriterule{\etar[\to]}
  {
    \fnc x \alpha \cut{v}{\app x \alpha}
    &\red v
    &
    (x &\notin \FV(v))
  }
  \\
  \rewriterule{\etar[\oplus]}
  {
    \casesum{x}{\cut{\inl x}{e}}{y}{\cut{\inr y}{e}}
    &\red
    e
    &
    (\alpha &\notin \FV(e))
  }
\end{align*}
These two rules state that matching on an input or output and then rebuilding it
exactly as it was is the same thing as doing nothing.

However, it is not always safe to just apply these two $\eta$ laws as stated
above.  In fact, when used incorrectly, they can re-introduce the
non-determinism that we were seeking to avoid!  For example, we have an
analogous critical pair in the call-by-value $\lmtmQ$-calculus using the
$\etar[\to]$ rule to derive two unrelated results (where $x$, $\alpha$, and
$\delta$ are all unused):
\begin{align*}
  \cut{\fnc x \alpha \cut{\outp\delta c_0}{\app x \alpha}}{\inp x c_1}
  &\red[{\betatmu[\CBV]}]
  c_1
  \\
  \cut{\fnc x \alpha \cut{\outp\delta c_0}{\app x \alpha}}{\inp x c_1}
  &\red[{\etar[\to]}]
  \cut{\outp\delta c_0}{\inp x c_1}
  \red[{\betamu[\CBV]}]
  c_0
\end{align*}
Therefore, the full $\etar[\to]$ extensionality rule for functions is not sound
in call-by-value.  This is not just a problem with functions or call-by-value
reduction; a similar critical pair can be constructed in the call-by-name
$\lmtmT$ using the $\etar[\oplus]$ rule (where $\delta$ and $z$ are unused):
\begin{align*}
  \cut
  {\outp\delta c_0}
  {\casesum{x}{\cut{\inl x}{\inp z c_1}}{y}{\cut{\inr y}{\inp z c_1}}}
  &\red[{\betamu[\CBN]}]
  c_0
  \\
  \cut
  {\outp\delta c_0}
  {\casesum{x}{\cut{\inl x}{\inp z c_1}}{y}{\cut{\inr y}{\inp z c_1}}}
  &\red[{\etar[\oplus]}]
  \cut
  {\outp\delta c_0}
  {\inp z c_1}
  \red[{\betatmu[\CBN]}]
  c_1
\end{align*}
So the soundness of extensional reasoning depends on the evaluation strategy of
the language.%
\footnote{Analogous counter-examples appear in the $\lambda$-calculus using a
  non-terminating term (denoted by $\Omega$).  The instance
  $\fn x \Omega ~ x = \Omega$ of the $\eta$ law does not hold in the
  call-by-value $\lambda$-calculus.  Similarly, the extensionality of sum types
  in the call-by-value $\lambda$-calculus is captured by
  \begin{math}
    \Case v \Of \{\inl x \gives f ~ (\inl x) \mid \inr y \gives f ~ (\inr y)\}
    =
    f ~ v
    ,
  \end{math}
  but this equality does not hold in the call-by-name $\lambda$-calculus
  when $v$ is $\Omega$ and $f$ discards its input.}

\section{Polarization: \texorpdfstring{$\lmtmP$}{System L}}
\label{sec:intro-polarity}

Evaluation strategy can have an impact on the reasoning principles of certain
types: the full $\eta$ law for functions is only sound with call-by-name
evaluation, whereas the full $\eta$ law for sums is only sound with
call-by-value evaluation.  So does that mean one or the other is doomed to be
hamstrung by our inevitable choice?  No!  Instead of picking just one evaluation
strategy for an entire program, we can individually pick and choose which one to
apply based on the type of information we are dealing with.  That way, we can
always arrange that the full $\eta$ laws apply for every type and maximize our
ability for extensional reasoning.

\subsection{Local confluence solution: polarity of a command}

\begin{figure}
\centering

Sub-syntax with hereditary weak-head normal terms ($W$) and forcing {\co}terms
($F$):
\begin{align*}
  c &::= \cut{v}[\strat{S}]{e}
  &
  \strat{S} &::= \CBV \mid \CBN
  \\
  v &::= W \Alt \outp \alpha c
  &
  W &::= x
    \Alt \fnc x \alpha c
    \Alt \inl{W}
    \Alt \inr{W}
    % \Alt \cocaseforce{\alpha}{c}
    % \Alt \delay{W}
    % \Alt \cocaseunwrap{\alpha}{c}
    % \Alt \wrap{v}
  \\
  e &::= F \Alt \inp x c
  &
  F &::= \alpha
    \Alt \app{W}{F}
    \Alt \casesum{x}{c_1}{y}{c_2}
    % \Alt \force{F}  
    % \Alt \casedelay{x}{c}
    % \Alt \unwrap{e}
    % \Alt \casewrap{x}{c}
\end{align*}

Definitions of values and {\co}values in call-by-value ($V_\CBV, E_\CBV$) and
call-by-name ($V_\CBN, E_\CBN$):
\begin{align*}
  V_\CBV &::= W & E_\CBV &::= e
  &
  V_\CBN &::= v & E_\CBN &::= F
\end{align*}

Hybrid call-by-value and call-by-name reduction rules:
\begin{align*}
  \rewriterule{\betamu[\strat{S}]}
  {
    \cut{\outp\alpha c}[\strat{S}]{E_{\strat{S}}}
    &\red
    c\subst{\alpha}{E_{\strat{S}}}
  }
  &
  \rewriterule{\betar[\to]}
  {
    \cut{\fnc x \alpha c}[\CBN]{\app{W}{F}}
    &\red
    c\subs{\asub{x}{W},\asub{\alpha}{F}}
  }
  \\
  \rewriterule{\betatmu[\strat{S}]}
  {
    \cut{V_{\strat{S}}}[\strat{S}]{\inp x c} &\red c\subst{x}{V_{\strat{S}}}
  }
  &
  \rewriterule{\betar[\oplus]}
  {
    \cut{\inj{i}{W}}[\CBV]{\casesum{x_1}{c_1}{x_2}{c_2}}
    &\red
    c_i\subs{\asub{x_i}{W}}
  }
  % \\
  % \rewriterule{\betar[\ToPos]}
  % {
  %   \cut{\wrap{v}}[\CBV]{\casewrap{x}{c}}
  %   &\red
  %   c\subst{x}{v}
  % }
  % &
  % \rewriterule{\betar[\ToNeg]}
  % {
  %   \cut{\cocaseunwrap{\alpha}{c}}[\CBN]{\force{e}}
  %   &\red
  %   c\subst{\alpha}{e}
  % }
\end{align*}

\caption{The polarized $\lmtm$ sub-calculus: $\lmtmP$.}
\label{fig:lmtmP}
\end{figure}

Since we will have multiple options for evaluation strategy within the same
program, we have to decide which one to use at each reduction.  This
decision---call-by-value versus call-by-name---might be different at each step
but should be statically determined before evaluation, so it must be made
apparent in the program itself.  And since commands are the only syntactic
expressions which can reduce, we only need to know what happens in $\cut{v}{e}$:
should the interaction between $v$ and $e$ be evaluated according to the
call-by-value or call-by-name priority?  We can make this decision manifest in
the syntax of programs by annotating each command with the sign $\strat{S}$
denoting its polarity (positive or negative) as $\cut{v}[\strat{S}]{e}$.  This
way, $\strat{S}$ will tell us whether to use the call-by-value priority (for
positive $\strat{S} = \CBV$) or the call-by-name priority (for negative
$\strat{S} = \CBN$).  We can formalize this idea of a polarized sequent
calculus, called $\lmtmP$, with the grammar in \cref{fig:lmtmP}.  Note that, as
far as function and sum types are concerned, this sub-syntax lies in the
intersection of both $\lmtmT$ and $\lmtmQ$.  In particular, we have the notion
of a \emph{weak-head normal term} $W$ similar to the call-by-value
$\lmtmQ$-calculus' values, and a \emph{forcing {\co}term} $F$ similar to the
call-by-name $\lmtmT$ calculus' {\co}values.  Furthermore, all of the
restrictions imposed by $\lmtmT$ and $\lmtmQ$ are satisfied; injections
$\Inj{i}$ can only be applied to weak-head normal terms $W$, and call stacks can
only contain a $W$ and an $F$.  These restrictions can be summarized by the
following rule of thumb: constructors and call stacks only apply to terms that
immediately return values and {\co}terms that immediately force their output.

Now, consider the semantics for reducing commands of the above polarized
$\lmtmP$.  Recall that previously we characterized the two different priorities
between $\mu$ and $\tmu$ via the notion of substitutability: which values can
replace variables and which {\co}values can replace {\co}variables.  The two
different priorities for call-by-value ($\CBV$) and call-by-name ($\CBN$) can
now coexist in the same language by the definition of positive and negative
values and {\co}values in \cref{fig:lmtmP}.  With this definition of polarized
(\co)values, we can now give generic $\betamu$ and $\betatmu$ reduction rules
that perform the correct operation for any given sign $\strat{S}$ of a command.
The other rules for reducing function calls and case analysis are similar as
before.

The polarization of a command offers a local solution to the confluence problem
discussed in the previous section.  For example, the outermost command in
$\cut{\outp\delta \cut{x}{\alpha}}{\inp z \cut{y}{\beta}}$ can either be
positive, as in $\cut{\outp\delta \cut{x}{\alpha}}[\CBV]{\inp z \cut{y}{\beta}}$
which allows only the call-by-value reduction to $\cut{x}{\alpha}$, or negative,
as in $\cut{\outp\delta \cut{x}{\alpha}}[\CBN]{\inp z \cut{y}{\beta}}$ which
allows only the call-by-name reduction to $\cut{y}{\beta}$.  But it cannot be
both.  This forces each command to decide which of the two priorities to use,
resolving the critical pair in one direction or the other like so:
\begin{align*}
  \cut{x}{\alpha}
  \unred[{\betamu[\CBV]}]
  \cut{\outp\delta \cut{x}{\alpha}}[\CBV]{\inp z \cut{y}{\beta}}
  &\not\red[{\betatmu[\CBV]}]
  \\
  \not\unred[{\betamu[\CBN]}]
  \cut{\outp\delta \cut{x}{\alpha}}[\CBN]{\inp z \cut{y}{\beta}}
  &\red[{\betatmu[\CBN]}]
  \cut{y}{\beta}
\end{align*}

However, we must take care that these polarity decisions are made consistently
throughout the program.  For example, we could have the following critical pair,
which relies on a misalignment between the original occurrence of
$\outp\delta c_0$ in a $\CBN$ command and its second reference through the bound
variable $x$ in a $\CBV$ command:
\begin{align*}
  c_1
  \unred[{\betatmu[\CBN]}]
  \cut{\outp\delta c_0}[\CBN]{\inp x c_1\subst{y}{x}}
  \unred[{\betamu[\CBV]}]
  \cut{\outp\delta c_0}[\CBN]{\inp x \cut{x}[\CBV]{\inp y c_1}}
  &\red[{\betatmu[\CBN]}]
  \cut{\outp\delta c_0}[\CBV]{\inp y c_1}
  \red[{\betamu[\CBV]}]
  c_0
\end{align*}
where again, $\delta$, $x$, and $y$ do not occur free in $c_0$ and $c_1$.  So to
restore confluence, we need to ensure that misaligned signs are ruled out. But
before presenting the full solution to this issue, we examine how extensionality
interacts with polarity.

\subsection{The impact of polarity on extensionality}

The goal of combining both call-by-value and call-by-name reduction, as above,
is so that we could safely validate the full $\eta$ axioms for both the function
and sum types.  We can now explicitly enforce the requirement that each axiom
only applies to terms or {\co}terms following a specific evaluation strategy
with $\CBV$ and $\CBN$ sign annotations like so:
\begin{align*}
  \rewriterule{\etar[\to][\CBN]}
  {
    \fnc x \alpha \cut{v}[\CBN]{\app x \alpha} &\red v
  }
  \\
  \rewriterule{\etar[\oplus][\CBV]}
  {
    \casesum{x}{\cut{\inl x}[\CBV]{e}}{y}{\cut{\inr y}[\CBV]{e}}
    &\red
    e
  }
\end{align*}
Notice how, so long as the signs in the command align---meaning that each term
and {\co}term consistently appears in only positive or only negative
commands---we prevent the possibility of non-confluence.  For example, the
following command is confluent due to this alignment, since only the
call-by-name $\betatmu[\CBN]$ rule applies both before and after
$\eta$-reduction:
\begin{align*}
  c_1
  \unred[{\betatmu[\CBN]}]
  \cut
  {\fnc x \alpha \cut{\outp\delta c_0}[\CBN]{\app x \alpha}}
  [\CBN]
  {\inp x c_1}
  \red[{\etar[\to][\CBN]}]
  \cut{\outp\delta c_0}[\CBN]{\inp x c_1}
  \red[{\betatmu[\CBN]}]
  c_1
\end{align*}
where again we assume that $x$, $\alpha$, and $\delta$ are not used.  In
contrast, the counter-example to confluence with $\etar[\to][\CBN]$ has
mismatching sign annotations like so:
\begin{align*}
  c_1
  \unred[{\betatmu[\CBV]}]
  \cut
  {\fnc x \alpha \cut{\outp\delta c_0}[\CBN]{\app x \alpha}}
  [\CBV]
  {\inp x c_1}
  \red[{\etar[\to][\CBN]}]
  \cut{\outp\delta c_0}[\CBV]{\inp x c_1}
  \red[{\betamu[\CBV]}]
  c_0
\end{align*}
The misalignment between signs is due to the fact that a function abstraction is
a negative term, but it is used in a positive command.

\begin{figure}
\centering

Determining the polarity of a type:
\begin{align*}
  A, B, C &::= X \Alt \overline{X} \Alt A \to B \Alt A \oplus B
\end{align*}
\begin{gather*}
  \axiom{X : \CBV}
  \qquad
  \axiom{\overline{X} : \CBN}
  \qquad
  \infer
  {A \oplus B : \CBV}
  {A : \CBV & B : \CBV}
  \qquad
  \infer
  {A \to B : \CBN}
  {A : \CBV & B : \CBN}
\end{gather*}

Inference rules that are generic for any type:
\begin{gather*}
  \infer[\CutRule]
  {\cut{v}[\strat{S}]{e} : (\Gamma \entails \Delta)}
  {
    \Gamma \entails v : A \act \Delta
    &
    A : \strat{S}
    &
    \Gamma \act e : A \entails \Delta
  }
  \\[1ex]
  \axiom[\VR]{\Gamma, x:A \entails x : A \stoup \Delta}
  \qqqquad
  \axiom[\VL]{\Gamma \stoup \alpha : A \entails \alpha:A, \Delta}
  \\[1ex]
  \infer[\AR]
  {\Gamma \entails \outp \alpha c : A \act \Delta}
  {c : (\Gamma \entails \alpha:A, \Delta)}
  \qqqquad
  \infer[\AL]
  {\Gamma \act \inp x c : A \entails \Delta}
  {c : (\Gamma, x:A \entails \Delta)}
  \\[1ex]
  \infer[\FR]
  {\Gamma \entails W : A \act \Delta}
  {\Gamma \entails W : A \stoup \Delta}
  \qqqquad
  \infer[\FL]
  {\Gamma \act F : A \entails \Delta}
  {\Gamma \stoup F : A \entails \Delta}
\end{gather*}

Inference rules for specific types:
\begin{gather*}
  \infer[\toR]
  {\Gamma \entails \fnc x \alpha c : A \to B \stoup \Delta}
  {c : (\Gamma, x:A \entails \alpha:B, \Delta)}
  \qqqquad
  \infer[\toL]
  {\Gamma \stoup \app W F : A \to B \entails \Delta}
  {
    \Gamma \entails W : A \stoup \Delta
    &
    \Gamma \stoup F : B \entails \Delta
  }
  \\[1ex]
  \infer[{\oplusR[1]}]
  {\Gamma \entails \inl W : A \oplus B \stoup \Delta}
  {\Gamma \entails W : A \stoup \Delta}
  \qqqquad
  \infer[{\oplusR[2]}]
  {\Gamma \entails \inr W : A \oplus B \stoup \Delta}
  {\Gamma \entails W : B \stoup \Delta}
  \\[1ex]
  \infer[\oplusL]
  {\Gamma \stoup \casesum{x}{c_1}{y}{c_2} : A \oplus B \entails \Delta}
  {
    c_1 : (\Gamma, x:A \entails \Delta)
    &
    c_2 : (\Gamma, y:B \entails \Delta)
  }
\end{gather*}

\caption{Type system for $\lmtmP$.}
\label{fig:lmtmP-types}
\end{figure}

\subsection{Polarity in the type system}
One way to ensure that the signs in a program are consistently aligned is to
link polarity with types.  In other words, we can distinguish the two different
kinds of types according to their polarity: positive types (of kind $\CBV$)
following call-by-value evaluation and negative types (of kind $\CBN$) following
call-by-name evaluation.  For example, consider the basic grammar of types which
includes functions ($A \to B$) and sum ($A \oplus B$) types given in
\cref{fig:lmtmP-types}.  For the base cases, we also have atomic type variables
of positive ($X$) and negative ($\overline{X}$) kinds.  The distinction between
positive and negative types is given by the inference rules for checking
$A:\CBV$ versus $A:\CBN$.  Since the full $\eta$ law for sum types requires
call-by-value evaluation, $A \oplus B$ is positive, and dually since the full
$\eta$ law for function types requires call-by-name evaluation, $A \to B$ is
negative.  Also note that the polarity of a type is hereditary.  The types $A$
and $B$ inside $A \oplus B$ must also be positive.  Interestingly, there is some
added subtlety to functions, the return type $B$ of $A \to B$ must also be
negative, but the input type $A$ is instead positive.  The reason that the
polarity of function inputs are reversed is due to the reversal of information
flow caused by a function: outputs have the same polarity as the connective,
whereas inputs have the opposite polarity.

The type system for $\lmtmP$ which ensures confluence is given in
\cref{fig:lmtmP-types}.  Since there are several different syntactic categories
in the language, we have the following different typing judgments:
\begin{itemize}
\item $c : (\Gamma \entails \Delta)$ says that the command $c$ is well-typed,
  assuming the typing assignments for variables in the environment $\Gamma$ and
  {\co}variables in the environment $\Delta$.%
  \footnote{In each of these judgments, the environment
    $\Gamma = x_1:A_1,\dots,x_n:A_n$ is a set of variable-type pairs and
    $\Delta = \alpha_1:A_1,\dots,\alpha_n:A_n$ is a set of {\co}variable-type
    pairs, such that all (\co)variables are distinct in both.}
\item $\Gamma \entails v : A \act \Delta$ says that the term $v$ eventually
  produces an output of type $A$.
\item $\Gamma \entails W : A \stoup \Delta$ says that $W$ is already a value of
  type $A$.
\item $\Gamma \act e : A \entails \Delta$ says that the {\co}term $e$ eventually
  consumes an input of type $A$.
\item $\Gamma \stoup F : A \entails \Delta$ says that $F$ is already a
  {\co}value of type $A$.
\end{itemize}
The $\CutRule$ rule forms a command between a term and {\co}term of the same
type $A$, while also ensuring that the sign $\strat{S}$ in the command matches
the kind of $A$.  The $\VR$ and $\VL$ rules refer to a (\co)variable in the
environment, and the $\AR$ and $\AL$ rules bind a (\co)variable with a $\mu$- or
$\tmu$-abstraction, respectively.  And the $\FR$ and $\FL$ rules correspond to
the inclusion of $W$ in terms and $F$ in {\co}terms.  As a result, all
well-typed commands are confluent, and all well-typed expressions (commands,
terms, and {\co}terms) are confluent.  For example, the troublesome command
above with misaligned signs is not typable in this system, as shown by the
following derivation that \emph{fails} on one of the highlighted premises
$A:\CBN$ or $A:\CBV$, because there is no type $A$ that is both positive and
negative:
\begin{small}
\begin{gather*}
  \infer[\CutRule]
  {
    \typecmd{\Gamma}{\Delta}
    {\cut{\outp\delta c_0}[\CBN]{\inp x \cut{x}[\CBV]{\inp y c_1}}}
  }
  {
    {
      \infer[\AR]
      {\typetm{\Gamma}{\Delta}{\outp\delta c_0}{A}}
      {\typecmd{\Gamma}{\delta:A,\Delta}{c_0}}
    }
    &
    \hi{A : \CBN}
    &
    {
      \infer[\AL]
      {\typecotm{\Gamma}{\Delta}{\inp x \cut{x}[\CBV]{\inp y c_1}}{A}}
      {
        \infer[\CutRule]
        {\typecmd{\Gamma,x:A}{\Delta}{\cut{x}[\CBV]{\inp y c_1}}}
        {
          \infer[\FR]
          {\typetm{\Gamma,x:A}{\Delta}{x}{A}}
          {
            \axiom[\VR]{\typetmfoc{\Gamma,x:A}{\Delta}{x}{A}}
          }
          &
          \hi{A : \CBV}
          &
          \infer[\AL]
          {\typecotm{\Gamma,x:A}{\Delta}{\inp y c_1}{A}}
          {\typecmd{\Gamma,x:A,y:A}{\Delta}{c_1}}
        }
      }
    }
  }
\end{gather*}
\end{small}%

\begin{intermezzo}
\label{rm:focus}

We used different judgments for typing general terms versus weak-head normal
terms ($\Gamma \entails v : A \act \Delta$ versus
$\Gamma \entails W : A \stoup \Delta$) and general {\co}terms versus forcing
{\co}terms.  One motivation for doing so is to capture the notion of sub-syntax
in the polarized $\lmtmP$ inside the type system itself: the fact that a call
stack must consist of a $W$ and an $F$ can be read off from the judgments in the
$\toL$ inference rule, even if we erased the terms and {\co}terms.  It turns out
that the impact of these syntactic restrictions on typing corresponds to the
notion of \emph{focusing} in logic \cite{Andreoli1992LPFPLL}, which was
originally developed to aid proof search.  The \emph{focused} judgments
$\Gamma \entails W : A \stoup \Delta$ and $\Gamma \stoup F : A \entails \Delta$
are more restrictive than their counterpart $\Gamma \entails v : A \act \Delta$
and $\Gamma \act e : A \entails \Delta$.  For example, there is no way to derive
$\Gamma \entails \outp\alpha c : A \stoup \Delta$, even if
$\Gamma \entails \outp\alpha c : A \act \Delta$ is derivable.  This limits the
number of derivations that can be formed, in the same way that the various
sub-syntaxes limit the number of programs that can be written.
\end{intermezzo}

\subsection{Restoring expressiveness with polarity shifts}

Unfortunately, it turns out that just distinguishing the two kinds of types, as
in the type system above, severely weakens the system.  The identity function
$A \to A$ is no longer well-formed!  If $A$ is positive, then it cannot be the
return type of a function, and if $A$ is negative, it cannot be the input type.
This limitation is clearly untenable, and so we need \emph{polarity shifts}
which explicitly coerce between the two kinds of types: the down-shift $\ToPos$
converts a negative type to a positive type, and the up-shift $\ToNeg$ converts
a positive type in a negative type.  For example, we could write the identity
function as either $\ToPos A \to A$, when $A$ is negative, or $A \to \ToNeg A$,
when $A$ is positive.  With these polarity shifts, the expressiveness of the
system is restored.

\begin{figure}
\centering

Extending the syntax of types and programs with shifts:
\begin{align*}
  A, B, C &::= \dots \Alt \ToPos A \Alt \ToNeg B
  &
  W &::= \dots \Alt \wrap v \Alt \cocaseunwrap\alpha c
  &
  F &::= \dots \Alt \casewrap x c \Alt \unwrap e
\end{align*}

Reduction rules for shifts:
\begin{align*}
  \rewriterule{\betar[\ToPos]}
  {
    \cut{\wrap{v}}[\CBV]{\casewrap{x}{c}}
    &\red
    c\subst{x}{v}
  }
  &
  \rewriterule{\betar[\ToNeg]}
  {
    \cut{\cocaseunwrap{\alpha}{c}}[\CBN]{\unwrap{e}}
    &\red
    c\subst{\alpha}{e}
  }
\end{align*}

Typing rules for shifts:
\begin{gather*}
  \infer
  {\ToPos A : \CBV}
  {A : \CBN}
  \qqqquad
  \infer[\ToPosR]
  {\Gamma \entails \wrap v : \ToPos A \stoup \Delta}
  {\Gamma \entails v : A \act \Delta}
  \qqqquad
  \infer[\ToPosL]
  {\Gamma \stoup \casewrap x c : \ToPos A \entails \Delta}
  {c : (\Gamma, x:A \entails \Delta)}
  \\[1ex]
  \infer
  {\ToNeg A : \CBN}
  {A : \CBV}
  \qqqquad
  \infer[\ToNegR]
  {\Gamma \entails \cocaseunwrap \alpha c : \ToNeg A \stoup \Delta}
  {c : (\Gamma \entails \alpha:A, \Delta)}
  \qqqquad
  \infer[\ToNegL]
  {\Gamma \stoup \unwrap e : \ToNeg A \entails \Delta}
  {\Gamma \act e : A \entails \Delta}
\end{gather*}

\caption{Polarity shifts in $\lmtmP$.}
\label{fig:lmtmP-shifts}
\end{figure}

$\lmtmP$ can be extended with the shifts given in \cref{fig:lmtmP-shifts}.  The
terms and {\co}terms for the polarity shifts are similar to those for function
and sum types.  The up-shift $\ToNeg A$ is analogous to a nullary function
$1 \to A$ (where $1$ is a unit type) with the exception that $A$ is positive,
rather than negative.  As such, it contains a $\lambda$-abstraction of the form
$\cocaseunwrap{\alpha}{c}$ (with no argument parameter $x$) and a call stack of
the form $\unwrap{e}$ (with no restriction on $e$ because all {\co}terms force
their input in call-by-value).  The down-shift $\ToPos A$ is analogous to a
degenerate sum type $0 \oplus A$ (where $0$ is an empty type) with the exception
that $A$ is negative, rather than positive.  It contains a single injection
$\wrap{v}$ (with no restriction on $v$ because all terms are already values in
call-by-name) and a pattern-matching abstraction of the form $\casewrap{x}{c}$
for unwrapping the injection.

\section{Compiling Call-by-Name and Call-by-Value to
  \texorpdfstring{$\lmtmP$}{System L}}
\label{sec:CBVCBNintoL}
\begin{figure}
\centering

Call-by-value polarizing compilation of $\lmtmQ$ into $\lmtmP$.
\begin{gather*}
\begin{aligned}
  X^\CBV &\defeq X
  &\qquad
  (A \to B)^\CBV &\defeq \ToPos (A^\CBV \to (\ToNeg B^\CBV))
  &\qquad
  (A \oplus B)^\CBV &\defeq A^\CBV \oplus B^\CBV
\end{aligned}
\\[1ex]
\cut{v}{e}^\CBV \defeq \cut{v^\CBV}[\CBV]{e^\CBV}
\\
\begin{aligned}
  x^\CBV &\defeq x
  &\qquad
  \alpha^\CBV &\defeq \alpha
  \\
  (\outp \alpha c)^\CBV &\defeq \outp \alpha (c^\CBV)
  &\qquad
  (\inp x c)^\CBV &\defeq \inp x (c^\CBV)
  \\
  (\fnc{x}{\alpha}c)^\CBV
  &\defeq
  \wrap{\left(
      \fnc{x}{\beta}\cut{\cocaseunwrap{\alpha}{c^\CBV}}[\CBN]{\beta}
    \right)}
  &\qquad
  (\app{V}{e})^\CBV
  &\defeq
  \casewrap{x}{\cut{x}[\CBN]{\app{V^\CBV}{[\unwrap{e^\CBV}]}}}
  \\
  (\inj{i}{V})^\CBV
  &\defeq
  \inj{i}{(V^\CBV)}
  &
  \caseof{\recv{\inj{i}{x_i}}{c_i}}^\CBV
  &\defeq
  \caseof{\recv{\inj{i}{x_i}}{c_i^\CBV}}
\end{aligned}
\end{gather*}

Call-by-name polarizing compilation of $\lmtmT$ into $\lmtmP$.
\begin{gather*}
\begin{aligned}
  X^\CBN &\defeq \overline{X}
  &\qquad
  (A \to B)^\CBN &\defeq (\ToPos A^\CBN) \to B^\CBN
  &\qquad
  (A \oplus B)^\CBN &\defeq \ToNeg ((\ToPos A^\CBN) \oplus (\ToPos B^\CBN))
\end{aligned}
\\[1ex]
\cut{v}{e}^\CBN \defeq \cut{v^\CBN}[\CBN]{e^\CBN}
\\
\begin{aligned}
  x^\CBN &\defeq x
  &\qquad
  \alpha^\CBN &\defeq \alpha
  \\
  (\outp \alpha c)^\CBN &\defeq \outp \alpha (c^\CBN)
  &\qquad
  (\inp x c)^\CBN &\defeq \inp x (c^\CBN)
  \\
  (\fnc{x}{\alpha}c)^\CBN
  &\defeq
  \fnc{y}{\alpha}\cut{y}[\CBV]{\casewrap{x}{c^\CBN}}
  &
  (\app{v}{E})^\CBN
  &\defeq
  \app{(\wrap{v^\CBN})}{E^\CBN}
  \\
  (\inj{i}{v})^\CBN
  &\defeq
  \cocaseunwrap{\beta}{\cut{\inj{i}{(\wrap{v^\CBN})}}[\CBV]{\beta}}
  &
  \caseof{\recv{\inj{i}{x_i}}{c_i}}^\CBN
  &\defeq
  \unwrap
  {\left[
    \caseof{\recv{\inj{i}{y_i}}{\cut{y_i}[\CBV]{\casewrap{x_i}{c_i^\CBN}}}}
  \right]}
\end{aligned}
\end{gather*}

\caption{Compiling $\lmtmQ$ and $\lmtmT$ into $\lmtmP$.}
\label{fig:lmtmQ-encoding}
\label{fig:lmtmT-encoding}
\end{figure}

With the polarity shifts between positive and negative types, we can express
every program that could have been written in the unpolarized languages.  But
now since \emph{both} call-by-value and -name evaluation is denoted by different
types, the types themselves signify the calling convention.  Both compilations
of $\lmtmQ$ and $\lmtmT$ into $\lmtmP$ are shown in \cref{fig:lmtmQ-encoding}.
In the call-by-value compilation of $\lmtmQ$, every type is translated to a
positive type, in other words $A^\CBV : \CBV$.  This leaves the sum connective
$\oplus$ since it is already purely positive, but the function type requires two
shifts: one $\ToNeg$ to convert the positive return type of the function to a
negative type (satisfying the preconditions of the polarized function arrow) and
another $\ToPos$ to ensure the overall function type itself is positive instead
of negative.  In contrast, the call-by-name compilation of $\lmtmT$ translates
every type to a negative type, in other words $A^\CBN : \CBN$.  As a result, the
call-by-name compilation of a function requires only one shift to convert the
input type of the function to be positive, but the sum connective now requires
heavy shifting: two $\ToPos$ shifts to convert both sides of the sum to be
positive (satisfying the preconditions of the polarized sum) and a $\ToNeg$
shift to make the overall sum type itself negative.
% For example, purely call-by-name functions from input type $A : \CBN$ to
% output type $B : \CBN$ can be encoded as:
% \begin{align*}
%   A \overset{\CBN}\to B
%   &=
%   \ToPos A \to B
%   \\
%   \lambda^\CBN [\app x \alpha]. c
  
%   &=
%   \fnc{\wrap x}{\alpha} c
%   =
%   \fnc y \alpha \cut{y}[\CBV]{\casewrap x c}
%   \\
%   v \overset{\CBN}\cdot F
%   &=
%   \app{(\wrap v)}{F}
% \end{align*}
% where the nested (\co)pattern $\app{(\wrap x)}{\alpha}$ is expanded as shown
% above into a chain of two shallow $\lambda$ and $\tlam$ matching abstractions.
% In contrast, purely call-by-value functions from an input type $A : \CBV$ to
% an output type $B : \CBV$ can be encoded as:
% \begin{align*}
%   A \overset{\CBV}\to B
%   &=
%   \ToPos(A \to \ToNeg B)
%   \\
%   \lambda^\CBV [\app x \alpha]. c
%   &=
%   \wrap{(\fnc{x}{\unwrap\alpha} c)}
%   =
%   \wrap{(\fnc{x}{\beta}\cut{\cocaseunwrap\alpha c}[\CBN]{\beta})}
%   \\
%   W \overset{\CBV}\cdot e
%   &=
%   \casewrap x {\cut{x}[\CBV]{\app{W}{[\unwrap e]}}}
% \end{align*}
% where we again have a nested (\co)pattern $\app{x}{\unwrap\alpha}$ which is
% expanded similarly to the above.  Note that in this case, we must convert the
% call stack $W \overset{\CBV}\cdot e$ into a pattern-matching
% $\tlam$-abstraction, in order to extract the function from inside of the
% $\wrap{}$ constructor.

At a basic level, these two compilations make sense from the perspective of
typability (corresponding to provability in logic)---by inspection, all of the
types line up with their newly-assigned polarities.  But programs are meant to
be run, so we care about more than just typability.  At a deeper level, the
compilations are \emph{sound} with respect to equality of terms: if two terms
are equal, then their compiled forms are also equal.  We have not yet formally
defined equality, so we will return to this question later in
\cref{sec:encoding-type-iso}.

\subsection{More direct compilation with four shifts}

However, even though these compilations are correct, some of the cases are much
more direct than others.  In particular, observe how the negative compilation of
functions can be understood as just expanding a nested {\co}pattern match like
so:
\begin{displaymath}
  (\fnc{x}{\alpha}c)^\CBN
  =
  \fn{[\app{\wrap{x}}{\alpha}]}c^\CBN
\end{displaymath}
which matches the macro-expanded call stack
\begin{math}
  (\app{v}{E})^\CBN = \app{(\wrap{v^\CBN})}{E^\CBN}
  .
\end{math}
However, the positive compilation of functions is not so nice; it wraps a
$\lambda$-abstraction in a $\ToPos$ shift, which means that a $\lmtmQ$ call
stack of the form $\app{V}{e}$ is hidden inside a pattern-matching abstraction
in order to extract the contents of the shifted function.  Similarly, the
positive compilation of sum types just distributes over constructors and pattern
matching, since sum types are already purely positive, but the negative
compilation of a sum injection $(\inj{i}{v})^\CBN$ is hidden inside a
$\lambda$-abstraction for the same reason.

For now, this mismatch between construction and pattern matching is tolerable,
but when considering other compilations (\cref{sec:intro-need}) it will become
unacceptable.  The issue in both cases is with the outermost shift: the
constructor or destructor is applied to the wrong side (term versus {\co}term)
which forces us to wrap a $\lambda$-abstraction instead of a call stack or sum
injection.  Fortunately, there is an easy way to remedy the situation: use a
different formulation of the two shifts.

\begin{figure}
\centering

Extending the syntax of types and programs with shifts:
\begin{align*}
  A, B, C &::= \dots \Alt \FromPos A \Alt \FromNeg B
  &
  W &::= \dots \Alt \delay W \Alt \cocaseforce\alpha c
  &
  F &::= \dots \Alt \casedelay x c \Alt \force F
\end{align*}

Reduction rules for shifts:
\begin{align*}
  \rewriterule{\betar[\FromNeg]}
  {
    \cut{\cocaseforce{\alpha}{c}}[\CBV]{\force{F}}
    &\red
    c\subst{\alpha}{F}
  }
  &
  \rewriterule{\betar[\FromPos]}
  {
    \cut{\delay{W}}[\CBN]{\casedelay{x}{c}} &\red c\subst{x}{W}
  }
\end{align*}

Typing rules for shifts:
\begin{gather*}
  \infer
  {\FromNeg A : \CBV}
  {A : \CBN}
  \qqqquad
  \infer[\FromNegR]
  {\Gamma \entails \cocaseforce \alpha c : \FromNeg A \stoup \Delta}
  {c : (\Gamma \entails \alpha:A, \Delta)}
  \qqqquad
  \infer[\FromNegL]
  {\Gamma \stoup \force F : \FromNeg A \entails \Delta}
  {\Gamma \stoup F : A \entails \Delta}
  \\[1ex]
  \infer
  {\FromPos A : \CBN}
  {A : \CBV}
  \qqqquad
  \infer[\FromPosR]
  {\Gamma \entails \delay W : \FromPos A \stoup \Delta}
  {\Gamma \entails W : A \stoup \Delta}
  \qqqquad
  \infer[\FromPosL]
  {\Gamma \stoup \casedelay x c : \FromPos A \entails \Delta}
  {c : (\Gamma, x:A \entails \Delta)}
\end{gather*}

\caption{Alternative polarity shifts in $\lmtmP$.}
\label{fig:lmtmP-other-shifts}
\end{figure}

Notice how---as a side effect of the fact that $\wrap{v}$ can contain any term
$v$ and $\unwrap{e}$ can contain any {\co}term $e$---the typing rules $\ToPosR$
and $\ToNegL$ \emph{lose} focus (see \cref{rm:focus}) since their premise is the
unfocused judgment $\Gamma \entails v : A \act \Delta$ and
$\Gamma \act e : A \entails \Delta$, respectively.  These are the first rules
with this property: all the other typing rules with focused conclusions have
focused premises when possible.  The only exceptions to this generalization are
rules like $\toR$ and $\oplusL$, because commands do not have any notion of
focus like terms and {\co}terms do.  It turns out, there is an alternative way
to formulate the shifts so that focus is \emph{gained} rather than \emph{lost}.
In other words, we could have defined polarity shifts that follow the rule of
thumb where constructors and destructors only apply to $W$ terms and $F$
{\co}terms.

These alternative shifts can be added to the calculus, as shown in
\cref{fig:lmtmP-other-shifts}, where the up-shift $\FromPos$ has the same role
of converting positive to negative, and the down-shift $\FromNeg$ still converts
from negative to positive.  The constructor $\Delay$ only applies to a $W$, and
$\Force$ only applies to a $F$.  Instead, the change in polarity is captured by
the fact that $\cocaseforce\alpha c$ is a positive value distinct from the
non-value $\outp\alpha c$, and $\casedelay x c$ is a negative {\co}value which
forces its input unlike the non-strict $\inp x c$.  Note that it is important
for the conclusions of the $\FromNegR$ and $\FromPosL$ rule be in focus (\ie use
the stoup separator $\stoup$ rather than $\act$).  Otherwise, appearances of the
thunk-like abstraction $\cocaseforce\alpha c$ in an injection
($\inj{i}{(\cocaseforce\alpha c)}$) or call stack
($\app{(\cocaseforce\alpha c)}{F}$) would not type-check.  But the premises of
these rules are commands, like in the $\toR$ and $\oplusL$ rules, which cannot
be in focus.  From this standpoint, the thick shifts ($\FromNeg$ and $\FromPos$)
have different focusing properties from the thin shifts ($\ToPos$ and $\ToNeg$).

The compilation translations in \cref{fig:lmtmQ-encoding} can now be cleaned up
with the following alternative compilation using both styles of shifts, whose
impact on terms and {\co}terms is completely captured as a macro-expansion of
patterns and {\co}patterns like so:
\begin{small}
\begin{align*}
  (A \to B)^\CBV
  &=
  \FromNeg (A^\CBV \to (\ToNeg B^\CBV))
  &
  (\app{V}{e})^\CBV
  &=
  \force{[\app{V^\CBV}{[\unwrap{e^\CBV}]}]}
  &
  (\fn{[\app x \alpha]} c)^\CBV
  &=
  \cocaseforce{[\app{x}{[\unwrap\alpha]}]}{c^\CBV}
  \\
  (A \oplus B)^\CBN
  &=
  \FromPos ((\ToPos A^\CBN) \oplus (\ToPos B^\CBN))
  &
  (\inj{i}{v})^\CBN
  &=
  \delay{(\inj{i}{(\wrap{v^\CBN})})}
  &
  \caseof{\recv{\inj{i}{x_i}}{c_i}}^\CBN
  &=
  \caseof{\recv{\delay{(\inl{(\wrap{x_i})})}}{c_i^\CBN}}
\end{align*}
\end{small}%
Note that the direction of the shifts are the same as before, the only thing
that changes is that the outer arrows are thick ($\FromNeg$ or $\FromPos$),
whereas the inner arrows are thin ($\ToPos$ or $\ToNeg$).  While this is just an
aesthetic change for now, the impact of the two different styles of shifts in
compilation becomes crucial once we go beyond just call-by-value and
call-by-name evaluation.
%\end{intermezzo}

\section{Classical Call-by-Need and its Dual}
\label{sec:intro-need}

Polarity neatly integrates the best of both call-by-value and call-by-name
languages.  However, there is more to computation than just those two evaluation
strategies.  Call-by-need evaluation is a different way to integrate aspects of
both call-by-name and -value: like call-by-name results are only computed when
they are needed, but like call-by-value they are never evaluated more than once.
For example, let $I = \fn x x$ be the identity function, so that $f$ will never
be evaluated in the expression $\Let f = I~I \In 5$ because it is not needed to
return the result $5$.  In contrast, $f$ will be evaluated exactly once in
\begin{math}
  \Let f = I~I \In I~f~f~5
\end{math}
like so:
\begin{align*}
  \Let f = I~I \In \underline{I ~ f} ~ f ~ 5
  &\red
  \Let f = \underline{I~I} \In f ~ f ~ 5
  \\
  &\red
  \underline{\Let f = I \In f ~ f ~ 5}
  \\
  &\red
  \underline{I ~ I} ~ 5
  \red
  \underline{I ~ 5}
  \red
  5  
\end{align*}
where the next reduction to take place at each step is underlined for clarity.
Notice how, in order to simultaneously delay and share the computation of $f$,
reduction occurs \emph{underneath} $\Let$ bindings.  First, the inner redex
$I~f$ is performed, after which $f$ is needed to continue evaluating the
expression $f~f~5$ which forces the evaluation of $f$.  The ability to reduce
inside of $\Let$s is the key to the call-by-need $\lambda$-calculus
\cite{AMOFW1995CBNLC,AriolaFelleisen1997CBNLC}.

\subsection{Sharing work in the sequent calculus}

If delayed $\Let$ binding is the key to the call-by-need $\lambda$-calculus,
what is the key to a call-by-need sequent calculus?  Since $\tmu$-abstractions
are the analogue to $\Let$s, a call-by-need priority between $\betamu$ and
$\betatmu$ reductions relies on delayed $\tmu$-bindings \cite{AHS2011CCBND}.
Rather than forcing one substitution or the other in the command
$\cut{\outp\beta c'}{\inp x c}$ like call-by-value and call-by-name do,
call-by-need reduces the command $c$ underneath the $\tmu$-binding of $x$.  In
general, this reduction will proceed like so:
\begin{align*}
  \cut{\outp\beta c'}{\inp x \underline{c}}
  &\reds
  \cut
  {\outp\beta c'}
  {
    \inp x
    \cut{v_1}{\inp{y_1} \dots \cut{v_n}{\inp{y_n} \cut{x}{E}}}
  }
\end{align*}
where $E$ is some {\co}value which forces the evaluation of the bound variable
$x$.  Note how there may be many more delayed $\tmu$-bindings allocated during
this evaluation step, which might be referenced inside $E$.  At this point, no
more progress can be made under the $\tmu$-binding of $x$, so that the whole
$\tmu$-abstraction is done.  We must now refocus our attention to the term bound
to $x$, which is modeled by finally now performing a $\betamu$ on the command:
\begin{align*}
  \cut
  {\outp\beta c'}
  {
    \inp x
    \cut{v_1}{\inp{y_1} \dots \cut{v_n}{\inp{y_n} \cut{x}{E}}}
  }
  &\red[\betamu]
  c'
  \subst
  {\beta}
  {
    \inp x
    \cut{v_1}{\inp{y_1} \dots \cut{v_n}{\inp{y_n} \cut{x}{\alpha}}}
  }
\end{align*}
This step allows for the term $\outp\beta c'$ to be evaluated by running the
command $c'$ with a substitution for $\beta$.  As $c'$ runs, a value may be
``returned'' to its consumer $\beta$ like so
\begin{align*}
  \cut
  {V}
  {
    \inp x
    \cut{v_1}{\inp{y_1} \dots \cut{v_n}{\inp{y_n} \cut{x}{\alpha}}}
  }
  &\red[\betatmu]
  \cut{v_1\subst{x}{V}}{\inp{y_1} \dots \cut{v_n\subst{x}{V}}{\inp{y_n}\cut{V}{\alpha}}}
\end{align*}
Since $\beta$ may be used many times, several values may be ``returned'' to the
consumer of $\outp\beta c'$ along $\beta$, but $c'$ is run at most once: this is
the essence of classical call-by-need evaluation.

With this intuition, we can simulate the same example above from the
call-by-need $\lambda$-calculus using the more machine-like sequent calculus,
where the identity function is written as $\fnc x \alpha \cut{x}{\alpha}$.  This
example illustrates the back-and-forth reductions caused by evaluating
underneath $\tmu$-bindings.
\begin{align*}
  \cut
  {\outp\beta\cut{I}{\app{I}{\beta}}}
  {
    \inp f
    \underline
    {
      \cut
      {I}
      {\app{f}{\app{f}{\app{5}{\alpha}}}}
    }
  }
  &\red[{\betar[\to]}]
  \underline
  {
    \cut
    {\outp\beta\cut{I}{\app{I}{\beta}}}
    {
      \inp f
      \cut
      {f}
      {\app{f}{\app{5}{\alpha}}}
    }
  }
  \\
  &\red[{\betamu}]
  \underline
  {
    \cut
    {I}
    {\app{I}{\inp f \cut{f}{\app{f}{\app{5}{\alpha}}}}}
  }
  \\
  &\red[{\betar[\to]}]
  \underline
  {
    \cut
    {I}
    {\inp f \cut{f}{\app{f}{\app{5}{\alpha}}}}
  }
  \\
  &\red[{\betatmu}]
  \underline{\cut{I}{\app{I}{\app{5}{\alpha}}}}
  \red[{\betar[\to]}]
  \underline{\cut{I}{\app{5}{\alpha}}}
  \red[{\betar[\to]}]
  \underline{\cut{5}{\alpha}}
\end{align*}
Notice how the expression $\outp\beta\cut{I}{\app{I}{\beta}}$ is delayed in the
first step but never copied, even though the bound variable $f$ is used multiple
times.  Also notice how every reduction step, besides the first one, occurs at
the top of the command.

\subsection{Classical call-by-need}

Since the sequent calculus corresponds to Gentzen's system LK of classical
logic, defining call-by-need evaluation for the sequent calculus gives a notion
of ``classical call-by-need,'' meaning call-by-need evaluation with first-class
control effects \cite{Griffin1990FATNC}.  Even though call-by-need is a more
intricate evaluation strategy than call-by-value and -name, it can be formulated
using the same technique: managing the priority between producers and consumers
through a notion of substitutability.  The only question is, which terms are
substitutable values, and dually which {\co}terms are {\co}values?

Substitutability means that values can be copied and deleted, and on this front
both call-by-value and call-by-need agree.  General $\mu$-abstractions, which
represents arbitrary computation, cannot be copied in either evaluation
strategies.  Instead, $\lambda$-abstractions---which are first-class
values---can be copied, and so too can injections of the form $\inj{i}{V}$ where
$V$ is another value.  The fact that $\inj{i}{v}$ is not a value in call-by-need
corresponds to the fact that the arguments to constructors are shared, not
computed more than once.  In contrast, the notion of {\co}values in call-by-need
corresponds to {\co}terms which ``force'' their input.  This includes all the
call-by-name forms of {\co}values, but this alone is not enough.

In the previous example of sharing in the sequent calculus, notice how we were
forced to substitute $\inp f \cut{f}{\app{f}{\app{5}{\alpha}}}$ in the step
\begin{align*}
  \cut
  {\outp\beta\cut{I}{\app{I}{\beta}}}
  {
    \inp f
    \cut
    {f}
    {\app{f}{\app{5}{\alpha}}}
  }
  &\red[{\betamu}]
  \cut
  {I}
  {\app{I}{\inp f \cut{f}{\app{f}{\app{5}{\alpha}}}}}
\end{align*}
If $\tmu$-abstractions could never be substituted like in call-by-name, then
this step would not be allowed and the computation would get stuck.  But we can
see that this specific instance is okay, because the input (named $f$ by the
$\tmu$-abstraction) is immediately needed in its body.  In general,
$\tmu$-abstractions of the form
\begin{math}
  \inp x \cut{v_1}{\inp {y_1} \dots \cut{v_n}{\inp {y_n} \cut{x}{E}}}
\end{math}
(where $x$ is different than $y_1 \dots y_n$) forces its input, because the
bound $x$ is immediately consumed by another forcing {\co}value $E$.

\begin{figure}
\centering

Sub-syntax for call-by-need evaluation:
\begin{align*}
  c &::= \cut{v}{e}
  &
  H &::= \hole \Alt \cut{v}{\inp x H}
  \\
  v &::= V \Alt \outp \alpha c
    &
  V &::= x \Alt \fnc x \alpha c \Alt \inl V \Alt \inr V
  \\
  e &::= E \Alt \inp x c
    &
  E &::= \alpha \Alt \app V E \Alt \casesum{x}{c_1}{y}{c_2}
    \Alt \inp x H[\cut{x}{E}]
\end{align*}

Call-by-need reduction rules:
\begin{align*}
  \rewriterule{\betamu[\CBNeed]}
  {
    \cut{\outp\alpha c}{E} &\red c\subst{\alpha}{E}
  }
  &
  \rewriterule{\betatmu[\CBNeed]}
  {
    \cut{V}{\inp x c} &\red c\subst{x}{V}
  }
  \\
  \rewriterule{\betar[\to][\CBNeed]}
  {
    \cut{\fnc x \alpha c}{\app{V}{E}}
    &\red
    c\subs{\asub{x}{V},\asub{\alpha}{E}}
  }
  &
  \rewriterule{\betar[\oplus][\CBNeed]}
  {
    \cut{\inj{i}{V}}{\casesum{x_1}{c_1}{x_2}{c_2}}
    &\red
    c_i\subs{\asub{x_i}{V}}
  }
\end{align*}

\caption{A call-by-need $\lmtm$ sub-calculus.}
\label{fig:lmtm-need}
\end{figure}

To make sure that computation does not get stuck, $\tmu$-abstractions of this
form must be considered substitutable {\co}values, which gives us the sub-syntax
for call-by-need $\lmtm$ shown in \cref{fig:lmtm-need}.  Intuitively, the set of
contexts $H$ represent a heap of variables bound to thunks (potentially
unevaluated terms).  Note that there is a side condition for {\co}values of the
form $\inp x H[\cut{x}{E}]$: $H$ must not contain a binding for $x$ which is in
scope over its hole $\hole$.  With this sub-syntax, we can use the reduction
rules given in \cref{fig:lmtm-need} for performing call-by-need evaluation.
Notice how the general form of the reduction rules are the same as in the other
call-by-name, call-by-value, and polarized calculi; the only real difference is
the definition of $V$ and $E$.

\subsection{The dual of call-by-need evaluation}

\begin{figure}
\centering

Sub-syntax for call-by-{\co}need evaluation:
\begin{align*}
  c &::= \cut{v}{e}
  &
  H &::= \hole \Alt \cut{\outp \alpha H}{e}
  \\
  v &::= V \Alt \outp \alpha c
    &
  V &::= x \Alt \fnc x \alpha c \Alt \inl V \Alt \inr V
    \Alt \outp \alpha H[\cut{V}{\alpha}]
  \\
  e &::= E \Alt \inp x c
    &
  E &::= \alpha \Alt \app V E \Alt \casesum{x}{c_1}{y}{c_2}
\end{align*}

Call-by-{\co}need reduction rules:
\begin{align*}
  \rewriterule{\betamu[\CoNeed]}
  {
    \cut{\outp\alpha c}{E} &\red c\subst{\alpha}{E}
  }
  &
  \rewriterule{\betatmu[\CoNeed]}
  {
    \cut{V}{\inp x c} &\red c\subst{x}{V}
  }
  \\
  \rewriterule{\betar[\to][\CoNeed]}
  {
    \cut{\fnc x \alpha c}{\app{V}{E}}
    &\red
    c\subs{\asub{x}{V},\asub{\alpha}{E}}
  }
  &
  \rewriterule{\betar[\oplus][\CoNeed]}
  {
    \cut{\inj{i}{V}}{\casesum{x_1}{c_1}{x_2}{c_2}}
    &\red
    c_i\subs{\asub{x_i}{V}}
  }
\end{align*}

\caption{The dual of call-by-need in the $\lmtm$ calculus.}
\label{fig:lmtm-coneed}
\end{figure}

One of the strengths of the sequent calculus is the way it presents the duality
between call-by-value and call-by-name evaluation
\cite{CurienHerbelin2000DC,Wadler2003CBVDCBN}.  Since producers are dual to
consumers, the dual can be found by just swapping terms with {\co}terms and
flipping the two sides of a command.  This simple syntactic transformation makes
it straightforward to discover the dual to classical call-by-need evaluation
\cite{AHS2011CCBND}, here called call-by-{\co}need.  We only need to exchange
the roles of values and {\co}values from the definition of call-by-need above,
and to dualize heaps so that {\co}variables are bound to delayed demands (a
consumer which might not need its input yet), as shown in
\cref{fig:lmtm-coneed}.  As with call-by-need, the reduction rules share the
same form, but with a new definition for $V$ and $E$.

Call-by-need evaluation can be seen as a more efficient approach to laziness
compared with call-by-name, where the computation of named results are
remembered and shared.  Once control effects are involved, call-by-{\co}need
evaluation can be seen as a more efficient approach to strictness compared with
call-by-value, where the computation of labeled call sites are remembered and
shared.  For example, consider the following use of the control operator
$\Callcc$ in the $\lambda$-calculus:
\begin{align*}
  \Let x = \Callcc(\fn k v) \In \Let n = \mathit{fib}~30 \In n + n
\end{align*}
where the call to $\mathit{fib}~30$ represents an expensive computation.  In
call-by-value, the invocation of $\Callcc$ is the first step, copying the
context $\Let x = \hole \In \Let n = \mathit{fib}~30 \In n+n$, and every time
$k$ is called $\mathit{fib}~30$ will be recomputed.  However, in
call-by-{\co}need, $\mathit{fib}~30$ is only computed once---the first time
$k$ is called with a value---after which the context is replaced with the
normalized $\Let x = \hole \In 1664080$ that can immediately return its result.

\section{Compiling Call-by-need and its Dual}
\label{sec:polarizingCBNeed}

Previously in \cref{sec:intro-polarity}, we were able to encode the
call-by-value and -name $\lmtm$-calculi into $\lmtmP$.  How could we accomplish
a similar thing for call-by-need and its dual?  First of all, the main key to
polarizing compilation is to use polarity shifts to mediate between the baked-in
evaluation strategies of the basic connectives ($\to$ and $\oplus$) and the
intended evaluation strategy of the source language.  This means that $\lmtmP$
on its own is not enough: before we can encode a language with a different
evaluation strategy like call-by-need, we will need to have enough shifts to
mediate between call-by-need and the call-by-value and -name strategies of the
basic connectives.

But compilation should do more than just get the compiled program to type check;
it should also exactly reflect the semantics of the source program in the target
calculus.  And on this point, we need to take great care with how the shifts
interact with the semantics of evaluation strategies besides call-by-value and
-name.  Consider the first polarizing compilation we gave in
\cref{fig:lmtmT-encoding}.  Recall how the translation of a sum injection
$(\inj{i}{V})^\CBN$ ends up burying the sum inside of a $\lambda$-abstraction.
This style of compilation is not acceptable for a call-by-need language, because
it is sensitive to \emph{sharing}: which results of sub-computations are
computed only once and shared from then on.

In a call-by-need language, all named terms are shared, and moreover, the
contents of a constructor like $\Inj{i}$ are \emph{also} shared.  In contrast,
the body of a $\lambda$-abstraction is never shared in a call-by-need language,
and is re-computed every time the function is called.  This dichotomy does not
fit with the compilation given in \cref{fig:lmtmT-encoding} for
$(\inj{i}{V})^\CBN$ because the extra $\lambda$-abstraction surrounds the
contents of the $\Inj{i}$ constructor and breaks sharing.  And for the dual
reason, the call-by-value compilation of $(\app{V}{E})^\CBV$ in
\cref{fig:lmtmQ-encoding} does not work for call-by-{\co}need, because the extra
$\rlam$-abstraction buries the function application destructor and breaks
sharing of $V$ and $E$.

However, in \cref{sec:CBVCBNintoL} we discussed how to get rid of these extra
$\lambda$-abstraction barriers in the compiled code.  The key was to use a
different style of shift ($\FromNeg$ and $\FromPos$) to convert the outside of
the connective compared with the shifts ($\ToPos$ and $\ToNeg$) used to wrap the
sub-formulas of the connective.  This alternate compilation avoids burying
constructors and destructors inside of abstractions, and will therefore have the
correct sharing semantics for an evaluation strategy like call-by-need.

So to compile another evaluation strategy into a language like $\lmtmP$, we will
have to extend $\lmtmP$ with the two different styles of shift to mediate
between that evaluation strategy and call-by-value and -name.  We can use the
optimized compilation in \cref{sec:CBVCBNintoL} to guide our generalization.  We
need to use the thin shifts ($\ToPos$ and $\ToNeg$) to convert a
call-by-(\co)need type to either call-by-value or call-by-name, so these shifts
should be generalized on the kinds of types they can accept as inputs.  And
since we need to use the thick shifts ($\FromNeg$ and $\FromPos$) to convert
from a call-by-value or -name type built by one of $\to$ or $\oplus$ to a
call-by-(\co)need type, those shifts will should be generalized on the kinds of
types they can give as results.  In the end, we must extend the set of signs
$\strat{S}$ to include symbols denoting call-by-need ($\CBNeed$) and
call-by-{\co}need ($\CoNeed$), and generalize both styles of the shift
connectives as in \cref{fig:generalized-shifts}, where the unannotated shifts
from before would now be written more explicitly as $\ToPos[\CBN]$,
$\ToNeg[\CBV]$, $\FromNeg[\CBV]$, and $\FromPos[\CBN]$.

\begin{figure}
\centering
\begin{gather*}
  \strat{S} ::= \CBV \Alt \CBN \Alt \CBNeed \Alt \CoNeed
  \qquad
  A,B,C
  ::= X_{\strat{S}} \Alt A \to B \Alt A \oplus B
  \Alt \ToPos[\strat{S}]A \Alt \FromPos[\strat{S}]A
  \Alt \ToNeg[\strat{S}]A \Alt \FromNeg[\strat{S}]A
  \\[1ex]
  \axiom{X_{\strat{S}} : \strat{S}}
  \qqqquad
  \infer
  {\ToPos[\strat{S}]A : \CBV}
  {A : \strat{S}}
  \qqqquad
  \infer
  {\FromPos[\strat{S}]A : \strat{S}}
  {A : \CBV}
  \qqqquad
  \infer
  {\ToNeg[\strat{S}]A : \CBN}
  {A : \strat{S}}
  \qqqquad
  \infer
  {\FromNeg[\strat{S}]A : \strat{S}}
  {A : \CBN}
\end{gather*}
\caption{Generalized shifts to and from call-by-(\co)need.}
\label{fig:generalized-shifts}
\end{figure}

\begin{figure}
\centering

Generic polarizing compilation for
$\strat{S} \in {\CBNeed, \CoNeed, \CBV, \CBN}$.
\begin{gather*}
\begin{aligned}
  X^{\strat{S}}
  &\defeq
  X_{\strat{S}}
  &\quad
  (A \to B)^{\strat{S}}
  &\defeq
  \FromNeg[\strat{S}]
  ((\ToPos[\strat{S}] A^{\strat{S}}) \to (\ToNeg[\strat{S}] B^{\strat{S}}))
  &\quad
  (A \oplus B)^{\strat{S}}
  &\defeq
  \FromPos[\strat{S}]
  ((\ToPos[\strat{S}] A^{\strat{S}}) \oplus (\ToPos[\strat{S}] B^{\strat{S}}))
\end{aligned}
\\[1ex]
\cut{v}{e}^{\strat{S}} \defeq \cut{v^{\strat{S}}}[\strat{S}]{e^{\strat{S}}}
\\
\begin{aligned}
  x^{\strat{S}} &\defeq x
  &\qquad
  \alpha^{\strat{S}} &\defeq x
  \\
  (\outp \alpha c)^{\strat{S}}
  &\defeq
  \outp \alpha (c^{\strat{S}})
  &\qquad
  (\inp x c)^{\strat{S}}
  &\defeq
  \inp x (c^{\strat{S}})
  \\
  (\fnc{x}{\alpha}c)^{\strat{S}}
  &\defeq
  \fn
  {
    [
    \force[\strat{S}]
    {
      [\app{[\wrap[\strat{S}]{x}]}{[\unwrap[\strat{S}]{\alpha}]}]
    }
    ]
  }
  c^{\strat{S}}
  &\qquad
  (\app{V}{E})^{\strat{S}}
  &\defeq
  \force[\strat{S}]
  {
    [\app{[\wrap[\strat{S}]{V^{\strat{S}}}]}{[\unwrap[\strat{S}]{E^{\strat{S}}}]}]
  }
  \\
  (\inj{i}{V})^{\strat{S}}
  &\defeq
  \delay[\strat{S}]
  {
    (\inj{i}{(\wrap[\strat{S}]{V^{\strat{S}}})})
  }
  &\qquad
  (\caseof{\recv{\inj{i}{x_i}}{c_i}})^{\strat{S}}
  &\defeq
  \caseof
  {
    \recv
    {\delay[\strat{S}]{(\inj{i}{(\wrap[\strat{S}]{x_i})})}}
    {c_i^{\strat{S}}}
  }
\end{aligned}
\end{gather*}

\caption{Generic compilation of call-by-need and call-by-{\co}need using the
  four shifts.}
\label{fig:generic-polarization}
\end{figure}

The generalized shifts let us give a generic polarizing compilation shown in
\cref{fig:generic-polarization} for the $\lmtm$-calculi following \emph{any} of
these four evaluation strategies.  For clarity and terseness, this generic
compilation is written using nested (\co)patterns.  But the nested (\co)patterns
can be mechanically flattened out into the simpler forms by matching one step at
a time, analogous to the desugaring of pattern matching in a functional
language, like so:
\begin{align*}
  (\fnc{x}{\alpha}c)^{\strat{S}}
  &\defeq
  \cocaseforce[\strat{S}]{\gamma}
  {
    \cut
    {
      \fnc{y}{\beta}
      \cut
      {y}
      [\CBV]
      {
        \casewrap[\strat{S}]{x}
        {
          \cut
          {\cocaseunwrap[\strat{S}]{\alpha}{c^{\strat{S}}}}
          [\CBN]
          {\beta}
        }
      }
    }
    [\CBN]
    {\gamma}
  }
  \\
  (\caseof{\recv{\inj{i}{x_i}}{c_i}})^{\strat{S}}
  &\defeq
  \casedelay[\strat{S}]{z}
  {
    \cut
    {z}
    [\CBV]
    {
      \caseof
      {
        \recv
        {\inj{i}{y_i}}
        {
          \cut
          {y_i}
          [\CBV]
          {\casewrap[\strat{S}]{x_i}{c_i^{\strat{S}}}}
        }
      }
    }
  }
\end{align*}

Note that in this generic compilation for an evaluation strategy denoted by
$\strat{S}$, type variables $X$ are just translated to type variables of the
appropriate sign, $X_{\strat{S}}$.  Furthermore, every binary connective is
expanded into a polarized connective interspersed with three shifts: two shifts
(some combination of $\ToPos[\strat{S}]$ and $\ToNeg[\strat{S}]$) to convert the
two sub-formulas, and one shift (either $\FromPos[\strat{S}]$ or
$\FromNeg[\strat{S}]$) to convert the final type.  In this sense, the
compilation of $\oplus$ is exactly like the optimized call-by-name compilation
given in \cref{sec:CBVCBNintoL}, except generalized to an arbitrary $\strat{S}$,
because it already had the maximal number of shifts.  In contrast, the
compilation of $\to$ has one more shift than the optimized call-by-value
compilation given in \cref{sec:CBVCBNintoL} because the type of the function
argument is not known to be positive already.

\section{The Dual Core Calculus: System \texorpdfstring{$\DUAL$}{D}}
\label{sec:dual-core}

So far, we have seen two independent methods of integrating call-by-value and
call-by-name features into a classical calculus: polarity which directly mixes
the two, and call-by-(\co)need evaluation which shares delayed computations.  We
then illustrated how we might reconcile these two, so that they can be used in
tandem, by generalizing beyond a binary polarity to one which mixes sharing with
the polarized connectives.  Doing so allows us to embed each of call-by-value,
-name, -need, and -{\co}need into a common core language.  However, there
remains some issues which still need to be resolved.  How do we combine the
semantics of all four evaluation strategies into a cohesive calculus?  How do we
extend our results beyond the simple functions-and-sums testing ground to a wide
variety of other types that commonly appear in programming languages?

We resolve these issues with the design of a core calculus, called
$\DUAL$,\footnote{So named because it is a fully $\D$ual calculus with many
  $\D$isciplines.} which in addition to subsuming all of call-by-value, -name,
-need, and -{\co}need evaluation it presents the following features:
\begin{itemize}
\item $\DUAL$ only includes a fixed, finite set of basic types. As discussed in
  more detail later in \cref{sec:dual-compile,sec:encoding-type-iso}, these
  basic types are expressive enough to represent a wide range of types that a
  programmer might define and use in a programming language.
\item Duality and symmetry plays a central role to our framework, and so we
  prefer a fully-dual presentation whenever possible.  As such, $\DUAL$ avoids
  the antisymmetric function type in favor of the symmetric \emph{par} type
  $A \parr B$---inspired by Girard's linear logic \cite{Girard19871LL}---which
  is exactly dual to a tuple type.  Unlike function types, where $A \to B$ is
  generally quite different from $B \to A$, par types have the pleasant property
  that $A \parr B$ is isomorphic to $B \parr A$.
\item Modern programming languages with static types typically have some
  mechanism for type abstraction, so $\DUAL$ includes two dual quantifiers which
  have a solid grounding in logic.  The universal quantifier $\forall$ models
  parametric polymorphism \cite{Girard1972PhD,Reynolds1974TTTS,GTL1989PT}, and
  the existential quantifier $\exists$ models modules and abstract data types
  \cite{Pierce2002TPL,Harper2012PFPL}.  Although these quantifiers are standard
  extensions, they impose some serious constraints when reasoning about the
  correctness of the encodings, as discussed later in
  \cref{sec:encoding-type-iso}.
\item Since $\DUAL$ is based on the classical sequent calculus, it already has a
  natural notion of control effect due to the duality between information flow
  and control flow.  Therefore, our semantics and theory is robust in the face
  of effects like exceptions and recursion.
\end{itemize}

\subsection{Syntax}
\label{sec:dual-core-syntax}

The syntax of the $\DUAL$ calculus, given in
\cref{fig:dual-core-syntax}, is an extension of $\lmtmP$ that was
presented in \cref{sec:intro-polarity}. Notice how commands
$\cut{v}[\ann{A}{\strat{S}}]{e}$ are not just marked with a  sign
$\strat{s}$ but also with the type $A$ of both $v$ and $e$.  The only purpose of
annotating the type $A$ in the command is to aid in type checking, and it can
otherwise be ignored.  As such, we will just use the shorthand
$\cut{v}[\strat{S}]{e}$ when $A$ is clear from context or irrelevant to the
particular example.  As additional shorthand, for clarity we will write binary
connectives as infix (\eg $A \oplus B$ instead of ${\oplus}~A~B$) and the
quantifiers with their binding (\eg $\forall X{:}\strat{S}. A$ instead of
$\forall_{\strat{S}} (\fn{X{:}\strat{S}} A)$).

Recall from \cref{sec:intro-sequent,sec:intro-polarity,sec:intro-need} how
substitutability was an essential notion for understanding the
semantics of programs.  Substitutability also plays a central role of our core
calculus, which has a built-in notion of \emph{discipline}, denoted by ($\strat{S}$). 
A discipline  is used to stipulate which terms ($V_{\strat{S}}$) might
replace variables and which {\co}terms ($E_{\strat{S}}$) might replace
{\co}variables.  This is enough to distinguish the different semantics of the
four evaluation strategies, which are marked by a different discipline symbol:
call-by-value ($\CBV$), call-by-name ($\CBN$), call-by-need ($\CBNeed$), and
call-by-{\co}need ($\CoNeed$).

The four different disciplines are then defined through the respective
$V_{\strat{S}}$ and $E_{\strat{S}}$, which are essentially the same as the
previous $\lmtm$-calculi from \cref{sec:intro-polarity,sec:intro-need}.  The
only difference is that a heap $H$ can contain now any mixture of both delayed
call-by-need ($\CBNeed$) variable bindings, and delayed call-by-{\co}need
($\CoNeed$) {\co}variable bindings.

 The $\DUAL$ core calculus has
several additional basic types than was discussed in
\cref{sec:intro-sequent,sec:intro-polarity,sec:intro-need}, which can be
understood in the following groups.

%Recall from \cref{sec:intro-sequent,sec:intro-polarity,sec:intro-need} how
%understanding substitutability was an essential step for understanding the
%semantics of programs.  Substitutability plays a central role of our core
%calculus, which has a built-in notion of \emph{discipline}: a definition of a
%subset of terms called \emph{values} and a subset of {\co}terms called
%\emph{{\co}values}.  A discipline is syntactically denoted by a sign
%($\strat{S}$), and is used to stipulate which terms ($V_{\strat{S}}$) might
%replace variables and which {\co}terms ($E_{\strat{S}}$) might replace
%{\co}variables.  This is enough to distinguish the different semantics of the
%four evaluation strategies, which are marked by a different discipline symbol:
%call-by-value ($\CBV$), call-by-name ($\CBN$), call-by-need ($\CBNeed$), and
%call-by-{\co}need ($\CoNeed$).
%
%The four different disciplines are then defined through the respective
%$V_{\strat{S}}$ and $E_{\strat{S}}$, which are essentially the same as the
%previous $\lmtm$-calculi from \cref{sec:intro-polarity,sec:intro-need}.  The
%only difference is that a heap $H$ can contain now any mixture of both delayed
%call-by-need ($\CBNeed$) variable bindings, and delayed call-by-{\co}need
%($\CoNeed$) {\co}variable bindings.

\begin{figure}

Commands $c$, terms $v$, {\co}terms $e$, weak-head normal terms $W$, and forcing
{\co}terms $F$:
\begin{gather*}
\begin{aligned}
  c &::= \cut{v}[\ann{A}{\strat{S}}]{e}
  &\qqqquad
  v &::= W \Alt \outp{\alpha} c
  % &\Alt \fix{\annvar x} v &
  &\qqqquad
  e &::= F \Alt \inp{x} c
  % &\Alt \cofix{\annvar \alpha} e &
\end{aligned}
\\
\begin{alignedat}{12}
  W
  &::= x &
  &\Alt \cocaseof{\many[i]{\send{q_i}{c_i} \mid}} &
  &\Alt \inl W &
  &\Alt \inr W &
  &\Alt (W, W') &
  &\Alt () &
  &\Alt \cont F &
  &\Alt \pack A W &
  &\Alt \wrap[\strat{S}] V_{\strat{S}} &
  &\Alt \delay[\strat{S}] W &
  \\
  F
  &::= \alpha &
  &\Alt \caseof{\many[i]{\recv{p_i}{c_i} \mid}} &
  &\Alt \fst F &
  &\Alt \snd F &
  &\Alt [F, F'] &
  &\Alt [] &
  &\Alt \throw W &
  &\Alt \spec A F &
  &\Alt \unwrap[\strat{S}] E_{\strat{S}} &
  &\Alt \force[\strat{S}] F &
\end{alignedat}
\end{gather*}

Patterns $p$ and {\co}patterns $q$:
\begin{gather*}
% \begin{aligned}
%   \annvar x, \annvar y, \annvar z
%   &::= \anns{x}{A}{k}
%   &\qqqquad\qqqquad
%   \annvar\alpha, \annvar\beta, \annvar\delta
%   &::= \anns{\alpha}{A}{k}
% \end{aligned}
% \\
\begin{alignedat}{12}
  p
  &::= \inl x &
  &\Alt \inr x &
  &\Alt (x, y) &
  &\Alt () &
  &\Alt \cont \alpha &
  &\Alt \pack X y &
  &\Alt \wrap[\strat{S}] x &
  &\Alt \delay[\strat{S}] x &
  \\
  q
  &::= \fst \alpha &
  &\Alt \snd \alpha &
  &\Alt [\alpha, \beta] &
  &\Alt [] &
  &\Alt \throw x &
  &\Alt \spec X \alpha &
  &\Alt \unwrap[\strat{S}] \alpha &
  &\Alt \force[\strat{S}] \alpha &
\end{alignedat}
\end{gather*}

Discipline-specific values $V_{\strat{S}}$, {\co}values $E_{\strat{S}}$, and
heap contexts $H$:
\begin{gather*}
\begin{aligned}
  V_{\CBN} &::= v
  &\quad
  V_{\CBV} &::= W
  &\quad
  V_{\CBNeed} &::= W
  &\quad
  V_{\CoNeed} &::= W \Alt \outp\alpha H[\cut{V_{\CoNeed}}[\ann{A}{\CoNeed}]{\alpha}]
  \\
  E_{\CBN} &::= F
  &\quad
  E_{\CBV} &::= e
  &\quad
  E_{\CBNeed} &::= F \Alt \inp x H[\cut{x}[\ann{A}{\CBNeed}]{E_{\CBNeed}}]
  &\quad
  E_{\CoNeed} &::= F
\end{aligned}
\\
\begin{aligned}
  H
  &::= \hole
  \Alt \cut{v}[\ann{A}{\CBNeed}]{\inp x H}
  \Alt \cut{\outp\alpha H}[\ann{A}{\CoNeed}]{e}
\end{aligned}
\end{gather*}

Types $A$, kinds $k$, disciplines $\strat{S}$, and connectives $\mk{F}$:
\begin{gather*}
\begin{aligned}
  A,B,C
  &::= X
  \Alt \mk{F}
  \Alt \fn{X{:}k}A
  \Alt A ~ B
  &\qquad
  k, l
  &::= \strat{S}
  \Alt k \to l
  &\qquad
  \strat{R}, \strat{S}, \strat{T}
  &::= \CBV
  \Alt \CBN
  \Alt \CoNeed
  \Alt \CBNeed
  % &\qquad
  % \annvar X, \annvar Y, \annvar Z
  % &::= \ann{X}{k}
  \\
  \mk{F}, \mk{G}
  &::=
  \rlap
  {$
    \with
    \Alt \parr
    \Alt \top
    \Alt \bot
    \Alt \oplus
    \Alt \otimes
    \Alt 0
    \Alt 1
    \Alt \involn
    \Alt \involp
    \Alt \forall_{\strat{S}}
    \Alt \exists_{\strat{S}}
    \Alt \ToPos[\strat{S}]
    \Alt \ToNeg[\strat{S}]
    \Alt \FromPos[\strat{S}]
    \Alt \FromNeg[\strat{S}]
    $}
\end{aligned}
\end{gather*}

\caption{Syntax of System $\DUAL$: a core, fully dual, multi-discipline sequent
  calculus.}
\label{fig:dual-core-syntax}
\end{figure}

%With this in mind, the syntax of the $\DUAL$ dual calculus is given in
%\cref{fig:dual-core-syntax}, which is an extension of $\lmtmP$ that was
%presented in \cref{sec:intro-polarity}.  Notice how commands
%$\cut{v}[\ann{A}{\strat{S}}]{e}$ are not just marked with a discipline sign
%$\strat{s}$ but also the type $A$ of both $v$ and $e$.  The only purpose of
%annotating the type $A$ in the command is to aid in type checking, and it can
%otherwise be ignored.  As such, we will just use the shorthand
%$\cut{v}[\strat{S}]{e}$ when $A$ is clear from context or irrelevant to the
%particular example.  As additional shorthand, for clarity we will write binary
%connectives as infix (\eg $A \oplus B$ instead of ${\oplus}~A~B$) and the
%quantifiers with their binding (\eg $\forall X{:}\strat{S}. A$ instead of
%$\forall_{\strat{S}} (\fn{X{:}\strat{S}} A)$).  The $\DUAL$ core calculus has
%several additional basic types than was discussed in
%\cref{sec:intro-sequent,sec:intro-polarity,sec:intro-need}, which can be
%understood in the following groups.

\subsubsection*{Positive data types}

All positive data types hereditarily follow a call-by-value strategy.  They have
values which are built from other positive values with constructors and
{\co}values which consume their input by matching on the possible constructions
they might be given of the general form
$\caseof{\recv{p_1}{c_1}\mid\dots\mid\recv{p_n}{c_n}}$ (we will drop the braces
when there is only one branch).  In addition to the sum type $A \oplus B$ from
\cref{sec:intro-sequent}, we also have the tuple type $A \otimes B$ which
contains pair values of the form $(W_1, W_2)$ that can be matched with a
{\co}value of the form $\sn{(x,y)}c$ that unpacks the two components of the
pair.  We also include the nullary versions of sum and tuple types.  The nullary
tuple is the unit type $1$, which contains a unit value $()$ and the matching
{\co}value $\sn{()}c$ that just waits to receive a $()$ before running $c$.  The
nullary sum is the empty type $0$, which contains no constructed values, so that
the matching {\co}value $\caseof{}$ has no branches, because there are no
constructions it might receive.

\subsubsection*{Negative data types}

All negative {\co}data types hereditarily follow a call-by-name strategy.  They
have {\co}values which are built from other negative {\co}values with
destructors and values which respond to their observer by matching on the
possible destruction they might encounter of the general form
$\cocaseof{\send{q_1}{c_1}\mid\dots\mid\send{q_n}{c_n}}$.  Dual to the sum type,
we have the product type $A \with B$ which contains projection {\co}values of
the form $\proj{i}{F}$ which is matched by a value of the form
$\cocaseof{\send{\fst\alpha}{c_1}\mid\send{\snd\beta}{c_2}}$ that will run the
command $c_1$ if asked for its first component or $c_2$ if asked for its second
component.  Dual to the tuple type, we have the par type $A \parr B$ which
contains the compound {\co}value $[F_1, F_2]$ that pairs up two other consumers
and can be matched by a value of the form $\fn{[\alpha,\beta]}c$ that unpacks
the pair of consumers before running $c$.  As before, we include the nullary
versions of product and par types.  The nullary product type is the trivial type
$\top$, which contains no possible destructors, so that the matching value
$\cocaseof{}$ has no branches since it will never need to respond to any
request.  The nullary par type is the absurd type $\bot$, which contains the
terminal destructor $[]$ with the matching value $\fn{[]}c$ that just waits for
the empty signal $[]$ before running $c$.

\subsubsection*{Involutive negation}

The two negation types can be thought of as two dual way for representing
first-class continuations in a programming language.  One way to formulate a
continuation is by capturing the context as a first class value.  This
corresponds to the data type $\involp A$ which packages up a {\co}value $F$ as
the value $\cont F$, which can be later unpacked by pattern-matching in
$\sn{\cont\alpha}c$.  Another way to formulate continuations is through
functions that never return.  This corresponds to the {\co}data type $\involn A$
which has values of the form $\fn{\throw x}c$, which is analogous to a function
abstraction that does not bind a return pointer, and {\co}values of the form
$\throw W$, which is analogous to throwing a value to a continuation without
waiting for a returned result.  This analogy reverses the usual order of ideas
by making continuations primitive, so that function types are instead derived
like so:
\begin{align*}
  A \to B &= (\involn A) \parr B
  &
  \app{W}{F} &= [\throw W, F]
  &
  \fnc x \alpha c &= \fn{[\throw x, \alpha]}c
  =
  \fn{[\beta,\alpha]}\cut{\fn{\throw x}c}{\beta}
\end{align*}

Also note that the two negation connectives belong to different disciplines and
are, in fact, polarity inverting.  The $\involn$ form of negation converts a
positive type $A : \CBV$ to the negative $\involn A : \CBN$, and dually
$\involp$ converts a negative type $B : \CBN$ to the positive
$\involp A : \CBV$.  The inversion between positive and negative corresponds to
the inversion between input and output inherent in continuations, and is
essential to making negation involutive \cite{MunchMaccagnoni2014FASIN}:
$\involp \involn A$ is isomorphic to $A$, and likewise $\involn \involp B$ is
isomorphic to $B$.  This would not hold if we forced the negation connectives to
preserve the same polarity of the given type as in \cite{Zeilberger2009PhD}.

\subsubsection*{Type quantifiers}

Lastly, we have the two quantifiers for abstracting over types in programs.  The
sequent calculus helps to show the duality of these two forms of type
abstraction as a syntactic symmetry between producers and consumers.  The
universal quantifier $\forall$ is modeled as a negative {\co}data type analogous
to the quantifier of System F, and can quantify over any of the four kinds of
types ($\CBV$, $\CBN$, $\CBNeed$, or $\CoNeed$).  Values of type
$\forall X{:}\strat{S}. A$ abstract over types of kind $\strat{S}$ similar to a
function of the form $\fn{\spec{X}{\alpha}}{c}$, which matches against a
destructor of the form $\spec{B}{F}$.  Notice that this means the consumer
chooses a private type $B$ for its observation, which forces polymorphism on the
value it interacts with.  Dually, the existential quantifier $\exists$ is
modeled as a positive data type, with exactly the reversed roles.  Values of
type $\exists X{:}\strat{S}. A$ choose a private type used in another value
which is not visible from the outside, and the two are packaged together as
$\pack{B}{W}$.  Consumers of existential packages must then abstract over the
privately chosen type, of the form $\sn{\pack{X}{y}}{c}$ similar to pattern
matching on tuples.

\subsubsection*{Shifts}

So far, all the basic types above follow the standard polarized recipe.  Our
treatment of the shifts is where $\DUAL$ really differs from other calculi.
First, let's consider the shifts which are modeled as data types---that is, with
constructors---which coerce other types into and out of the call-by-value world.
The shift $\ToPos[\strat{S}]$ embeds an $\strat{S}$-value of type
$A : \strat{S}$ to a positive data structure $\ToPos[\strat{S}] A : \CBV$.  The
construction $\wrap[\strat{S}]{V_{\strat{S}}}$ can be thought of like a box
containing the value $V_{\strat{S}}$, where the notion of ``value'' depends on
$\strat{S}$.  For example, if $\strat{S}$ is $\CBN$, then $V_{\strat{S}}$ may be
any term, but if it is $\CBV$, then it must be some $W$.  The matching form
$\casewrap[\strat{S}]{x}{c}$ unwraps this box and binds the value to $x$ in $c$.
In contrast, the shift $\FromPos[\strat{S}]$ embeds a positive data structure
$A : \CBV$ as a value of type $\FromPos[\strat{S}] A : \strat{S}$ in any
discipline.  The construction $\delay[\strat{S}]W$ can be thought of as
returning the call-by-value answer $W$ in some other context, and
$\casedelay[\strat{S}]{x}{c}$ forces computations of type $\FromPos[\strat{S}]A$
to bind this answer to $x$ in $c$.

The shifts which are modeled as {\co}data types--that is, with destructors---are
dual to the above, and allow for coercions into and out of the call-by-name
world.  The shift $\ToNeg[\strat{S}]$ embeds an $\strat{S}$-{\co}value of type
$A : \strat{S}$, that is a substitutable consumer of $A$s, into a negative
observation of type $\ToNeg[\strat{S}] A : \CBN$.  The destructor
$\unwrap[\strat{S}]{E_{\strat{S}}}$ can be thought of as recognizing that
$E_{\strat{S}}$ is strict on its input, where the options for the {\co}value
$E_{\strat{S}}$ depend on the discipline $\strat{S}$.  For example, if
$\strat{S}$ is $\CBV$, then any $e$ is a strict {\co}value, but if $\strat{S}$
is $\CBN$, then only a forcing $F$ is strict.  The matching
$\cocaseunwrap[\strat{S}]{\alpha}{c}$ is then a negative computation which binds
$\alpha$ to this strict {\co}value.  In contrast, the shift
$\FromNeg[\strat{S}]$ embeds a negative observation of type $A : \CBN$ as a
{\co}value of type $\FromNeg[\strat{S}] A : \strat{S}$ in other discipline.  The
observation $\force[\strat{S}]F$ can be thought of as an observation for forcing
some thunk, which is represented by $\cocaseforce[\strat{S}]{\alpha}{c}$.

\begin{intermezzo}
\label{thm:type-subsyntax}

Each of the connectives in \cref{fig:dual-core-syntax} has an intended
discipline, which corresponds with an evaluation strategy.  For example, the
tuple type constructor takes two call-by-value types and returns another
call-by-value type, \ie $\otimes : \CBV \to \CBV \to \CBV$.  Calculi in this
vein sometimes embed some additional amount of well-formedness into the syntax
of types themselves, rather than as an auxiliary check after the fact,
especially when higher kinds of the form $k \to l$ are left out.  The
corresponding well-formed sub-syntax of non-higher-kinded $\DUAL$ types (where
type variables $X : \strat{S}$ are explicitly annotated as $X_{\strat{S}}$) is
given by the following grammar, where we have a different syntax of types
$A_{\strat{S}}$ for each discipline $\strat{S}$:
\begin{align*}
  A_\CBV
  &::= X_\CBV
    \Alt \FromPos[\CBV] A_\CBV
    \Alt \FromNeg[\CBV] A_\CBN
    \Alt A_\CBV \oplus B_\CBV
    \Alt A_\CBV \otimes B_\CBV
    \Alt 0
    \Alt 1
    \Alt \involp A_\CBN
    \Alt \exists X{:}\strat{S}. A_\CBV
    \Alt \ToPos[\strat{S}] A_{\strat{S}}
  \\
  A_\CBN
  &::= X_\CBN
    \Alt \FromPos[\CBN] A_\CBV
    \Alt \FromNeg[\CBN] A_\CBN
    \Alt A_\CBN \with B_\CBN
    \Alt A_\CBN \parr B_\CBN
    \Alt \top
    \Alt \bot
    \Alt \involn A_\CBV
    \Alt \forall X{:}\strat{S}. A_\CBN
    \Alt \ToNeg[\strat{S}] A_{\strat{S}}
  \\
  A_\CBNeed
  &::= X_\CBNeed
    \Alt \FromPos[\CBNeed] A_\CBV
    \Alt \FromNeg[\CBNeed] A_\CBN
  \\
  A_\CoNeed
  &::= X_\CoNeed
    \Alt \FromPos[\CoNeed] A_\CBV
    \Alt \FromNeg[\CoNeed] A_\CBN
\end{align*}
A similar exercise can be done to sub-divide expressions (terms, {\co}terms,
\etc[]) according to discipline based on which evaluation strategy they follow.
But the downside is an explosion of cases in the syntactic grammar.
\end{intermezzo}

\subsection{Type system}
\label{sec:dual-core-types}

\begin{figure}
\centering
\begin{equation*}
\begin{aligned}
  \oplus &: \CBV \to \CBV \to \CBV
  &
  \otimes &: \CBV \to \CBV \to \CBV
  &
  0 &: \CBV
  &
  1 &: \CBV
  &
  \involp &: \CBN \to \CBV
  &
  \exists_{\strat{S}} &: (\strat{S} \to \CBV) \to \CBV
  \\
  \with &: \CBN \to \CBN \to \CBN
  &
  \parr &: \CBN \to \CBN \to \CBN
  &
  \top &: \CBN
  &
  \bot &: \CBN
  &
  \involn &: \CBV \to \CBN
  &
  \forall_{\strat{S}} &: (\strat{S} \to \CBN) \to \CBN
  \\
  \ToPos[\strat{S}] &: \strat{S} \to \CBV
  &
  \ToNeg[\strat{S}] &: \strat{S} \to \CBN
  &
  \FromPos[\strat{S}] &: \CBV \to \strat{S}
  &
  \FromNeg[\strat{S}] &: \CBN \to \strat{S}
\end{aligned}
\end{equation*}
\begin{gather*}
  \axiom{\typety[\DUAL]{\Theta,X:k}{X}{k}}
  \quad
  \infer
  {\typety[\DUAL]{\Theta}{\mk{F}}{k}}
  {\mk{F}:k}
  \quad
  \infer
  {\typety[\DUAL]{\Theta}{\fn{X{:}k}A}{k \to l}}
   {\typety[\DUAL]{\Theta,X:k}{A}{l}}
 % {\typety[\FUN]{\Theta,X:k}{A}{l}}
  \quad
  \infer
  {\typety[\DUAL]{\Theta}{A~B}{l}}
  {\typety[\DUAL]{\Theta}{A}{k \to l} & \typety[\DUAL]{\Theta}{B}{k}}
  \\[1ex]
  % \infer
  % {\typety[\DUAL]{\Theta}{A = B}{k}}
  % {
  %   \typety[\DUAL]{\Theta}{A}{k}
  %   &
  %   A \eq[\alpha\beta\eta] B
  %   &
  %   \typety[\DUAL]{\Theta}{B}{k}
  % }
  % \qqquad
  \infer
  {\seqwf[\DUAL][\Theta]{\many{x:A:\strat{T}}}{\many{\beta:B:\strat{R}}}}
  {
    \many{(\typety[\DUAL]{\Theta}{A}{\strat{T}})}
    &
    \many{(\typety[\DUAL]{\Theta}{B}{\strat{R}})}
  }
\end{gather*}
\caption{Kinds of System $\DUAL$ types.}
\label{fig:dual-core-kinds}
\end{figure}

The type system for the dual core calculus $\DUAL$ is given in
\cref{fig:dual-core-kinds,fig:dual-core-generic-types,fig:dual-core-positive-types,fig:dual-core-negative-types,fig:dual-core-shift-types}.
\cref{fig:dual-core-kinds} gives the kind system for checking that types are
well-formed, which is standard for calculi based on System F with higher kinds.
More specifically, the kinding rules have judgments of the form
$\typety[\DUAL]{\Theta}{A}{k}$---where $\Theta$ assigns a kind to each free type
variable $X$ of $A$---and are equivalent to a simply-typed $\lambda$-calculus at
the type level with variables, type constructors (for connectives),
applications, and abstractions.  The subscript $\DUAL$ signifies that types are
built using only the core connectives of System $\DUAL$, as opposed to other
sets of connectives that are allowed later on in \cref{sec:dual-ext}.  We can
also check that a sequent is well-formed with
$\seqwf[\DUAL][\Theta]{\Gamma}{\Delta}$ by checking that each type in $\Gamma$
and $\Delta$ has some kind $\strat{S}$ in the environment $\Theta$.  Finally, we
will treat the standard $\alpha\beta\eta$ equality between same-kinded types as
implicit for the purpose of type checking.  Formally, if
$A_1 \eq[\alpha\beta\eta] A_2$ and $\typety[\DUAL]{\Theta}{A_i}{k}$, then $A_1$
and $A_2$ are considered the same type in the environment $\Theta$.

\begin{figure}
\centering
\newcommand{\ruleskip}{0.25ex}
\newcommand{\axiomskip}{0.25ex}
\begin{gather*}
  \infer[\CutRule]
  {\typecmd[\DUAL][\Theta]{\Gamma}{\Delta}{\cut{v}[A{:}\strat{S}]{e}}}
  {
    \typetm[\DUAL][\Theta]{\Gamma}{\Delta}{v}{A}
    &
    \typety[\DUAL]{\Theta}{A}{\strat{S}}
    &
    \typecotm[\DUAL][\Theta]{\Gamma}{\Delta}{e}{A}
  }
  \\[\axiomskip]
  \axiom[\VR]{\typetmfoc[\DUAL][\Theta]{\Gamma,x:A}{\Delta}{x}{A}}
  \qqqquad
  \axiom[\VL]{\typecotmfoc[\DUAL][\Theta]{\Gamma}{\alpha:A,\Delta}{\alpha}{A}}
  \\[\ruleskip]
  \infer[\AR]
  {\typetm[\DUAL][\Theta]{\Gamma}{\Delta}{\outp\alpha c}{A}}
  {\typecmd[\DUAL][\Theta]{\Gamma}{\alpha:A,\Delta}{c}}
  \qqqquad
  \infer[\AL]
  {\typecotm[\DUAL][\Theta]{\Gamma}{\Delta}{\inp x c}{A}}
  {\typecmd[\DUAL][\Theta]{\Gamma,x:A}{\Delta}{c}}
  \\[\ruleskip]
  \infer[\FR]
  {\typetm[\DUAL][\Theta]{\Gamma}{\Delta}{v}{A}}
  {\typetmfoc[\DUAL][\Theta]{\Gamma}{\Delta}{v}{A}}
  \qqqquad
  \infer[\FL]
  {\typecotm[\DUAL][\Theta]{\Gamma}{\Delta}{e}{A}}
  {\typecotmfoc[\DUAL][\Theta]{\Gamma}{\Delta}{e}{A}}
  \\[\ruleskip]
  \infer[\BR]
  {\typetmfoc[\DUAL][\Theta]{\Gamma}{\Delta}{V_{\strat{S}}}{A}}
  {
    \typetm[\DUAL][\Theta]{\Gamma}{\Delta}{V_{\strat{S}}}{A}
    &
    \typety[\DUAL]{\Theta}{A}{\strat{S}}
  }
  \qqqquad
  \infer[\BL]
  {\typecotmfoc[\DUAL][\Theta]{\Gamma}{\Delta}{E_{\strat{S}}}{A}}
  {
    \typety[\DUAL]{\Theta}{A}{\strat{S}}
    &
    \typecotm[\DUAL][\Theta]{\Gamma}{\Delta}{E_{\strat{S}}}{A}
  }
  % \\[\ruleskip]
  % \infer[\TCR]
  % {\typetm[\DUAL][\Theta]{\Gamma}{\Delta}{v}{B}}
  % {
  %   \typetm[\DUAL][\Theta]{\Gamma}{\Delta}{v}{A}
  %   &
  %   \typety[\DUAL]{\Theta}{A = B}{k}
  % }
  % \qqqquad
  % \infer[\TCL]
  % {\typecotm[\DUAL][\Theta]{\Gamma}{\Delta}{e}{A}}
  % {
  %   \typety[\DUAL]{\Theta}{A = B}{k}
  %   &
  %   \typecotm[\DUAL][\Theta]{\Gamma}{\Delta}{e}{B}
  % }
  % \\[\ruleskip]
  % \infer[\TCR]
  % {\typetmfoc[\DUAL][\Theta]{\Gamma}{\Delta}{v}{B}}
  % {
  %   \typetmfoc[\DUAL][\Theta]{\Gamma}{\Delta}{v}{A}
  %   &
  %   \typety[\DUAL]{\Theta}{A = B}{k}
  % }
  % \qqqquad
  % \infer[\TCL]
  % {\typecotm[\DUAL][\Theta]{\Gamma}{\Delta}{e}{A}}
  % {
  %   \typety[\DUAL]{\Theta}{A = B}{k}
  %   &
  %   \typecotmfoc[\DUAL][\Theta]{\Gamma}{\Delta}{e}{B}
  % }
\end{gather*}
\caption{Typing rules of System $\DUAL$ that apply to any type.}
\label{fig:dual-core-generic-types}
\end{figure}

The generic typing rules that apply to any type are given in
\cref{fig:dual-core-generic-types}.  The main departure from the previous
polarized type system in \cref{sec:intro-polarity} is that typing judgments
using a stoup ($\stoup$) have been generalized beyond syntactic forms $W$ and
$F$ to the more semantic notions of value and {\co}value.  That is, the judgment
$\typetmfoc[\DUAL][\Theta]{\Gamma}{\Delta}{v}{A}$ is well-formed when $v$ is a
value $V_{\strat{S}}$ and $A$ is an $\strat{S}$-type such that
$\typety[\DUAL]{\Theta}{A}{\strat{S}}$ is derivable, and similarly for the dual
judgment.  By relaxing syntactic stipulations on the stoup, we can accommodate
call-by-need (having a notion of {\co}value that is bigger than $F$ and smaller
than $e$) and call-by-{\co}need (with values between $W$ and $v$).  To formalize
this idea, we add the $\BR$ and $\BL$ rules---opposite in direction to $\FR$ and
$\FL$---that \emph{only} apply to values and {\co}values as defined by the
discipline of their type.  By the analogy in \cref{rm:focus}, this gives a
systematic approach to the connection between evaluation and focusing that
applies to more strategies than just call-by-value and call-by-name.

Note that type checking is still decidable even though (\co)variables binders
are not annotated with their type.  This is possible because of the sequent
calculus' \emph{sub-formula property} \citep{Gentzen1935UULS1}---every type
appearing in a premise of an inference rule appears somewhere in the conclusion
of that rule.  The only exception to the sub-formula property is the $\CutRule$
rule that introduces arbitrary new types in the premises.  The new type
introduced by a $\CutRule$ is made manifest in the syntax of expressions as the
only necessary typing annotation, so that there is never any need to ``guess''
or infer during type checking.  The ability to avoid inference altogether is an
interesting property of the sequent calculus not shared by the
$\lambda$-calculus.  Effectively, every elimination rule of the
$\lambda$-calculus is ``hiding'' an implicit cut, which forces us to reconstruct
the type of the implicit cut during checking.  The sub-formula property also
aids in separating typing from checking that judgments are well-formed by just
checking that the type of a $\CutRule$ makes sense in the current environment:
if the conclusion of an inference rule is well-formed (\ie
$\seqwf[\DUAL][\Theta]{\Gamma}{\Delta}$), then so too are the premises.
% Furthermore, we have type conversion rules $\TCR$ and $\TCL$ that replace
% equal types when type checking terms and {\co}terms.

\begin{figure}
\centering
% \begin{math}
% \begin{aligned}
%   \oplus &: \CBV \to \CBV \to \CBV
%   &
%   \otimes &: \CBV \to \CBV \to \CBV
%   &
%   0 &: \CBV
%   &
%   1 &: \CBV
%   &
%   \involp &: \CBN \to \CBV
%   &
%   \exists_{\strat{S}} &: (\strat{S} \to \CBV) \to \CBV
% \end{aligned}
% \end{math}
% \\[1ex]
\renewcommand{\oftype}{{\,:\,}}
\renewcommand{\ofkind}{{\,:\,}}
\renewcommand{\entails}{{\,\vdash}}
\renewcommand{\act}{{\,|\,}}
\newcommand{\ruleskip}{0.5ex}
\newcommand{\axiomskip}{0.5ex}
\newcommand{\namespace}{}%{\!\!\!}
\begin{math}
\begin{gathered}
  \infer[\namespace{\oplusR[i]}]
  {\typetmfoc[\DUAL][\Theta]{\Gamma}{\Delta}{\inj{i} W}{A_1 \oplus A_2}}
  {\typetmfoc[\DUAL][\Theta]{\Gamma}{\Delta}{W}{A_i}}
  \qqqquad
  \infer[\namespace\oplusL]
  {
    \typecotmfoc[\DUAL][\Theta]{\Gamma}{\Delta}
    {\caseof{\send{\inl x}{c_1}|\send{\inr y}{c_2}}}
    {A \oplus B}
  }
  {
    \typecmd[\DUAL][\Theta]{\Gamma, x \oftype A}{\Delta}{c_1}
    &
    \typecmd[\DUAL][\Theta]{\Gamma, y \oftype B}{\Delta}{c_2}
  }
  \\[\ruleskip]
  \infer[\namespace\otimesR]
  {\typetmfoc[\DUAL][\Theta]{\Gamma}{\Delta}{(W_1,W_2)}{A \otimes B}}
  {
    \typetmfoc[\DUAL][\Theta]{\Gamma}{\Delta}{W_1}{A}
    &
    \typetmfoc[\DUAL][\Theta]{\Gamma}{\Delta}{W_2}{B}
  }
  \qqqquad
  \infer[\namespace\otimesL]
  {
    \typecotmfoc[\DUAL][\Theta]{\Gamma}{\Delta}
    {\sn{(x, y)}c}
    {A \otimes B}
  }
  {\typecmd[\DUAL][\Theta]{\Gamma, x \oftype A, y \oftype B}{\Delta}{c}}
  \\[\ruleskip]
  \text{no $\zeroR$ rules}
  \qquad
  \axiom[\namespace{\zeroL}]
  {\typecotmfoc[\DUAL][\Theta]{\Gamma}{\Delta}{\caseof{}}{0}}
  \qquad
  \axiom[\namespace\oneR]
  {\typetmfoc[\DUAL][\Theta]{\Gamma}{\Delta}{()}{1}}
  \qquad
  \infer[\namespace\oneL]
  {\typecotmfoc[\DUAL][\Theta]{\Gamma}{\Delta}{\sn{()}c}{1}}
  {\typecmd[\DUAL][\Theta]{\Gamma}{\Delta}{c}}
  \\[\ruleskip]
  \infer[\namespace\involpR]
  {\typetmfoc[\DUAL][\Theta]{\Gamma}{\Delta}{\cont F}{\involp A}}
  {\typecotmfoc[\DUAL][\Theta]{\Gamma}{\Delta}{F}{A}}
  \qqqquad
  \infer[\namespace\involpL]
  {
    \typecotmfoc[\DUAL][\Theta]{\Gamma}{\Delta}
    {\sn{\cont\alpha}c}{\involp A}
  }
  {\typecmd[\DUAL][\Theta]{\Gamma}{\alpha \oftype A, \Delta}{c}}
  \\[\ruleskip]
  \infer[\namespace\existsR]
  {\typetmfoc[\DUAL][\Theta]{\Gamma}{\Delta}{\pack B W}{\exists_{\strat{S}} A}}
  {
    \typety[\DUAL]{\Theta}{B}{\strat{S}}
    &
    \typetmfoc[\DUAL][\Theta]{\Gamma}{\Delta}{W}{A~B}
  }
  \qqqquad
  \infer[\namespace\existsL]
  {
    \typecotmfoc[\DUAL][\Theta]{\Gamma}{\Delta}
    {\sn{(\pack X y)}c}{\exists_{\strat{S}} A}
  }
  {\typecmd[\DUAL][\Theta, X{\ofkind}\strat{S}]{\Gamma, y \oftype A~X}{\Delta}{c}}
\end{gathered}
\end{math}
\caption{Typing rules of the System $\DUAL$ positive connectives.}
\label{fig:dual-core-positive-types}
\end{figure}
\begin{figure}
\centering
% \begin{math}
% \begin{aligned}
%   \with &: \CBN \to \CBN \to \CBN
%   &
%   \parr &: \CBN \to \CBN \to \CBN
%   &
%   \top &: \CBN
%   &
%   \bot &: \CBN
%   &
%   \involn &: \CBV \to \CBN
%   &
%   \forall_{\strat{S}} &: (\strat{S} \to \CBN) \to \CBN
% \end{aligned}
% \end{math}
% \\[1ex]
\renewcommand{\oftype}{{\,:\,}}
\renewcommand{\ofkind}{{\,:\,}}
\renewcommand{\entails}{{\,\vdash}}
\renewcommand{\act}{{\,|\,}}
\newcommand{\ruleskip}{0.5ex}
\newcommand{\axiomskip}{0.5ex}
\newcommand{\namespace}{}%{\!\!\!}
\begin{math}
\begin{gathered}
  \infer[\namespace{\withL[i]}]
  {\typecotmfoc[\DUAL][\Theta]{\Gamma}{\Delta}{\proj{i} F}{A_1 \with A_2}}
  {\typecotmfoc[\DUAL][\Theta]{\Gamma}{\Delta}{F}{A_i}}
  \qqqquad
  \infer[\namespace\withR]
  {
    \typetmfoc[\DUAL][\Theta]{\Gamma}{\Delta}
    {\cocaseof{\send{\fst\alpha}{c_1}\mid\send{\snd\beta}{c_2}}}
    {A \with B}
  }
  {
    \typecmd[\DUAL][\Theta]{\Gamma}{\alpha \oftype A,\Delta}{c_1}
    &
    \typecmd[\DUAL][\Theta]{\Gamma}{\beta \oftype B,\Delta}{c_2}
  }
  \\[\ruleskip]
  \infer[\namespace\parrL]
  {\typecotmfoc[\DUAL][\Theta]{\Gamma}{\Delta}{[F_1,F_2]}{A \parr B}}
  {
    \typecotmfoc[\DUAL][\Theta]{\Gamma}{\Delta}{F_1}{A}
    &
    \typecotmfoc[\DUAL][\Theta]{\Gamma}{\Delta}{F_2}{B}
  }
  \qqqquad
  \infer[\namespace\parrR]
  {
    \typetmfoc[\DUAL][\Theta]{\Gamma}{\Delta}
    {\fn{[\alpha,\beta]}c}
    {A \parr B}
  }
  {\typecmd[\DUAL][\Theta]{\Gamma}{\alpha \oftype A, \beta \oftype B,\Delta}{c}}
  \\[\ruleskip]
  \axiom[\namespace{\topR}]
  {\typetmfoc[\DUAL][\Theta]{\Gamma}{\Delta}{\cocaseof{}}{\top}}
  \qquad
  \text{no $\topL$ rules}
  \qquad
  \axiom[\namespace\botL]
  {\typecotmfoc[\DUAL][\Theta]{\Gamma}{\Delta}{[]}{\bot}}
  \qquad
  \infer[\namespace\botR]
  {\typetmfoc[\DUAL][\Theta]{\Gamma}{\Delta}{\fn{[]}c}{\bot}}
  {\typecmd[\DUAL][\Theta]{\Gamma}{\Delta}{c}}
  \\[\ruleskip]
  \infer[\namespace\involnL]
  {\typecotmfoc[\DUAL][\Theta]{\Gamma}{\Delta}{\throw W}{\involn A}}
  {\typetmfoc[\DUAL][\Theta]{\Gamma}{\Delta}{W}{A}}
  \qqqquad
  \infer[\namespace\involnR]
  {
    \typetmfoc[\DUAL][\Theta]{\Gamma}{\Delta}
    {\fn{\throw x}c}{\involn A}
  }
  {\typecmd[\DUAL][\Theta]{\Gamma, x \oftype A}{\Delta}{c}}
  \\[\ruleskip]
  \infer[\namespace\forallL]
  {\typecotmfoc[\DUAL][\Theta]{\Gamma}{\Delta}{\spec B F}{\forall_{\strat{S}} A}}
  {
    \typety[\DUAL]{\Theta}{B}{\strat{S}}
    &
    \typecotmfoc[\DUAL][\Theta]{\Gamma}{\Delta}{F}{A~B}
  }
  \qqqquad
  \infer[\namespace\forallR]
  {
    \typetmfoc[\DUAL][\Theta]{\Gamma}{\Delta}
    {\fn{[\spec X \alpha]}c}{\forall_{\strat{S}} A}
  }
  {\typecmd[\DUAL][\Theta, X{\ofkind}\strat{S}]{\Gamma}{\alpha \oftype A~X,\Delta}{c}}
\end{gathered}
\end{math}
\caption{Typing rules of the System $\DUAL$ negative connectives.}
\label{fig:dual-core-negative-types}
\end{figure}

\begin{figure}
\centering
% \begin{math}
% \begin{aligned}
%   \ToPos[\strat{S}] &: \strat{S} \to \CBV
%   &
%   \FromNeg[\strat{S}] &: \CBN \to \strat{S}
%   \\
%   \ToNeg[\strat{S}] &: \strat{S} \to \CBN
%   &
%   \FromPos[\strat{S}] &: \CBV \to \strat{S}
% \end{aligned}
% \end{math}
% \\[1ex]
\renewcommand{\oftype}{{\,:\,}}
\renewcommand{\ofkind}{{\,:\,}}
\renewcommand{\entails}{{\,\vdash}}
\renewcommand{\act}{{\,|\,}}
\newcommand{\ruleskip}{0.5ex}
\newcommand{\axiomskip}{0.5ex}
\newcommand{\namespace}{}%{\!\!\!}
\begin{math}
\begin{gathered}
  \infer[\namespace{\ToPosR}]
  {
    \typetmfoc[\DUAL][\Theta]{\Gamma}{\Delta}
    {\wrap[\strat{S}]{V_{\strat{S}}}}{\ToPos[\strat{S}] A}
  }
  {\typetmfoc[\DUAL][\Theta]{\Gamma}{\Delta}{V_{\strat{S}}}{A}}
  \qqqquad
  \infer[\namespace{\ToPosL}]
  {
    \typecotmfoc[\DUAL][\Theta]{\Gamma}{\Delta}
    {\sn{\wrap[\strat{S}] x}c}
    {\ToPos[\strat{S}] A}
  }
  {\typecmd[\DUAL][\Theta]{\Gamma, x \oftype A}{\Delta}{c}}
  \\[\ruleskip]
  \infer[\namespace{\FromNegR}]
  {
    \typetmfoc[\DUAL][\Theta]{\Gamma}{\Delta}
    {\fn{\force[\strat{S}] \alpha}c}
    {\FromNeg[\strat{S}] A}
  }
  {\typecmd[\DUAL][\Theta]{\Gamma}{\alpha \oftype A,\Delta}{c}}
  \qqqquad
  \infer[\namespace{\FromNegL}]
  {
    \typecotmfoc[\DUAL][\Theta]{\Gamma}{\Delta}
    {\force[\strat{S}] F}{\FromNeg[\strat{S}] A}
  }
  {\typecotmfoc[\DUAL][\Theta]{\Gamma}{\Delta}{F}{A}}
  \\[\ruleskip]
  \infer[\namespace{\ToNegR}]
  {
    \typetmfoc[\DUAL][\Theta]{\Gamma}{\Delta}
    {\fn{\unwrap[\strat{S}]\alpha}c}
    {\ToNeg[\strat{S}] A}
  }
  {\typecmd[\DUAL][\Theta]{\Gamma}{\alpha \oftype A,\Delta}{c}}
  \qqqquad
  \infer[\namespace{\ToNegL}]
  {
    \typecotmfoc[\DUAL][\Theta]{\Gamma}{\Delta}
    {\unwrap[\strat{S}] E_{\strat{S}}}{\ToNeg[\strat{S}] A}
  }
  {\typecotmfoc[\DUAL][\Theta]{\Gamma}{\Delta}{E_{\strat{S}}}{A}}
  \\[\ruleskip]
  \infer[\namespace{\FromPosR}]
  {
    \typetmfoc[\DUAL][\Theta]{\Gamma}{\Delta}
    {\delay[\strat{S}] W}{\FromPos[\strat{S}] A}
  }
  {\typetmfoc[\DUAL][\Theta]{\Gamma}{\Delta}{W}{A}}
  \qqqquad
  \infer[\namespace{\FromPosL}]
  {
    \typecotmfoc[\DUAL][\Theta]{\Gamma}{\Delta}
    {\sn{\delay[\strat{S}] x}c}
    {\FromPos[\strat{S}] A}
  }
  {\typecmd[\DUAL][\Theta]{\Gamma, x \oftype A}{\Delta}{c}}
\end{gathered}
\end{math}
\caption{Typing rules of the System $\DUAL$ polarity shift connectives.}
\label{fig:dual-core-shift-types}
\end{figure}

Finally, we have the typing rules for specific connectives: positive types are
given in \cref{fig:dual-core-positive-types}, negative types in
\cref{fig:dual-core-negative-types}, and the shifts in
\cref{fig:dual-core-shift-types}.  The typing rules for positive and negative
types are standard for a sequent-calculus presentation like System L
\cite{MunchMaccagnoni2013PhD}, and the typing rules for shifts are the same as
in \cref{sec:intro-polarity}.  Note how the more systematic approach to focusing
shows the difference between the two different styles of shifts, the data shifts
$\ToPos[\strat{S}]$ and $\FromPos[\strat{S}]$ both lose focus on the left
(because their premise is a command), whereas the dual {\co}data shifts
$\ToNeg[\strat{S}]$ and $\FromNeg[\strat{S}]$ both lose focus on the right.

\subsection{Equational theory}
\label{sec:dual-core-equation}

In order to reason about the behavior of expressions, we give an equational
theory of the $\DUAL$ calculus in \cref{fig:dual-core-equality}.  Notice that it
has the same $\betamu$ and $\betatmu$ rules as $\lmtmP$ from
\cref{sec:intro-polarity}: the discipline sign $\strat{S}$ in the command is
used to decide which notion of value or {\co}value is allowed.  To make sure
that this discipline is being respected during a substitution, we use the
restricted form $\subst{\anns{x}{A}{\strat{S}}}{V_{\strat{S}}}$ and
$\subst{\anns{\alpha}{A}{\strat{S}}}{E_{\strat{S}}}$ such that the syntactic sign
$\strat{S}$ is matched by the syntactic category $V_{\strat{S}}$ or
$E_{\strat{S}}$.  In other words, each discipline has its own notion of valid
substitutions, and the discipline of a bound variable tells us what values or
{\co}values might replace it.  In contrast, the type annotating a command has no
impact on the reductions it may take, and so they can be ignored.

\begin{figure}
\centering
Disciplined substitutions ($\rho$):
\begin{gather*}
  \rho
  ::=
  \many{\asub{\ann{X}{k}}{A}},~
  \many{\asub{\anns{x}{B}{\strat{T}}}{V_{\strat{T}}}},~
  \many{\asub{\anns{\alpha}{C}{\strat{R}}}{E_{\strat{R}}}}
\end{gather*}

Reduction rules:
\begin{gather*}
\begin{aligned}
  \rewriterule{\betamu}
  {\!\!
    \cut{\outp{\alpha}c}[\ann{A}{\strat{S}}]{E_{\strat{S}}}
    &\red
    \rlap{$c\subst{\anns{\alpha}{A}{\strat{S}}}{E_{\strat{S}}}$}
    &&
  }
  &
  \rewriterule{\betatmu}
  {\!\!
    \cut{V_{\strat{S}}}[\ann{A}{\strat{S}}]{\inp{x}c}
    &\red
    \rlap{$c\subst{\anns{x}{A}{\strat{S}}}{V_{\strat{S}}}$}
    &&
  }
  \\
  \rewriterule{\etamu}
  {\!\!
    \outp{\alpha}\cut{v}[\ann{A}{\strat{S}}]{\alpha}
    &\red
    v
    &&\!\!
    {\scriptstyle(\alpha \notin \FV(v))}
  }
  &
  \rewriterule{\etatmu}
  {\!\!
    \inp{x}\cut{x}[\ann{A}{\strat{S}}]{e}
    &\red
    e
    &&\!\!
    {\scriptstyle(x \notin \FV(e))}
  }
  \\
  \rewriterule{\betar[q]}
  {\!\!
    \cut
    {{\cocaseof{\many[i]{\send{q_i}{c_i}}}}}
    [\ann{A}{\strat{S}}]
    {q\subs{\rho}}
    &\red
    c_i\subs{\rho}
    &&\!\!
    {\scriptstyle(q = q_i)}
  }
  &
  \rewriterule{\betar[p]}
  {\!\!
    \cut
    {p\subs{\rho}}
    [\ann{A}{\strat{S}}]
    {{\caseof{\many[i]{\recv{p_i}{c_i}}}}}
    &\red
    c_i\subs{\rho}
    &&\!\!
    {\scriptstyle(p = p_i)}
  }
  \\
  \rewriterule{\etar[q]}
  {\!\!
    \cocaseof{\many{\send{q}{\cut{x}[\ann{A}{\strat{S}}]{q}}}}
    &\red
    x
    &&\!\!
    {\scriptstyle(x \notin \FV(q))}
  }
  &
  \rewriterule{\etar[p]}
  {\!\!
    \caseof{\many{\recv{p}{\cut{p}[\ann{A}{\strat{S}}]{\alpha}}}}
    &\red
    \alpha
    &&\!\!
    {\scriptstyle(\alpha \notin \FV(p))}
  }
  % \\
  % \rewriterule{\gcr[\CBNeed]}
  % {
  %   \cut{v}[\ann A \CBNeed]{\inp x c} &\red c
  %   &&
  %   {\scriptstyle(x \notin \FV(c))}
  % }
  % &
  % \rewriterule{\gcr[\CoNeed]}
  % {
  %   \cut{\outp \alpha c}[\ann A \CoNeed]{e} &\red c
  %   &&
  %   {\scriptstyle(\alpha \notin \FV(c))}
  % }
\end{aligned}
\\[1ex]
\begin{aligned}
  \rewriterule{\swapr[\CBNeed]}
  {
    \cut
    {\outp\alpha \cut{v}[\ann B \CBNeed]{\inp y c}}
    [\ann A \CBNeed]
    {e}
    &\red
    \cut
    {v}
    [\ann B \CBNeed]
    {\inp y \cut{\outp\alpha c}[\ann A \CBNeed]{e}}
    &\qquad&
    {\scriptstyle(x \notin \FV(e), \alpha \notin \FV(v))}
  }
  \\
  \rewriterule{\swapr[\CoNeed]}
  {
    \cut
    {v}
    [\ann A \CoNeed]
    {\inp x \cut{\outp\beta c}[\ann B \CoNeed]{e}}
    &\red
    \cut
    {\outp\beta\cut{v}[\ann A \CoNeed]{\inp x c}}
    [\ann B \CoNeed]
    {e}
    &\qquad&
    {\scriptstyle(x \notin \FV(e), \alpha \notin \FV(v))}
  }
\end{aligned}
\end{gather*}

Typed equality relation:
\begin{gather*}
\begin{gathered}
  \infer
  {\typecmd[\DUAL][\Theta]{\Gamma}{\Delta}{c = c'}}
  {
    \typecmd[\DUAL][\Theta]{\Gamma}{\Delta}{c}
    &
    c \red c'
  }
  \qqqquad
  \infer
  {\typetm[\DUAL][\Theta]{\Gamma}{\Delta}{v = v'}{A}}
  {
    \typetm[\DUAL][\Theta]{\Gamma}{\Delta}{v}{A}
    % &
    % \typety[\DUAL]{\Theta}{A}{\strat{S}}
    &
    v \red v'
  }
  \qqqquad
  \infer
  {\typecotm[\DUAL][\Theta]{\Gamma}{\Delta}{e = e'}{A}}
  {
    \typetm[\DUAL][\Theta]{\Gamma}{\Delta}{e}{A}
    % &
    % \typety[\DUAL]{\Theta}{A}{\strat{S}}
    &
    e \red e'
  }
\end{gathered}
\end{gather*}
plus inference rules for compatibility, reflexivity, symmetry, transitivity of equality
\caption{Equational theory for System $\DUAL$.}
\label{fig:dual-core-equality}
\end{figure}

The various $\beta$ rules for specific data and {\co}data types (like sum types
and shifts) are all summarized by the general $\betar[p]$ and $\betar[q]$ rules,
which perform a pattern or {\co}pattern match.  The $\betar[p]$ rule says that
if a constructed value can be decomposed as a pattern $p$ and substitution
$\rho$, then the matching branch of the {\co}value can be taken.  For example,
we would have the following instance of the rule for tuple patterns (ignoring
the type annotation):
\begin{align*}
  \cut{(W_1,W_2)}[\CBV]{\sn{(x,y)}c}
  &=
  \cut
  {(x,y)\subs{\asub{\anns{x}{A}{\CBV}}{W_1},\asub{\anns{y}{A}{\CBV}}{W_2}}}
  [\CBV]
  {\sn{(x,y)}c}
  \red[{\betar[p]}]
  c\subs{\asub{\anns{x}{A}{\CBV}}{W_1},\asub{\anns{y}{A}{\CBV}}{W_2}}
\end{align*}
The $\betar[q]$ rule {\co}pattern matches against a destruction, like in the
following specific instance of the rule for product {\co}patterns:
\begin{align*}
  \cut
  {\cocaseof{\send{\fst\alpha}{c_1}\mid\send{\snd\beta}{c_2}}}
  [\CBN]
  {\snd F}
  &=
  \cut
  {\cocaseof{\send{\fst\alpha}{c_1}\mid\send{\snd\beta}{c_2}}}
  [\CBN]
  {(\snd\beta)\subst{\anns{\beta}{A}{\CBN}}{F}}
  \red[{\betar[q]}]
  c_2\subst{\anns{\beta}{A}{\CBN}}{F}
\end{align*}

We also have reductions for recognizing extensionality, which further brings out
the duality between data and {\co}data.  The $\etamu$ and $\etatmu$ rules are
the counterpart of $\betamu$ and $\betatmu$, and eliminate a redundant $\mu$- or
$\tmu$-abstraction that introduces a (\co)variable just to immediately use it
exactly once, before forgetting it.  Likewise, the $\etar[q]$ and $\etar[p]$
rules perform the same kind of reduction but for redundant (\co)pattern matching
instead of a plain (\co)variable.  Note that the important difference between
$\etamu$ and $\etar[q]$ is that $\etamu$ applies to any underlying term $v$,
whereas $\etar[q]$ does not.  This is because the more general case is not sound
for a {\co}data type like $\FromNeg[\CBV] A$: there can be a difference between
a term $v$ which \emph{computes} a thunk and
$\cocaseof{\send{\Force[\CBV]\alpha}{\cut{v}[\CBV]{\Force[\CBV]\alpha}}}$ which
\emph{is} a thunk.  For example, consider the term
$\outp{\alpha}c_0 : \FromNeg[\CBV]A$ which throws away its observer because
$\alpha$ does not appear in $c_0$.  This term's $\etar[q]$-expansion could give
a different result when run with the {\co}term $\inp z c_1$ which throws away
its input because $z$ does not appear in $c_1$, like so:
\begin{align*}
  \cut{\outp{\alpha}c_0}
  [\CBV]%[{\FromNeg[\CBV]A}{:}\CBV]
  {\inp z c_1}
  &\sred[{\betamu[\CBV]}]
  c_0
  \\
  \cut
  {
    \cocaseof
    {
      \send
      {\Force[\CBV]\alpha'}
      {
        \cut
        {\outp{\alpha}c_0}
        [\CBV]%[{\FromNeg[\CBV]A}{:}\CBV]
        {\Force[\CBV]\alpha'}
      }
    }
  }
  [\CBV]%[{\FromNeg[\CBV]A}{:}\CBV]
  {\inp z c_1}
  &\sred[{\betamu[\CBV]}]
  c_1
\end{align*}
Therefore, the $\etar[q]$ rule is restricted to only apply to a variable $x$,
which semantically stands for any value of the appropriate discipline, thereby
avoiding the above counterexample.  But for negative {\co}data types, like
products, we can derive the full $\eta$ law as follows
\cite{DownenAriola2014DC}:
\begin{align*}
  v
  &\unred[{\etamu}]
  \outp\alpha \cut{v}[\ann{A{\with}B}\CBN]{\alpha}
  \\
  &\unred[{\betamu[\CBN]}]
  \outp\alpha
  \cut{v}[\ann{A{\with}B}\CBN]{\inp x \cut{x}[\CBN]{\alpha}}
  \\
  &\unred[{\etar[q]}]
  \outp\alpha
  \cut{v}
  [\ann{A{\with}B}\CBN]
  {
    \inp x
    \cut
    {
      \cocaseof
      {
        \send{\fst{\beta_1}}{\cut{x}[\CBN]{\fst{\beta_1}}}
        \mid
        \send{\snd{\beta_2}}{\cut{x}[\CBN]{\fst{\beta_2}}}
      }
    }
    [\CBN]
    {\alpha}
  }
  \\
  &\red[{\betamu[\CBN]}]
  \outp\alpha
  \cut
  {
    \cocaseof
    {
      \send{\fst{\beta_1}}{\cut{v}[\CBN]{\fst{\beta_1}}}
      \mid
      \send{\snd{\beta_2}}{\cut{v}[\CBN]{\fst{\beta_2}}}
    }
  }
  [\CBN]
  {\alpha}
  \\
  &\red[{\etamu}]
  \cocaseof
  {
    \send{\fst{\beta_1}}{\cut{v}[\CBN]{\fst{\beta_1}}}
    \mid
    \send{\snd{\beta_2}}{\cut{v}[\CBN]{\fst{\beta_2}}}
  }
\end{align*}
Dually, $\etatmu$ applies to any underlying {\co}term but $\etar[p]$ does not
because it would not be sound for a data type like $\FromPos[\CBN]A$: there can
be a difference between a {\co}term $e$ that \emph{computes} a strict
observation and
$\caseof{\recv{\Delay[\CBN]x}{\cut{\Delay[\CBN]x}[{\FromPos[\CBN]A}{:}\CBN]{e}}}$
which \emph{is} a strict observation.  For the same reason as before, the
$\etar[p]$ rule is restricted to {\co}variables, which still derives the full
$\eta$ law for positive data types, like the sum types in
\cref{sec:intro-polarity}.

Finally, we have the re-association rules $\swap$, which swap the bindings of a
variable and {\co}variable both of kind $\CBNeed$ or $\CoNeed$.  Observe how on
both the left- and right-hand sides of $\swapr[\CBNeed]$ rule, a call-by-need
evaluation order will first evaluate the {\co}term $e$, and if $e$ reduces to
{\co}value $E_\CBNeed$, then the two sides step to the same command
$\cut{v}[B{:}\CBNeed]{\inp y c\subst{\anns\alpha{}\CBNeed}{E_\CBNeed}}$ via a
$\betamu[\CBNeed]$ reduction.  Dually, in both the left- and right-hand sides of
the $\swapr[\CoNeed]$ rule, a call-by-{\co}need evaluation order will first
evaluate $v$, and if $v$ reduces to a value then both sides step to the same
command via a $\betatmu[\CoNeed]$ reduction.  Note that instances of these rules
for a pair of call-by-value ($\CBV$) or call-by-name ($\CBN$) bindings are also
sound, and can be derived from the $\betamu[\CBV]$ and $\betatmu[\CBN]$ rules,
respectively, like so:
\begin{align*}
  \cut{\outp\alpha\cut{v}[\CBV]{\inp y c}}[\CBV]{e}
  &\red[{\betamu[\CBV]}]
  \cut{v}[\CBV]{\inp y c\subst{\anns{\alpha}{}{\CBV}}{e}}
  \unred[{\betamu[\CBV]}]
  \cut{v}[\CBV]{\inp y \cut{\outp\alpha c}[\CBV]{e}}
  \\
  \cut{v}[\CBN]{\inp x \cut{\outp\beta c}[\CBN]{e}}
  &\red[{\beta\CBN}]
  \cut{\outp\beta c\subst{\anns{x}{}{\CBV}}{v}}[\CBN]{e}
  \unred[{\beta\CBN}]
  \cut{\outp\beta \cut{v}[\CBN]{\inp x c}}[\CBN]{e}
\end{align*}
However, trying to reassociate arbitrary variable and {\co}variable bindings in
this way is not sound \cite{MunchMaccagnoni2013PhD}; in particular,
reassociating a call-by-value {\co}term and a call-by-name term binding can
break confluence, as in (where $\delta$ and $z$ are unused):
\begin{align*}
  c_0
  \unred[{\betamu[\CBV]}]
  \cut{\outp\delta c_0}[\CBV]{\inp x \cut{\outp\alpha c}[\CBN]{\inp z c_1}}
  &\neq_\swap
  \cut{\outp\alpha \cut{\outp\delta c_0}[\CBV]{\inp x c}}[\CBN]{\inp z c_1}
  \red[{\betatmu[\CBN]}]
  c_1
\end{align*}

\subsection{Operational semantics}
\label{sec:dual-core-operation}

To give an operational semantics to System $\DUAL$, we only need
to pick a standard reduction from the rules of \cref{fig:dual-core-equality}.
The goal is to define a standard reduction relation on commands, $c \sred c'$,
with the following properties:
\begin{property}[Standard reduction]
\label{thm:sred-properties}
\

\begin{lemmaenum}
\item \emph{Determinism}: If $c_1 \unsred c \sred c_2$ then $c_1$ and $c_2$ must
  be the same command.
\item \emph{Compatibility with heaps}: If $c \sred c'$ then $H[c] \sred H[c']$
  for any heap context $H$.
\item \emph{Closure under substitution}: If $c \sred c'$ then
  $c\subs{\rho} \sred c'\subs{\rho}$ for any well-disciplined substitution
  $\rho$ matching the free (\co)variables of $c$.
\end{lemmaenum}
\end{property}
As a first cut, we can just keep the essential reduction rules for computation
(the $\beta$ rules) and exclude the additional reductions that are only used for
reasoning about program equivalence (the $\eta$ and $\swap$ rules).  Allowing
reduction only on the top-level of a command is not enough, due to the delayed
bindings caused by call-by-(\co)need evaluation (as in \cref{sec:intro-need}),
so we should allow reduction inside of heaps.

However, this immediately causes a problem with determinism.  First of all,
reduction inside heaps opens up the choice between substituting an unneeded
binding early or not.  For example, we could have two possible reductions in the
command
\begin{align*}
  \cut{V_\CBNeed}[\CBNeed]{\inp x \cut{(W_1,W_2)}[\CBV]{\sn{(y_1,y_2)}c}}
  &\red[{\betatmu[\CBNeed]}]
  \cut{(W_1,W_2)}[\CBV]{\sn{(x,y)}c}\subst{\anns{x}{}{\CBNeed}}{V_\CBNeed}
  \\
  \cut{V_\CBNeed}[\CBNeed]{\inp x \cut{(W_1,W_2)}[\CBV]{\sn{(y_1,y_2)}c}}
  &\red[{\betar[p]}]
  \cut
  {V_\CBNeed}
  [\CBNeed]
  {\inp x c\subs{\asub{\anns{y_1}{}{\CBV}}{W_1},\asub{\anns{y_2}{}{\CBV}}{W_2}}}
\end{align*}
To avoid this ambiguity, we will only fire a $\betatmu[\CBNeed]$ reduction when
the $\tmu$-abstraction cannot take a step itself, \ie it is a {\co}value.
Secondly, forcing a chain of variable bindings causes another ambiguity for
which variable should be substituted first.  Consider the command
\begin{align*}
  \cut{x}[\CBNeed]{\inp y \cut{y}[\CBNeed]{\inp z \cut{z}[\CBNeed]{\alpha}}}
\end{align*}
where $z$ is forced by $\alpha$, which forces evaluation of $y$, and again
forces evaluation of $x$.  Because of this chain of demand, each of $\alpha$,
$\inp z \cut{z}[\CBNeed]{\alpha}$, \emph{and}
$\inp y \cut{y}[\CBNeed]{\inp z \cut{z}[\CBNeed]{\alpha}}$ are all call-by-need
{\co}values.  So which variable substitution should occur first?  If we say the
outermost one, then standard reduction is not compatible with heaps (that could
add another outer-most binding).  If we say the innermost one, then standard
reduction is not closed under substitution (that could replace $\alpha$ and add
another inner-most binding).  But at the end of the day this decision doesn't
matter; we still need to know the value of the last variable $x$ to make
meaningful progress.  So instead we rule out this scenario altogether:
$\betatmu[\CBNeed]$ reduction will not fire for variables, and
$\betamu[\CoNeed]$ reduction will not fire for {\co}variables.

\begin{figure}
\centering
Standard reduction rules:
\begin{align*}
  \rewriterule{\betamu[\NotCoNeed]}
  {
    \cut{\outp{\alpha}c}[\ann{A}{\strat{S}}]{E_{\strat{S}}}
    &\sred
    c\subst{\anns{\alpha}{A}{\strat{S}}}{E_{\strat{S}}}
    &
    (\strat{S} &\in \{\CBV,\CBN,\CBNeed\})
  }
  \\
  \rewriterule{\betatmu[\NotCBNeed]}
  {
    \cut{V_{\strat{S}}}[\ann{A}{\strat{S}}]{\inp{x}c}
    &\sred
    c\subst{\anns{x}{A}{\strat{S}}}{V_{\strat{S}}}
    &
    (\strat{S} &\in \{\CBV,\CBN,\CoNeed\})
  }
  \\
  \rewriterule{\betamuconeed}
  {
    % \cut
    % {\outp{\alpha}H[\cut{V_\CoNeed}[\ann{A'}{\CoNeed}]{\alpha}]}
    % [\ann A \CoNeed]
    % {E_\CoNeed}
    % &\sred
    % H[\cut{V_\CoNeed}[\ann{A'}{\CoNeed}]{\alpha}]
    % \subst{\anns{\alpha}{A}{\CoNeed}}{E_\CoNeed}
    % &
    % (E_\CoNeed &\neq \beta)
    \cut{\outp{\alpha}c}[\ann{A}{\CoNeed}]{E_\CoNeed}
    &\sred
    c\subst{\anns{\alpha}{A}{\CoNeed}}{E_\CoNeed}
    &
    (\outp{\alpha}c &\in V_\CoNeed, ~ E_\CoNeed \neq \beta)
  }
  \\
  \rewriterule{\betatmuneed}
  {
    % \cut
    % {V_\CBNeed}
    % [\ann A \CBNeed]
    % {\inp{x}H[\cut{x}[\ann{A'}{\CBNeed}]{E_\CBNeed}]}
    % &\sred
    % H[\cut{x}[\ann{A'}{\CBNeed}]{E_\CBNeed}]
    % \subst{\anns{x}{A}{\CBNeed}}{V_\CBNeed}
    % &
    % (V_\CBNeed &\neq y)
    \cut{V_\CBNeed}[\ann{A}{\CBNeed}]{\inp x c}
    &\sred
    c\subst{\anns{x}{A}{\CBNeed}}{V_\CBNeed}
    &
    (\inp x c &\in E_\CBNeed, V_\CBNeed \neq y)
  }
  \\
  \rewriterule{\betar[p]}
  {
    \cut
    {p\subs\rho}
    [A{:}\strat{S}]
    {\caseof{\many[i]{\recv{p_i}{c_i}}}}
    &\sred
    c_i\subs\rho
    &
    (p &\eq p_i)
  }
  \\
  \rewriterule{\betar[q]}
  {
    \cut
    {\cocaseof{\many[i]{\send{q_i}{c_i}}}}
    [A{:}\strat{S}]
    {q\subs\rho}
    &\sred
    c_i\subs\rho
    &
    (q &\eq q_i)
  }
  % \\
  % \rewriterule{\unroll}
  % {
  %   \fix{\ann x A}v
  %   &\sred
  %   v\subst{\anns x A \CBN}{\fix{\ann x A}v}
  % }
  % \\
  % \rewriterule{\counroll}
  % {
  %   \cofix{\ann \alpha A}e
  %   &\sred
  %   e\subst{\anns \alpha A \CBV}{\cofix{\ann \alpha A}e}
  % }
\end{align*}

Reduction inside heap contexts:
\begin{gather*}
  \infer
  {H[c] \sred H[c']}
  {c \sred c'}
\end{gather*}
\caption{Operational semantics of System $\DUAL$.}
\label{fig:dual-core-operation}
\end{figure}

With these observations in mind, the operational semantics of $\DUAL$ is given
in \cref{fig:dual-core-operation}.  The standard reduction relation can reduce
any command inside of a heap context and includes the (\co)pattern-matching
rules $\betar[p]$ and $\betar[q]$.  For reducing $\mu$- and $\tmu$-abstractions,
we use a special restricted rule $\betatmuneed$ for call-by-need $\tmu$s and
$\betamuconeed$ for call-by-{\co}need $\mu$s.  In all other cases, we just use
the unaltered $\betamu$ and $\betatmu$ rules from the equational theory in
\cref{fig:dual-core-equality}.  Notice that this definition of standard
reduction satisfies each of the properties in \cref{thm:sred-properties}.  It is
compatible with heaps by definition.  And both determinism and closure under
substitution follow by induction on the syntax of commands (using the fact that
$c$ has no standard reduction if $\inp x c$ is a $\CBNeed$-{\co}value, or dually
if $\outp \alpha c$ is a $\CoNeed$-value).

So now that we know how to run commands, how do we know when they're done?
Since we will need to be able to run \emph{open} commands which may have free
variables and {\co}variables due to (\co)call-by-need evaluation, the idea of
\emph{needed (\co)variables} will become the cornerstone of a command's status
by expressing when the value of a variable or {\co}variable needs to be supplied
for computation to continue.  The key intuition about needed (\co)variables, as
defined above, is that
\begin{enumerate*}[(1)]
\item they are only needed in commands that can no longer take a standard
  reduction step, and
\item they are always found in the eye of a heap context $H$ which does not bind
  them.
\end{enumerate*}
\begin{definition}[Need]
  \label{def:neededvar}
  The set of \emph{needed} (\co)variables of a command, written $\neededvar(c)$,
  is the smallest set such that:
  \begin{lemmaenum}
  \item $x \in \neededvar(c)$ if and only if $c \not\sred$ and
    $c = H[\cut{x}[\ann{A}{\strat{S}}]{E_{\strat{S}}}]$ such that $x$ is not
    bound by $H$.
  \item $\alpha \in \neededvar(c)$ if and only if $c \not\sred$ and
    $c = H[\cut{V_{\strat{S}}}[\ann{A}{\strat{S}}]{\alpha}]$ such that $\alpha$
    is not bound by $H$.
  \end{lemmaenum}
\end{definition}
The status of a command is then defined by its ability to take another step, or
else which (\co)variables it needs to continue stepping.  This lets us formulate
a notion of type safety that is both perfectly symmetric (with respect to free
inputs and outputs) and open enough to capture call-by-{\co}need evaluation.  In
particular, a command is finished when supplying the (\co)value for a
(\co)variable might spur it on to continue computation, but is stuck when no
such substitution can ever restart computation.  For example, the command
$\cut{\inl{x}}[\CBV]{\cocaseof{}}$ is forever stuck.
\begin{definition}[Status]
  \label{def:end-states}
  A command is \emph{finished} when $\neededvar(c)$ is non-empty, and
  \emph{stuck} when $c \not\sred$ and $\neededvar(c)$ is empty.
\end{definition}
\begin{restatable}[Type safety]{theorem}{thmtypesafety}
\label{thm:type-safety}
\begin{lemmaenum}
\item \emph{Progress}: If $\typecmd[\DUAL][\Theta]{\Gamma}{\Delta}{c}$, then $c$
  is not stuck.
\item \emph{Preservation}: If $\typecmd[\DUAL][\Theta]{\Gamma}{\Delta}{c}$ and
  $c \sred c'$ then $\typecmd[\DUAL][\Theta]{\Gamma}{\Delta}{c'}$.
\end{lemmaenum}
\end{restatable}

\section{Enriching the Core: System \texorpdfstring{$\XDUAL$}{CD}}
\label{sec:dual-ext}

Our first stated goal of the dual core calculus $\DUAL$ is to be able to encode
a wide range of types that a programmer might define and use in a programming
language.  Evaluating if this goal has been met is more subtle than just
checking that enough well-typed programs can be written.  For example, it is
well-known that unary function types ($A \to B$) can encode binary and
higher-arity functions through currying, like $A \to (B \to C)$.  However, in
the call-by-value $\lambda$-calculus, the curried function type
$A \to (B \to C)$ is not the same thing as the binary function type
$(A \otimes B) \to C$: there is no function of type $(A \otimes B) \to C$ that
corresponds to the function $\fn x \Omega : A \to (B \to C)$ (where $\Omega$ is
a term that loops forever) that diverges after being partially applied to one
argument.  Similarly, nesting sum types like $A \oplus (B \oplus C)$ as a way to
encode a ternary choice has the analogous issue in the call-by-name
$\lambda$-calculus: there is no term of a ternary choice type corresponding to
$\inr{\Omega}$.  Instead, we should ask for more than just well-typed programs:
the encodings should be between types that are truly isomorphic to one another.

So how can we properly evaluate what class of types are representable by the
$\DUAL$ calculus?  We will define an extension of $\DUAL$, called
$\XDUAL$,\footnote{So named because it is a calculus of both $\D$ata and
  $\C$o$\D$ata.} that is a higher-level language that is rich with many types,
meeting the following objectives:
\begin{itemize}
\item Formalizing a fully-dual mechanism for programmatically declaring new data
  and {\co}data types.  To ensure generality, there will be no restrictions on
  the number of constructors or destructors, inputs or outputs, and disciplines
  involved.  Declared (\co)data types will also be able to introduce publically
  and privately quantified types, to give a general form of type abstraction.
\item Removing the restrictions on constructor and destructor applications.  In
  the extended $\DUAL$ calculus, they can be applied to any terms and
  {\co}terms, not just values and {\co}values.
\end{itemize}
On the one hand, the extension will clearly subsume $\DUAL$ (that is, it is a
conservative extension), so that nothing is lost.  But on the other hand, we
will show that the extended language $\XDUAL$ can still be faithfully encoded
back into the $\DUAL$ core calculus via macro expansions, so that nothing of
importance is gained.  This way, the $\XDUAL$ calculus reveals the expressive
power of the $\DUAL$ core calculus.

\subsection{Extending \texorpdfstring{$\DUAL$}{D} with fully dual (\co)data types}

The syntax for the $\XDUAL$ calculus is given in \cref{fig:dual-ext-syntax}.
Notice how, without the restrictions of the core calculus from
\cref{fig:dual-core-syntax}, the syntax is much simpler.  Instead of the
specific patterns and {\co}patterns, we just have a general class of
constructors (denoted by $\mk{K}$) and destructors (denoted by $\mk{O}$) names
to choose from, which can be applied to any number of types, terms, and
{\co}terms.  This unifies all the different cases into just one case for a data
construction and one case for a {\co}data destruction.  Likewise at the
type-level, instead of the list of specific connectives, we just have a class of
connective names (denoted by $\mk{F}$ or $\mk{G}$).  These connective names are
given meaning by a data or {\co}data declaration, which must be non-recursive.%
\footnote{The requirement that {\co}data declarations be non-recursive is so
  that we will be able to fully encode them via macro expansions, which only
  terminates for non-recursive declarations.}

\begin{figure}
\centering

Commands $c$, terms $v$, {\co}terms $e$, patterns $p$, and {\co}patterns $q$
\begin{gather*}
\begin{aligned}
  c &::= \cut{v}[A{:}\strat{S}]{e}
  \\
  v
  &::= x
  \Alt \outp\alpha c
  \Alt \mk{K}\many{B}\many{e}\many{v}
  \Alt \cocaseof{\many{\send{q}{c} \mid}}
  &\qqqquad
  p
  &::= \mk{K} \many{X} \many{\alpha} \many{y}
  \\
  e
  &::= \alpha
  \Alt \inp x c
  \Alt \mk{O}\many{B}\many{v}\many{e}
  \Alt \caseof{\many{\recv{p}{c} \mid}}
  &\qqqquad
  q
  &::= \mk{O} \many{X} \many{y} \many{\alpha}
\end{aligned}
\end{gather*}

Types $A$, kinds $k$, disciplines $\strat{S}$, and (\co)data declarations
$\<decl>$
\begin{gather*}
\begin{aligned}
  A,B,C
  &::= X
  \Alt \mk{F}
  \Alt \fn{X{:}k}A
  \Alt A ~ B
  &\qquad
  k, l
  &::= \strat{S}
  \Alt k \to l
  &\qquad
  \strat{R}, \strat{S}, \strat{T}
  &::= \CBV
  \Alt \CBN
  \Alt \CBNeed
  \Alt \CoNeed
  % &\qqqquad
  % \annvar X, \annvar Y, \annvar Z
  % &::= \ann{X}{k}
\end{aligned}
\\
\begin{aligned}
  \<decl>
  &::=
  \begin{inlinedata}{\mk{F}\many{(X{:}k)} : \strat{S}}
    \many{
      \smalltypecon[\many{Y{:}\strat{S}'}]
      {\many{A{:}\strat{T}}}{\many{B{:}\strat{R}}}
      {\mk{K}}{\mk{F}\many{X}}{\strat{S}}
    }
  \end{inlinedata}
  \\
  &\phantom{:=}
  \Alt
  \begin{inlinecodata}{\mk{G}\many{(X{:}k)} : \strat{S}}
    \many
    {
      \smalltypeobs[\many{Y{:}\strat{S}'}]
      {\many{A{:}\strat{T}}}{\many{B{:}\strat{R}}}
      {\mk{O}}{\mk{G}\many{X}}{\strat{S}}
    }
  \end{inlinecodata}
\end{aligned}
\end{gather*}
\caption{Syntax of System $\XDUAL$: generalizing System $\DUAL$ with fully dual
  (\co)data.}
\label{fig:dual-ext-syntax}
\end{figure}

In general, a data type declaration is similar to an algebraic data type (ADT)
declaration from a functional language, and has the form:
\begin{align*}
  \begin{data}{\mk{F} \many{(\ann{X}{k})} : \strat{S}}
    &\smalltypecon[\many{Y_1:\strat{S}_1}]
    {\many{A_1:\strat{T}_1}}{\many{B_1:\strat{R}_1}}
    {\mk{K}_1}{\mk{F}\many{X}}{\strat{S}}
    \\
    &\dots
    \\
    &\smalltypecon[\many{Y_n:\strat{S}_n}]
    {\many{A_n:\strat{T}_n}}{\many{B_n:\strat{R}_n}}
    {\mk{K}_n}{\mk{F}\many{X}}{\strat{S}}
  \end{data}
\end{align*}
This declaration says that the connective $\mk{F}$ is parameterized by several
types (named $X$, each of kind $k$, respectively) and returns a type of the
discipline $\strat{S}$.  There are $n$ different constructors
($\mk{K}_1\dots\mk{K}_n$) for building values of $\mk{F}$ types.  The way to
read the signature of each constructor is that ``inputs'' are given on the left
of the turnstyle ($\entails$) and outputs are given on the right of the
turnstyle.  The main way this data declaration differs from a functional ADT is
the availability of several different output types besides the main return type
of the constructor.  For example, the first constructor $\mk{K}_1$ is
parameterized by several terms of types $\many{A_1:\strat{T}_1}$ (found on the
left) and several {\co}terms of types $\many{B_1:\strat{R}_1}$ (found on the
right), respectively, and will return a construction of type $\mk{F}\many{X}$
(the main output on the right between $\entails$ and $\act$).  The constructor
$\mk{K}_1$ is also parameterized by several types $\many{Y_1}$ of kind
$\many{\strat{S}_1}\,$.  These types may be referred to in the types of the term
and {\co}term arguments, but are hidden (\ie existentially quantified) to the
consumer of the construction.

{\Co}data type declarations are exactly symmetric to data type declarations, and
have a very similar form:
\begin{align*}
  \begin{codata}{\mk{G} \many{(\ann{X}{k})} : \strat{S}}
    &\smalltypeobs[\many{Y_1:\strat{S}_1}]
    {\many{A_1:\strat{T}_1}}{\many{B_1:\strat{R}_1}}
    {\mk{O}_1}{\mk{G}\many{X}}{\strat{S}}
    \\
    &\dots
    \\
    &\smalltypeobs[\many{Y_n:\strat{S}_n}]
    {\many{A_n:\strat{T}_n}}{\many{B_n:\strat{R}_n}}
    {\mk{O}_n}{\mk{G}\many{X}}{\strat{S}}
  \end{codata}
\end{align*}
As before, this declaration says that $\mk{G}$ is parameterized by several types
and returns a type of the discipline $\strat{S}$.  But now, there are $n$
different destructors ($\mk{O}_1\dots\mk{O}_n$) for building observations that
use values of $\mk{G}$ types.  The flow of information is still inputs on the
left and outputs on the right, as before.  For example, the first destructor
$\mk{O}_1$ is parameterized by several terms of types $\many{A_1:\strat{T}_1}$
(found on the left) and several {\co}terms of types $\many{B_1:\strat{R}_1}$
(found on the right), respectively, and will consume a value of type
$\mk{G}\many{X}$ (the main input on the left between $\act$ and $\entails$).
The type parameters $\many{Y_1}$ to the observer $\mk{O}_1$ are still brought
into scope for (\co)term parameters of $\mk{O}_1$, but are now hidden (\ie
universally quantified) to the value being consumed.

\begin{figure}
\centering
\newcommand{\declskip}{1ex}
Simple (\co)data types
\\
\begin{math}
\begin{aligned}
  &
  \begin{data}[t]{(X{:}\CBV) \oplus (Y{:}\CBV) : \CBV}
    \typecon{X{:}\CBV}{}{\Inl}{X \oplus Y}{\CBV}
    \\
    \typecon{Y{:}\CBV}{}{\Inr}{X \oplus Y}{\CBV}
  \end{data}
  &\qqqquad
  \begin{codata}[t]{(X{:}\CBN) \with (Y{:}\CBN) : \CBN}
    \typeobs{}{X{:}\CBN}{\Fst}{X \with Y}{\CBN}
    \\
    \typeobs{}{Y{:}\CBN}{\Snd}{X \with Y}{\CBN}
  \end{codata}
  \\[\declskip]
  &
  \begin{inlinedata}{0 : \CBV}
  \end{inlinedata}
  &\qqqquad
  \begin{inlinecodata}{\top : \CBN}
  \end{inlinecodata}
  \\[\declskip]&
  \begin{data}[t]{(X{:}\CBV) \otimes (Y{:}\CBV) : \CBV}
    \typecon{X{:}\CBV,Y{:}\CBV}{}{(\blank,\blank)}{X \otimes Y}{\CBV}
  \end{data}
  &\qqqquad
  \begin{codata}[t]{(X{:}\CBN) \parr (Y{:}\CBN) : \CBN}
    \typeobs{}{X{:}\CBN,Y{:}\CBN}{[\blank,\blank]}{X \parr Y}{\CBN}
  \end{codata}
  \\[\declskip]&
  \begin{inlinedata}{1 : \CBV}
    \typecon{}{}{()}{1}{\CBV}
  \end{inlinedata}
  &\qqqquad
  \begin{inlinecodata}{\bot : \CBV}
    \typeobs{}{}{[]}{\bot}{\CBN}
  \end{inlinecodata}
  \\[\declskip]&
  \begin{data}{\involp(X{:}\CBN) : \CBV}
    \typecon{}{X{:}\CBN}{\Cont}{\involp X}{\CBV}
  \end{data}
  &\qqqquad
  \begin{codata}{\involn(X{:}\CBV) : \CBN}
    \typeobs{X{:}\CBV}{}{\Throw}{\involn X}{\CBN}
  \end{codata}
\end{aligned}
\end{math}
\\[\declskip]
Quantifier (\co)data types
\\
\begin{math}
\begin{aligned}
  &
  \begin{data}{\exists_{\strat{S}} (X{:}\strat{S} {\to} \CBV) : \CBV}
    \smalltypefuncon[Y{:}\strat{S}]{X~Y{:}\CBV}
    {\pack{\blank}{\blank}}{\exists_{\strat{S}} X}{\CBV}
  \end{data}
  &\qqqquad
  \begin{codata}{\forall_{\strat{S}} (X{:}\strat{S} {\to} \CBN) : \CBN}
    \smalltypeobs[Y{:}\strat{S}]{}{X~Y{:}\CBN}
    {\spec{\blank}{\blank}}{\forall_{\strat{S}} X}{\CBN}
  \end{codata}
\end{aligned}
\end{math}
\\[\declskip]
Polarity shift (\co)data types
\\
\begin{math}
\begin{aligned}
  &
  \begin{data}{\ToPos[\strat{S}](X{:}\strat{S}) : \CBV}
    \typefuncon{X{:}\strat{S}}{\Wrap[\strat{S}]}{\ToPos[\strat{S}]X}{\CBV}
  \end{data}
  &\qqqquad
  \begin{data}{\FromPos[\strat{S}](X{:}\CBV) : \strat{S}}
    \typefuncon{X{:}\CBV}{\Delay[\strat{S}]}{\FromPos[\strat{S}]X}{\strat{S}}
  \end{data}
  \\[\declskip]&
  \begin{codata}{\ToNeg[\strat{S}](X{:}\strat{S}) : \CBN}
    \typeobs{}{X{:}\strat{S}}{\Unwrap[\strat{S}]}{\ToNeg[\strat{S}]X}{\CBN}
  \end{codata}
  &\qqqquad
  \begin{codata}{\FromNeg[\strat{S}]}
    \typeobs{}{X{:}\CBN}{\Force[\strat{S}]}{\FromNeg[\strat{S}]X}{\strat{S}}
  \end{codata}
\end{aligned}
\end{math}
\caption{(\Co)Data declarations of the dual System $\DUAL$ connectives.}
\label{fig:dual-basis}
\end{figure}

Several examples of data and {\co}data type declarations are given in
\cref{fig:dual-basis}.  More specifically, these declarations give definitions
for each of the connectives in the $\DUAL$ core calculus from
\cref{sec:dual-core}, which illustrates that the $\XDUAL$ calculus
completely subsumes the core calculus.  As another example of an exotic
declaration, we can declare function types as {\co}data, as well the dual to the
function type (known as subtraction or difference, and often written as $A-B$):
\begin{displaymath}
  \begin{codata}{(X{:}\CBV) \to (Y{:}\CBN) : \CBN}
    \typeobs{X{:}\CBV}{Y{:}\CBN}{\app\blank\blank}{X \to Y}{\CBN}
  \end{codata}
  \qqqquad
  \begin{data}{(X{:}\CBV) \coto (Y{:}\CBN) : \CBV}
    \typecon{X{:}\CBV}{Y{:}\CBN}{\yield\blank\blank}{X \coto Y}{\CBV}
  \end{data}
\end{displaymath}

\subsection{Type system and semantics}

In order to finish our extension to $\DUAL$, we also need to extends its type
system and semantics.  The type system for $\XDUAL$ keeps the rules from
\cref{fig:dual-core-kinds,fig:dual-core-generic-types}, and generalizes the form
of each judgment to be parameterized by a set of declarations (denoted by
$\conns{G}$) that replaces the $\DUAL$ annotation.  There are also
connective-specific typing rules for each declaration in $\conns{G}$, given in
\cref{fig:dual-declared-types}.
\begin{figure}
Given that $\conns{G}$ contains the data declaration
\begin{displaymath}
\begin{gathered}
  \begin{inlinedata}{\mk{F}\many{(\ann{X}{k})} : \strat{S}}
    \many[i]
    {
      \smalltypecon[{\many[j]{\ann{Y_{ij}}{\strat{S}'_{ij}}}}]
      {\many[j]{A_{ij}:\strat{T}_{ij}}}{\many[j]{B_{ij}:\strat{R}_{ij}}}
      {\mk{K}_i}{\mk{F}\many{X}}{\strat{S}}
    }
  \end{inlinedata}
\end{gathered}
% \in \conns{G}
\end{displaymath}
we have the rules:
\begin{small}
\begin{displaymath}
\begin{gathered}
  \infer[{\mk{F}}R_i]
  {
    \typetm[\conns{G}][\Theta]{\Gamma}{\Delta}
    {\mk{K}_i\many[j]{C'_j}~\many[j]{v_j}~\many[j]{e_j}}
    {\mk{F}\many{C}}
  }
  {
    \many[j]{(\typety[\conns{G}]{\Theta}{C'_j}{\strat{S}'_{ij}})}
    &
    \many[j]
    {(\typecotm[\conns{G}][\Theta]{\Gamma}{\Delta}{e_j}{B_{ij}\subs{\rho}})}
    &
    \many[j]
    {(\typetm[\conns{G}][\Theta]{\Gamma}{\Delta}{v_j}{A_{ij}\subs{\rho}})}
    &
    \rho = \many{\asub{X}{C}},\many[j]{\asub{Y_{ij}}{C'_{j}}}
  }
  \\[-0.5ex]
  \infer[{\mk{F}}L]
  {
    \typecotm[\conns{G}][\Theta]{\Gamma}{\Delta}
    {
      \caseof
      {
        \many[i]
        {
          \recv
          {(\mk{K}_i\many[j]{Y_{ij}}\many[j]{\alpha_{ij}}\many[j]{x_{ij}})}
          {c_i}
        }
      }
    }
    {\mk{F}\many{C}}
  }
  {
    \many[i]
    {
      \typecmd[\conns{G}][\Theta,{\many[j]{Y_{ij}:\strat{S}'_{ij}}}]
      {\Gamma,\many[j]{x_{ij}:A_{ij}\subs{\rho}}}
      {\many[j]{\alpha_{ij}:B_{ij}\subs{\rho}},\Delta}
      {c_i}
    }
    &
    \rho = \many{\asub{X}{C}}
  }
\end{gathered}
\end{displaymath}
\end{small}%
Given that $\conns{G}$ contains the {\co}data declaration
\begin{displaymath}
  \begin{inlinecodata}{\mk{G}\many{(\ann{X}{k})} : \strat{S}}
    \many[i]
    {
      \smalltypeobs[{\many[j]{\ann{Y_{ij}}{\strat{S}'_{ij}}}}]
      {\many[j]{A_{ij}:\strat{T}_{ij}}}{B_i:\strat{R}_i}
      {\mk{O}_i}{\mk{G}\many{X}}{\strat{S}}
    }
  \end{inlinecodata}
  % \in \conns{G}
\end{displaymath}
we have the rules:
\begin{small}
\begin{displaymath}
\begin{gathered}
  \infer[{\mk{G}}L_i]
  {
    \typecotm[\conns{G}][\Theta]{\Gamma}{\Delta}
    {\mk{O}_i\many[j]{C'_j}~\many[j]{v_j}~\many[j]{e_j}}
    {\mk{G}\many{C}}
  }
  {
    \many[j]{(\typety[\conns{G}]{\Theta}{C'_j}{\strat{S}'_{ij}})}
    &
    \many[j]
    {(\typetm[\conns{G}][\Theta]{\Gamma}{\Delta}{v_j}{A_{ij}\subs{\rho}})}
    &
    \many[j]
    {(\typecotm[\conns{G}][\Theta]{\Gamma}{\Delta}{e_j}{B_{ij}\subs{\rho}})}
    &
    \rho = \many{\asub{X}{C}},\many[j]{\asub{Y_{ij}}{C'_{j}}}
  }
  \\[-0.5ex]
  \infer[{\mk{G}}R]
  {
    \typetm[\conns{G}][\Theta]{\Gamma}{\Delta}
    {
      \cocaseof
      {
        \many[i]
        {
          \send
          {[\mk{O}_i\many[j]{Y_{ij}}\many[j]{x_{ij}}\many[j]{\alpha_{ij}}]}
          {c_i}
        }
      }
    }
    {\mk{F}\many{C}}
  }
  {
    \many[i]
    {
      \typecmd[\conns{G}][\Theta,{\many[j]{Y_{ij}:\strat{S}'_{ij}}}]
      {\Gamma,\many[j]{x_{ij}:A_{ij}\subs{\rho}}}
      {\many[j]{\alpha_{ij}:B_{ij}\subs{\rho}},\Delta}
      {c_i}
    }
    &
    \rho = \many{\asub{X}{C}}
  }
\end{gathered}
\end{displaymath}
\end{small}
\caption{System $\XDUAL$ typing rules for user-defined classical data and
  {\co}data types.}
\label{fig:dual-declared-types}
\end{figure}
Again, the special instances of these inference rule templates for the
declarations in \cref{fig:dual-basis} give exactly the right and left rules in
\cref{fig:dual-core-positive-types,fig:dual-core-negative-types,fig:dual-core-shift-types};
fully dual (\co)data declarations subsume the type system for the core $\DUAL$
calculus.

Notice that the equational theory and operational semantics in
\cref{fig:dual-core-equality,fig:dual-core-operation} automatically extend to
accommodate user-defined data and {\co}data types just by considering the
extended notions of patterns $p$ and {\co}patterns $q$ as given in
\cref{fig:dual-ext-syntax}.  Furthermore, the generalizations to weak-head
normal terms and forcing contexts can be stated in terms of these extended
(\co)patterns as follows:
\begin{align*}
  W &::= x \Alt p\subs{\rho} \Alt \cocaseof{\many{\send{q}{c} \mid}}
  &
  F &::= \alpha \Alt q\subs{\rho} \Alt \caseof{\many{\recv{p}{c} \mid}}
\end{align*}
The only aspect of the semantics that is left to be stated is how to handle the
occurrence of non-value or non-{\co}value parameters to constructors and
destructors.  We will handle this possibility in the same way as we originally
did in \cref{sec:intro-sequent}: define a translation into an appropriate
sub-syntax which rules out this possibility.  We can run arbitrary commands by
first translating them and then applying the operational semantics to the
well-behaved sub-syntax.  And to reason about equality, we will identify all
commands, terms, and {\co}terms with their translation in the equational theory.
But this translation will serve another purpose besides just defining the
semantics of $\XDUAL$: it is the first step of compiling all of $\XDUAL$ into
the core $\DUAL$.

%\section{Encoding user-defined classical (\co)data types into the core $\DUAL$}
\section{Compiling Classical (\Co)Data Types Into the Core
  \texorpdfstring{$\DUAL$}{D}}
\label{sec:dual-compile}

Our goal now is to demonstrate that all the new features of the extensible
$\XDUAL$ calculus can be expressed within the core $\DUAL$ calculus.  We will do
this by showing how to compile all of $\XDUAL$ into $\DUAL$ using only
macro-expansions, which shows that the two languages are equally expressive
\cite{Felleisen1991OEPPL}.  This compilation process involves two phases:
\begin{enumerate*}
\item a translation into a focused sub-syntax defined as a generalization of
  focusing described in
  \cref{sec:intro-sequent,sec:intro-polarity,sec:intro-need}, and
\item a translation of user-defined data and {\co}data types into the core
  connectives in $\DUAL$.
\end{enumerate*}
The composition of these two steps gives an embedding of $\XDUAL$ inside
$\DUAL$.

\subsection{Focusing the syntax}

In order to identify the focused sub-syntax of $\XDUAL$, we have to ensure that
constructors ($\mk{K}$) and destructors ($\mk{O}$) are only ever applied to
values and {\co}values.  The key idea is that the correct notion of value and
{\co}value for each parameter can be found from the signature of every
constructor and destructor.  That means there is a different sub-syntax for each
set of declarations, $\conns{G}$.

\begin{definition}[Focused sub-syntax]
\label{def:focused-subsyntax}
The \emph{$\conns{G}$-focused sub-syntax} of the $\XDUAL$ calculus restricts
constructions to the form
$\mk{K}\many{C}\many{E_{\strat{R}}}\many{V_{\strat{T}}}$ and destructions to the
form $\mk{O}\many{C}\many{E_{\strat{R}}}\many{V_{\strat{T}}}$ when,
respectively,
\begin{align*}
  &
  \typecon[\many{Y{:}\strat{S}}]
  {\many{A:\strat{T}}}{\many{B:\strat{R}}}
  {\mk{K}}{\mk{F}\many{X}}{\strat{S}}
  \in
  \conns{G}
  &&
  \typeobs[\many{\many{Y{:}\strat{S}}}]
  {\many{A:\strat{T}}}{\many{B:\strat{R}}}
  {\mk{O}}{\mk{G}\many{X}}{\strat{S}}
  \in
  \conns{G}
\end{align*}
\end{definition}
We can then compile any $\XDUAL$ expression into a corresponding
$\conns{G}$-focused expression through the following macro-expansion, that just
names the parameters to every constructor and destructor:
\begin{align*}
  (\mk{K}\many{C}\many{e}\many{v})^\lift
  &=
  \outp\alpha
  \cut
  {
    \outp\beta\dots
    \cut
    {v^\lift}
    {\inp y\dots \cut{\mk{K}\many{C}\many{\beta}{\many{y}}}{\alpha}}
  }
  {e^\lift}
  \\
  (\mk{O}\many{C}\many{v}\many{e})^\lift
  &=
  \inp x
  \cut
  {v^\lift}
  {
    \inp y\dots
    \cut
    {\outp\beta\dots \cut{x}{\mk{O}\many{C}\many{y}\many{\beta}}}
    {e^\lift}
  }
\end{align*}
where the type and discipline annotations on each command can be derived from
the signature of $\mk{K}$ and $\mk{O}$.

Recall that the above focusing translation is part of the definition of the
semantics for $\XDUAL$, which says how to handle terms and {\co}terms which do
not fit within the focused sub-syntax.  In the operational semantics, we can
just compile a command $c$ to $c^\lift$ prior to evaluation.  And formally
stated in the equational theory, we make the additional assertion that
$\typecmd[\conns{G}][\Theta]{\Gamma}{\Delta}{c = c^\lift}$ for any
$\typecmd[\conns{G}][\Theta]{\Gamma}{\Delta}{c}$ (and analogously for terms and
{\co}terms).

\subsection{Translating user-defined types}

\begin{figure}
\centering
\begin{gather*}
\begin{aligned}
  \den[\DUAL]{\mk{F}}[\conns{G}]
  &{\defeq}
  \many{\fn{X{:}k}}
  \FromPos[\strat{S}]
  (
  \many[i]
  {
    (
    \many[j]{\exists Y_{ij}{:}\strat{S}'_{ij}.}
    (
    \many[j]{(\involp (\ToNeg[\strat{R}_{ij}]B_{ij})) \otimes}
    (\many[j]{(\ToPos[\strat{T}_{ij}]A_{ij}) \otimes} 1)
    )
    )
    \oplus
  }
  0
  )
  \\
  \den[\DUAL]{\mk{G}}[\conns{G}]
  &{\defeq}
  \many{\fn{X{:}k}}
  \FromNeg[\strat{S}]
  (
  \many[i]
  {
    (
    \many[j]{\forall Y_{ij}{:}\strat{S}'_{ij}.}
    (
    \many[j]{(\involn (\ToPos[\strat{T}_{ij}]A_{ij})) \parr}
    (\many[j]{(\ToNeg[\strat{R}_{ij}]B_{ij}) \parr} \bot)
    )
    )
    \with
  }
  \top
  )
\end{aligned}
\\[1ex]
\begin{aligned}
  \den[\DUAL]
  {
    \mk{K}_i \many{Y} ~ \many{\alpha} ~ \many{x}
  }[\conns{G}]
  &\defeq
  \Delay[\strat{S}]
  \left(
    \Inr^i
    \left(
      \Inl
      \left(
        \many{\pack{Y}{}}
        \left(
          \many{\cont{[\unwrap[\strat{R}]{\alpha}]},}
          \left(
            \many{\wrap[\strat{T}]{x},}
            ()
          \right)
        \right)
      \right)
    \right)
  \right)
  \\
  \den[\DUAL]
  {
    \mk{O}_i \many{Y} ~ \many{x} ~ \many{\alpha}
  }[\conns{G}]
  &\defeq
  \Force[\strat{S}]
  \left[
    \Snd^i
    \left[
      \Fst
      \left[
        \many{\spec{Y}{}}
        \left[
          \many{\throw{[\wrap[\strat{T}]{x}]},}
          \left[
            \many{\unwrap[\strat{R}]{\alpha},}
            []
          \right]
        \right]
      \right]
    \right]
  \right]
  \\
  \densmall[\DUAL]{\cocaseof{\many[i]{\send{q_i}{c_i}}}}[\conns{G}]
  &\defeq
  \cocaseofsmall
  {
    \many[i]
    {
      \send
      {\den[\DUAL]{q_i}[\conns{G}]}
      {\den[\DUAL]{c_i}[\conns{G}]}
    }
  }
  \qqqquad
  \den[\DUAL]{p\subs{\rho}}[\conns{G}]
  \defeq
  \den[\DUAL]{p}[\conns{G}]\subs{\den[\DUAL]{\rho}[\conns{G}]}
  \\
  \densmall[\DUAL]{\caseof{\many[i]{\recv{p_i}{c_i}}}}[\conns{G}]
  &\defeq
  \caseofsmall
  {
    \many[i]
    {
      \recv
      {\den[\DUAL]{p_i}[\conns{G}]}
      {\den[\DUAL]{c_i}[\conns{G}]}
    }
  }
  \qqqquad
  \den[\DUAL]{q\subs{\rho}}[\conns{G}]
  \defeq
  \den[\DUAL]{q}[\conns{G}]\subs{\den[\DUAL]{\rho}[\conns{G}]}
  \\
  \densmall[\DUAL]
  {
    \many{\asub{\ann{X}k}{C}}
    \many{\asub{\anns{x}{}{\strat{T}}}{V_{\strat{T}}}}
    \many{\asub{\anns{\alpha}{}{\strat{R}}}{E_{\strat{R}}}}
  }[\conns{G}]
  &\defeq
  \many{\asub{\ann{X}k}{\den[\DUAL]{C}[\conns{G}]}}
  \many{\asub{\anns{x}{}{\strat{T}}}{\den[\DUAL]{V_{\strat{T}}}[\conns{G}]}}
  \many{\asub{\anns{\alpha}{}{\strat{R}}}{\den[\DUAL]{E_{\strat{R}}}[\conns{G}]}}
\end{aligned}
\\[1ex]
\begin{aligned}
  \Inj{i}^0(p) &= p
  &\qquad
  \Inj{i}^{n+1}(p) &= \Inj{i}(\Inj{i}^n(p))
  &\qqqquad
  \Proj{i}^0[q] &= q
  &\qquad
  \Proj{i}^{n+1}[q] &= \Proj{i}[\Proj{i}^n[q]]
\end{aligned}
\end{gather*}
\caption{Compiling System $\XDUAL$ into System $\DUAL$.}
\label{fig:compilingGeneralized}
\end{figure}

We are now prepared to give a translation of arbitrary user-defined (\co)data
types into the core $\DUAL$ connectives.  The translation is parameterized by a
global environment $\conns{G}$ so that we know the overall shape of each
connective, and is given in \cref{fig:compilingGeneralized}. The
translation of individual data and {\co}data type constructors assume that
$\conns{G}$ contains a data or {\co}data type of the forms listed in
\cref{fig:dual-declared-types}.  The translation of data types wraps every
component type in the appropriate shift ($\ToPos$ for inputs and $\ToNeg$ for
outputs) into the realm of positive types, $\involp$-negates every additional
output of each constructor, combines the components of each constructor with a
nested $\otimes$ pair, combines the choice between different constructors with a
nested $\oplus$ sum, and introduces all the private types with $\exists$
quantifiers.  Finally, the translated data type is injected into the desired
discipline with the outer-most $\FromPos$ shift to come back from the realm of
positive types.  The translation of multi-output {\co}data is exactly dual to
the data translation, making a round trip through the realm of negative types
for the purpose of using the basic {\co}data building blocks.  We can generalize
this translation to arbitrary types homomorphically in all other cases (\ie
$\den[\DUAL]{A~B}[\conns{G}]=\den[\DUAL]{A}[\conns{G}]~\den[\DUAL]{B}[\conns{G}]$,
\etc[]), iterating as needed to translate the connectives introduced by
declarations. The translation of (\co)patterns, (\co)pattern-matching objects,
and (\co)data structures follows the translation of types.

For example, with declarations of function and subtraction types
\begin{displaymath}
  \begin{codata}{(X{:}\CBV) \to (Y{:}\CBN) : \CBN}
    \typeobs{X{:}\CBV}{Y{:}\CBN}{\app\blank\blank}{X \to Y}{\CBN}
  \end{codata}
  \qqqquad
  \begin{data}{(X{:}\CBV) \coto (Y{:}\CBN) : \CBV}
    \typecon{X{:}\CBV}{Y{:}\CBN}{\yield\blank\blank}{X \coto Y}{\CBV}
  \end{data}
\end{displaymath}
in the set of declarations $\conns{G}$, we get the following translations:
\begin{align*}
  \den[\DUAL]{\to}[\conns{G}]
  &=
  \fn{X{:}\CBV}\fn{Y{:}\CBN}\FromNeg[\CBN]((\involn(\ToPos[\CBV] X) \parr (\ToNeg[\CBN]Y \parr \bot)) \with \top)
%  \fn{X{:}\CBV}\fn{Y{:}\CBN}\FromNeg[\CBN]((\involn(\ToPos[\CBV] \den[\DUAL]{X}[\conns{G}]) \parr (\ToNeg[\CBN]Y \parr \bot)) \with \top)
  \\
  \den[\DUAL]{\coto}[\conns{G}]
  &=
  \fn{X{:}\CBV}\fn{Y{:}\CBN}\FromPos[\CBV]((\involp(\ToNeg[\CBN]Y) \otimes (\ToPos[\CBV]X \otimes 1)) \oplus 0)
\end{align*}
But these are rather heavy translation for such simple connectives.  Instead, we
can sometimes simplify the results of translation based on the idea of
\emph{type isomorphisms}, which we will explore next in
\cref{sec:encoding-type-iso}.  For these specific types, we have the following
isomorphic translations:
\begin{align*}
  \den[\DUAL]{A \to B}[\conns{G}]
  &\iso
    (\involn \den[\DUAL]{A}[\conns{G}]) \parr \den[\DUAL]{B}[\conns{G}]
  &
  \den[\DUAL]{A \coto B}
  &\iso
    (\involp \den[\DUAL]{B}[\conns{G}]) \otimes \den[\DUAL]{A}[\conns{G}]
\end{align*}

\subsection{Correctness of compilation}

The only thing remaining is to ensure that this compilation of values and
{\co}values preserves the semantics of the original program that used custom
data and {\co}data types: the observable result of running a program (that is,
the set of needed variables and {\co}variables) is exactly the same before or
after compilation, and all the equalities of $\XDUAL$ are preserved in $\DUAL$.
\begin{theorem}[Operational Correspondence]
  \label{thm:operational-correspondence}
  For non-cyclic $\conns{G}$, command $c$ in $\XDUAL$, and non-empty set of
  (\co)variables $R$, the following statements are equivalent:
  \begin{itemize}
  \item $c \sreds c'$ for some finished $c'$ in $\XDUAL$ such that
    $\neededvar(c') = R$, and
  \item $\den[\DUAL]{c}[\conns{G}] \sreds c'$ for some finished $c'$ in $\DUAL$
    such that $\neededvar(c') = R$.
  \end{itemize}
\end{theorem}
\begin{proof}
  We apply techniques used previously in \cite{DMAV2014CPS} and
  \cite{DownenAriola2014CSCC}, which allow us to infer one direction of the
  bi-implication from the other.  First, observe the following properties of the
  compilation translation.
  \begin{lemmaenum}
  \item \emph{Heaps}: For all $c$ and $H$, there is a heap context $H'$ such
    that
    \begin{math}
      \den[\DUAL]{H[c]}[\conns{G}] \eq[\alpha] H'[\den[\DUAL]{c}[\conns{G}]]
    \end{math}.
  \item \emph{(\Co)Values}: For all $V_{\strat{S}}$ and $E_{\strat{S}}$,
    $\den[\DUAL]{V_{\strat{S}}}[\conns{G}]$ is an $\strat{S}$-value and
    $\den[\DUAL]{E_{\strat{S}}}[\conns{G}]$ is an $\strat{S}$-{\co}value.
  \item \emph{Need}: $\neededvar(c) = \neededvar(\den[\DUAL]{c}[\conns{G}])$.
  \item \emph{Substitution}: For all commands $c$ and substitutions $\rho$,
    \begin{math}
      \den[\DUAL]{c\subs{\rho}}[\conns{G}] \eq[\alpha]
      \den[\DUAL]{c}[\conns{G}]\subs{\den[\DUAL]{\rho}[\conns{G}]}
    \end{math}.
  \item \emph{Reduction}: For all $c$ and $c'$, if $c \sred c'$ then
    \begin{math}
      \den[\DUAL]{c}[\conns{G}] \sreds \den[\DUAL]{c'}[\conns{G}]      
    \end{math}
    in \emph{one or more} steps.
  \end{lemmaenum}
  Part (a) follows from the facts that the translation is compositional and does
  not change the signs on commands.  Furthermore, parts (b), (c), and (d) follow
  from part (a) and the facts that the translation is both compositional and
  hygienic \cite{DownenAriola2014CSCC} (\ie it does not capture or escape free
  (\co)variables of sub-expressions).  Finally, part (e) follows from part (a)
  and the fact that each operational step is preserved by the translation.  We
  know that the rules for $\betamu[\NotCoNeed]$, $\betatmu[\NotCBNeed]$,
  $\betamuconeed$, and $\betatmuneed$ are preserved by translation due to parts
  (b) and (d).  That leaves only the $\betar[p]$ and $\betar[q]$ rules, which
  take the following form after translation:
  \begin{align*}
  \rewriterule{\den[\DUAL]{\betar[q]}[\conns{G}]}
  {
    \cut
    {\cocaseofsmall{\dots \mid \send{\den[\DUAL]{q}[\conns{G}]}{c} \mid \dots}}
    [\strat{S}]
    {\den[\DUAL]{q}[\conns{G}]\subs{\rho}}
    &\sreds
    c\subs{\rho}
  }
  \\
  \rewriterule{\den[\DUAL]{\betar[p]}[\conns{G}]}
  {
    \cut
    {\den[\DUAL]{p}[\conns{G}]\subs{\rho}}
    [\strat{S}]
    {\caseofsmall{\dots \mid \recv{\den[\DUAL]{p}[\conns{G}]}{c} \mid \dots}}
    &\sreds
    c\subs{\rho}
  }
  \end{align*}
  Both of these reductions can be verified through routine calculation of the
  smaller-step pattern-matching rules from \cref{fig:dual-core-operation} by
  expanding out the nested (\co)patterns from the translation (see
  \cref{sec:isomorphism-proofs} for more details), which confirms that they
  always take \emph{at least} one step.  It then follows that if $c \sreds c'$
  and $c'$ is finished in $\XDUAL$ with $\neededvar(c') = R$, we know by
  transitivity of part (e) that
  $\den[\DUAL]{c}[\conns{G}] \sreds \den[\DUAL]{c'}[\conns{G}]$ and by part (c)
  $\den[\DUAL]{c'}[\conns{G}]$ is finished in $\DUAL$ with
  $\neededvar(\den[\DUAL]{c'}[\conns{G}]) = \neededvar(c')$.

  The reverse direction must also hold because the operational semantics of
  $\DUAL$ and $\XDUAL$ are deterministic \cite{DMAV2014CPS}, which implies two
  more facts about the uniqueness of their results:
  \begin{enumerate}
  \item There is an infinite reduction sequence starting with $c$ if and only if
    there is no finished or stuck $c'$ such that $c \sreds c'$.
  \item For every command $c$, there is \emph{at most one} finished or stuck
    $c'$ such that $c \sreds c'$.
  \end{enumerate}
  Suppose that
  \begin{math}
    \den[\DUAL]{c}[\conns{G}] \sreds c'
  \end{math}
  and $c'$ is finished in $\DUAL$ with $\neededvar(c') = R$.  It cannot be that
  there is an infinite reduction sequence beginning from $c$, because then part
  (e) would imply there is an infinite reduction from
  $\den[\DUAL]{c}[\conns{G}]$, contradicting the first uniqueness fact.  So
  there must be some finished or stuck $c''$ in $\XDUAL$ such that
  $c \sreds c''$.  By part (e) we know
  $\den[\DUAL]{c}[\conns{G}] \sreds \den[\DUAL]{c''}[\conns{G}]$ in $\XDUAL$.
  And from the second uniqueness fact and part (c), it must be that
  $\den[\DUAL]{c''}[\conns{G}]$ is $c'$ and
  \begin{math}
    \neededvar(c'')
    =
    \neededvar(\den[\DUAL]{c''}[\conns{G}])
    =
    \neededvar(c')
    =
    R
  \end{math}.
\end{proof}

\begin{theorem}[Equational Soundness]
  \label{thm:equational-soundness}
  For non-cyclic $\conns{G}$, if
  $\typecmd[\conns{G}][\Theta]{\Gamma}{\Delta}{c \eq c'}$ then
  $\typecmd[\DUAL][\Theta]{\den[\DUAL]{\Gamma}[\conns{G}]}{\den[\DUAL]{\Delta}[\conns{G}]}{\den[\DUAL]{c}[\conns{G}]\eq\den[\DUAL]{c'}[\conns{G}]}$.
\end{theorem}
\begin{proof}
  Follows similarly to \cref{thm:operational-correspondence}.  The equational
  theory introduces the $\swapr[\CBNeed]$ and $\swapr[\CoNeed]$ axioms, which
  are unaffected, and the generalized $\betamu[\strat{S}]$ and
  $\betatmu[\strat{S}]$ axioms, which are still sound due to the fact that
  translation commutes over substitution and preserves the (\co)values of every
  discipline.  It also introduces the $\etar[p]$ and $\etar[q]$ axioms which
  express extensionality of (\co)data abstractions, and take the following form
  after translation:
  \begin{align*}
  \rewriterule{\den[\DUAL]{\etar[q]}[\conns{G}]}
  {
    \cocaseof
    {
      \many
      {
        \send
        {\den[\DUAL]{q}[\conns{G}]}
        {\cutsmall{x}[\strat{S}]{\den[\DUAL]{q}[\conns{G}]}}
      }
    }
    &\eq
    x
  }
  \\
  \rewriterule{\den[\DUAL]{\etar[p]}[\conns{G}]}
  {
    \caseof
    {
      \many
      {
        \recv
        {\den[\DUAL]{p}[\conns{G}]}
        {\cutsmall{\den[\DUAL]{p}[\conns{G}]}[\strat{S}]{\alpha}}
      }
    }
    &\eq
    \alpha
  }
  \end{align*}
  As before, both of these equations can be verified from the smaller-step rules
  from \cref{fig:dual-core-equality} by expanding out the nested (\co)patterns
  from the translation (see \cref{sec:isomorphism-proofs}).  Finally, the
  inference rules of the equational theory (reflexivity, transitivity, symmetry,
  and compatibility) are sound due to the fact that the translation is
  compositional \cite{DownenAriola2014CSCC}.
\end{proof}

% \begin{intermezzo}
% \label{rm:recursive-encodings}
% \input{ex_nat-stream}
% \end{intermezzo}

\section{Correct Encodings as Type Isomorphisms}
\label{sec:encoding-type-iso}

The notion of compilation correctness given above is good for a whole-program
transformation: the semantics of a $\XDUAL$ program (represented as a command)
is preserved by compilation inside of a $\DUAL$ program.  However, what if we
instead want to apply the translation of (\co)data types locally?  Suppose we
translate, not as a way to compile an entire program, but as a way to encoding
only to a particular sub-(\co)term or definition in a larger program.
Correctness of a local encoding requires more than just correctness of a global
compilation.  A global compiler can control both the implementations
((\co)terms) and clients (contexts) to respect some invariants that are
essential for correctness.  But a local encoding can only control the
implementation; it has no control over its clients.  Therefore, to be correct,
an encoding must make sure that it gives an implementation with the exact same
semantics as the source for \emph{every} possible client context that might be
written in the resulting code, not just the clients that happen to be aware of
the details of the encoding.

So what is the appropriate notion of correctness for an encoding that translates
fully-dual user-defined data and {\co}data types into the core $\DUAL$
connectives?  We want to be sure that every equality between the program that
was originally written is preserved by the encoding; otherwise, optimizations
that were valid on the original program may be broken by an unfortunate
encoding.  This kind of idea is captured by an \emph{equational correspondence}
\cite{SabryFelleisen1993RAPCPS} between the source and target calculi of the
encoding process, which states that two expressions are equal exactly when they
are equal after translation.

However, our situation is a little more delicate: by encoding arbitrarily many
types into a small set of core connectives in $\DUAL$, we might accidentally
equate two different types in the source language, which in turn would
accidentally equate their terms and {\co}terms!  So we need to use more finesse
when reasoning about the relationship between the source and target languages
which takes the different types into account.  To that end, we will apply the
idea of a \emph{type isomorphism}---an invertible mapping between two
types---which will let us relate the original declared type with its encoding.
This gives us a principled method of framing encoding correctness, even when
multiple types are collapsed into one.

The only stumbling block with using type isomorphisms to establish correctness
is our use of type abstractions: the $\forall$ and $\exists$ quantifications
over types makes it difficult to compose type isomorphisms together into a big
step encoding.  However, the quantifiers of $\DUAL$ have some extra structure:
they are \emph{parametric}, meaning that processes with abstract types must do
the same thing for every specialization.  Therefore, we generalize the usual
notion of type isomorphism to take this parametricity into account, so that the
witnesses to isomorphisms are likewise parametric, with a \emph{single} witness
that applies uniformly to any particular instantiation.  As a result, we can
compose parametric type isomorphisms freely even under quantifiers, which lets
us incrementally build up the correctness of complex encodings into the core
$\DUAL$ calculus.

\subsection{Parametric type isomorphism}

The encodings we gave make a number of choices: the orders between different
constructors or destructors with respect to one another, between the components
of constructors and destructors, \etc[\/]  Did we make the correct choice to
preserve the semantics of the original type?  Or does this difference even
matter?  It turns out any of these choices are correct: each would give types
that are \emph{isomorphic} to one another.
\begin{definition}[Parametric Type Isomorphism]
  \label{def:dual-type-isomorphism}
  \label{def:parametric-type-isomorphism}
  A \emph{parametric isomorphism} between two types
  $\typety[\conns{G}]{\Theta}{A}{k}$ and $\typety[\conns{G}]{\Theta}{B}{k}$,
  written $\semtypety[\conns{G}]{\Theta}{A \iso B}{k}$, is defined by induction
  on $k$:
  \begin{itemize}
   \item $\semtypety[\conns{G}]{\Theta}{A \iso B}{\strat{S}}$ when there are
    commands
    \begin{math}
      \typecmd[\conns{G}][\Theta]{x:A}{\beta:B}{c'}
    \end{math}
    and
    \begin{math}
      \typecmd[\conns{G}][\Theta]{y:B}{\alpha:A}{c}
    \end{math}
    such that the following equalities are derivable:
    \begin{gather*}
      \typecmd[\conns{G}][\Theta]{x:A}{\alpha:A}
      {
        \cut{\outp\beta c'}[\ann{B}{\strat{S}}]{\inp y c}
        =
        \cut{x}[\ann{A}{\strat{S}}]{\alpha}
      }
      \\
      \typecmd[\conns{G}][\Theta]{y:B}{\beta:B}
      {
        \cut{\outp\alpha c}[\ann{A}{\strat{S}}]{\inp x c'}
        =
        \cut{y}[\ann{B}{\strat{S}}]{\beta}
      }
    \end{gather*}
  \item $\semtypety[\conns{G}]{\Theta}{A \iso B}{k \to l}$ when
    $\semtypety[\conns{G}]{\Theta,X:k}{A~X \iso B~X}{l}$ (for any
    $X \notin \Theta$)
   \end{itemize}
\end{definition}

\begin{figure}

Positive algebraic laws (for $A, B, C : \CBV$):
\begin{align*}
  (A \oplus B) \oplus C &\iso A \oplus (B \oplus C)
  &
  0 \oplus A &\iso A % \iso A \oplus 0
  &
  A \oplus B &\iso B \oplus A
  \\
  (A \otimes B) \otimes C &\iso A \otimes (B \otimes C)
  &
  1 \otimes A &\iso A % \iso A \otimes 1
  &
  A \otimes B &\iso B \otimes A
  \\
  A \otimes (B \oplus C) &\iso (A \otimes B) \oplus (A \otimes C)
  % &
  % (A \oplus B) \otimes C &\iso (A \otimes C) \oplus (B \otimes C)
  &
  A \otimes 0 & \iso 0 % \iso 0 \otimes A
\end{align*}

Negative algebraic laws (for $A, B, C : \CBN$):
\begin{align*}
  (A \with B) \with C &\iso A \with (B \with C)
  &
  \top \with A &\iso A % \iso A \with \top
  &
  A \with B &\iso B \with A
  \\
  (A \parr B) \parr C &\iso A \parr (B \parr C)
  &
  \bot \parr A &\iso A % \iso A \parr \bot
  &
  A \parr B &\iso B \parr A
  \\
  A \parr (B \with C) &\iso (A \parr B) \with (A \parr C)
  % &
  % (A \parr B) \with C &\iso (A \parr C) \with (B \parr C)
  &
  A \parr \top & \iso \top % \iso \top \parr A
\end{align*}

De Morgan laws (with $A, B : \CBV$ and $C, D : \CBN$):
\begin{align*}
  \involn(A \oplus B) &\iso (\involn A) \with (\involn B)
  &\!
  \involn(A \otimes B) &\iso (\involn A) \parr (\involn B)
  &\!
  \involn 0 &\iso \top
  &\!
  \involn 1 &\iso \bot
  &\!
  \involn (\involp C) &\iso C
  \\
  \involp(C \with D) &\iso (\involp C) \oplus (\involp D)
  &\!
  \involp(C \parr D) &\iso (\involp C) \otimes (\involp D)
  &\!
  \involp \top &\iso 0
  &\!
  \involp \bot &\iso 1
  &\!
  \involp (\involn A) &\iso A
\end{align*}

Shift laws (for $A:\CBV$ and $B:\CBN$):
\begin{align*}
  \ToPos[\CBN]B &\iso \FromNeg[\CBV]B
  &
  \ToNeg[\CBV]A &\iso \FromPos[\CBN]A
  &
  \ToPos[\CBV]A &\iso A \iso \FromPos[\CBV]A
  &
  \ToNeg[\CBN]B &\iso B \iso \FromNeg[\CBN]B
\end{align*}

Quantifier laws (for $A, B : \CBV$ and $C, D : \CBN$ where
$X \notin \FV(B) \cup \FV(D)$):
\begin{align*}
  \forall X{:}\strat{S}.\forall Y{:}\strat{T}. C
  &\iso
  \forall Y{:}\strat{T}. \forall X{:}\strat{S}. C
  &
  \exists X{:}\strat{S}.\exists Y{:}\strat{T}. A
  &\iso
  \exists Y{:}\strat{T}. \exists X{:}\strat{S}. A
  \\
  \forall X{:}\strat{S}. D
  &\iso
  D
  &
  \exists X{:}\strat{S}. B
  &\iso
  B
  % \\
  % (\forall X{:}\strat{S}. C) \with D
  % &\iso
  % \forall X{:}\strat{S}. (C \with D)
  % &
  % (\exists X{:}\strat{S}. A) \oplus B
  % &\iso
  % \exists X{:}\strat{S}. (A \oplus B)
  \\
  (\forall X{:}\strat{S}. C) \parr D
  &\iso
  \forall X{:}\strat{S}. (C \parr D)
  &
  (\exists X{:}\strat{S}. A) \otimes B
  &\iso
  \exists X{:}\strat{S}. (A \otimes B)
  \\
  \involp(\forall X{:}\strat{S}. C)
  &\iso
  \exists X{:}\strat{S}. (\involp C)
  &
  \involn(\exists X{:}\strat{S}. A)
  &\iso
  \forall X{:}\strat{S}. (\involn A)
\end{align*}

\caption{Core System $\DUAL$ type isomorphism laws.}
\label{fig:dual-core-laws}
\end{figure}

This notion of parametric type isomorphism lets us formally state how many
representations of types are actually equivalent to one another.  For example,
we have the various type isomorphism laws in \cref{fig:dual-core-laws}.  Some of
these laws take a familiar form from type theory---like the algebraic laws---but
they are stronger in $\DUAL$ than in a conventional programming language: side
effects commonly break these isomorphisms, but here they hold even in the
presence of effects.  The De Morgan laws are standard from classical logic, but
they too are much stronger properties in this setting: they assert not just
equi-provability that only guarantees well-typed mappings, but also assert that
those mappings are value-preserving inverses.  This is especially different from
intuitionistic logic and pure functional languages, which do not have all the De
Morgan laws in either sense.  Intuitively, the problem is intuitionistic logics
and languages only have the $\involn$-form of negation (which is sometimes
defined as the isomorphic $\involn A \iso A \to \bot$), but not $\involp$, which
is why negation is not involutive (\ie $\involn(\ToPos(\involn A)) \not\iso A$)
and does not distribute over products (\ie
$\involn(\ToPos(C \with D)) \not\iso (\ToPos\involn C)\oplus(\ToPos\involn C)$).
But the fully dual, classical setting is expressive enough to capture these
properties using both forms of involutive negation.

Type isomorphisms also let us state how some of the shift connectives are
redundant, as shown by the shift laws.  In particular, within the positive
($\CBV$) and negative ($\CBN$) subset, there are only two shifts of interest
since the two different shifts between $\CBN$ and $\CBV$ are isomorphic, and the
identity shifts on $\CBV$ and $\CBN$ are isomorphic to an identity on types.
But clearly the shifts involving $\CBNeed$ are not isomorphic, since they are
all different kinds to one another and the identity function on $\CBNeed$ types.
Recognizing that sometimes the generic encoding uses identity shifts which can
be elided up to isomorphism, the generic encodings of
$\densmall[\DUAL]{\to}[\conns{G}]$ and $\densmall[\DUAL]{\coto}[\conns{G}]$ are
isomorphic to the simplified ones given earlier.

Finally, due to the parametric nature of type isomorphisms in
\cref{def:parametric-type-isomorphism}, we can now state some laws of the
quantifiers, namely that adjacent pairs of the same quantifier can be reordered,
unused quantifiers can be deleted, and we can distribute some other connectives
over quantifiers.  The most interesting of these is the relationship between
quantifiers and negation.  Intuitionistic languages have an isomorphism
corresponding to
$\involn(\exists X{:}\strat{S}. A) \iso \forall X{:}\strat{S}. (\involn A)$ but
reject the reverse
$\involn(\forall X{:}\strat{S}. C) \iso \exists X{:}\strat{S}. (\involn C)$.
That's because every canonical proof of an existential in intuitionistic logic
must present a specific witness: pointing out the witness to a proof is stronger
than saying it is impossible that one cannot exist.  Thus, the constructive
notion of existentials corresponds to the invertible mapping between existential
parameters of the form $\throw{(\pack{B}{W})}$ (isomorphic to the more familiar
function call $f~(\pack{B}{V})$) and $\spec{B}{\throw{W}}$ (isomorphic to the
System F application $(f~B)~V$).

The polarized treatment of negation helps to illuminate why the
non-intuitionistic isomorphism is objectionable: it is not even well-formed!
The correct form of the second isomorphism is
$\involp(\forall X{:}\strat{S}. C) \iso \exists X{:}\strat{S}. (\involp C)$,
which uses the other negation not found in intuitionistic logic, and which is
quite a different statement.  More concretely, it says that a canonical proof of
an $\exists$ is the same as a canonical \emph{refutation} of a $\forall$.  In
other words, the refutation of a $\forall$ is a specific counter-example, which
is the same as saying that there exists a witness to prove the negative.
Exhibiting a real counter-example is much stronger than merely asserting the
impossibility that no such counter-example exists.  The equivalent strength of
existential witnesses and universal counter-examples corresponds to an
invertible mapping between captured specialization call stacks of the form
$\cont{(\spec{B}{F})}$ and existential packages containing a captured call stack
$\pack{B}{(\cont{F})}$.

\subsection{Parametric type isomorphism is an equivalence relation}

Type isomorphisms are easier to work with than global transformations because
they are compositional: small isomorphisms can be composed together to get
bigger ones.  This is due to the fact that they are \emph{equivalence
  relations}, meaning a reflexive, symmetric, and transitive relation between
types of the same kind.  However, this fact does not come for free, we must
actually show that the composite of inverse mappings are still inverses.  It
turns out that in terms of computation, this idea is captured in exactly one
place: the $\swapr$ law from the equational theory
(\cref{fig:dual-core-equality}) that reassociates variable and {\co}variable
bindings.
\begin{theorem}[Equivalence relation]
\label{thm:type-iso-equiv}
Parametric isomorphism is an equivalence relation, meaning
\begin{lemmaenum}
\item \emph{Reflexivity}: $\semtypety[\conns{G}]{\Theta}{A \iso A}{k}$ if
  $\typety[\conns{G}]{\Theta}{A}{k}$,
\item \emph{Symmetry}: $\semtypety[\conns{G}]{\Theta}{A \iso B}{k}$ if and only
  if $\semtypety[\conns{G}]{\Theta}{B \iso A}{k}$, and
\item \emph{Transitivity}: if $\semtypety[\conns{G}]{\Theta}{A \iso B}{k}$ and
  $\semtypety[\conns{G}]{\Theta}{B \iso C}{k}$ then
  $\semtypety[\conns{G}]{\Theta}{A \iso C}{k}$.
\end{lemmaenum}
\end{theorem}
\begin{proof}
  First, consider special case when $k = \strat{S}$.  Symmetry is immediate by
  \cref{def:dual-type-isomorphism}.  Reflexivity is witnessed by the identity
  command
  $\typecmd[\conns{G}]{x:A}{\alpha:A}{\cut{x}[\ann{A}{\strat{S}}]{\alpha}}$ in
  both directions, because
  \begin{displaymath}
    \cut
    {\outp\alpha\cut{x}[\ann{A}{\strat{S}}]{\alpha}}
    [\ann{A}{\strat{S}}]
    {\inp x \cut{x}[\ann{\strat{S}}]{\alpha}}
    \eq[\etamu\etatmu]
    \cut{x}[\ann{A}{\strat{S}}]{\alpha}
  \end{displaymath}
  Transitivity is the most difficult property to prove, and requires the use of
  $\swapr[\strat{S}]$ equality (which is an axiom for $\strat{S} = \CBNeed$ and
  $\CoNeed$, and can be derived from $\betamu$ or $\betatmu$ for
  $\strat{S} = \CBV$ and $\CBN$).  Suppose that we have the following witnesses
  to the assumed isomorphisms, respectively:
  \begin{align*}
    \typecmd[\conns{G}][\Theta]{x:A}{\beta:B}{c_B}
    &&
    \typecmd[\conns{G}][\Theta]{y:B}{\alpha:A}{c_A}
    \\
    \typecmd[\conns{G}][\Theta]{y':B}{\gamma:C}{c_C}
    &&
    \typecmd[\conns{G}][\Theta]{z:C}{\beta':B}{c_B'}
  \end{align*}
  We can then use these to form the following witnesses to
  $\semtypety[\conns{G}]{\Theta}{A \iso C}{\strat{S}}$:
  \begin{align*}
    c_C'
    &=
    \cut{\outp\beta c_B}[\ann{B}{\strat{S}}]{\inp{y'} c_C}
    :
    (x:A \entails_{\conns{G}}^\Theta \gamma:C)
    &
    c_A'
    &=
    \cut{\outp{\beta'} c_B'}[\ann{B}{\strat{S}}]{\inp y c_A}
    :
    (z:C \entails_{\conns{G}}^\Theta \alpha:A)
  \end{align*}
  Both compositions of the above are the identity command (with only one
  direction shown as the other is dual) like so:
  \begin{align*}
    \cut
    {\outp\alpha c_A'}[\ann{A}{\strat{S}}]{\inp x c_C'}
    &=
    \cut
    {\outp{\alpha} \cut{\outp{\beta'} c_B'}[\ann{B}{\strat{S}}]{\inp{y} c_A}}
    [\ann{A}{\strat{S}}]
    {\inp{x} \cut{\outp{\beta} c_B}[\ann{B}{\strat{S}}]{\inp{y'} c_C}}
    \\
    &\eq[{\swapr[\strat{S}]}]
    \cut
    {\outp{\beta'} c_B'}
    [\ann{B}{\strat{S}}]
    {
      \inp{y}
      \cut
      {\outp{\alpha} c_A}
      [\ann{A}{\strat{S}}]
      {\inp{x} \cut{\outp{\beta} c_B}[\ann{B}{\strat{S}}]{\inp{y'} c_C}}
    }
    \\
    &\eq[{\swapr[\strat{S}]}]
    \cut
    {\outp{\beta'} c_B'}
    [\ann{B}{\strat{S}}]
    {
      \inp{y}
      \cut
      {\outp{\beta}\cut{\outp{\alpha} c_A}[\ann{A}{\strat{S}}]{\inp{x} c_B}}
      [\ann{B}{\strat{S}}]
      {\inp{y'} c_C}
    }
    \\
    &\eq[\<iso>]
    \cut
    {\outp{\beta'} c_B'}
    [\ann{B}{\strat{S}}]
    {
      \inp{y}
      \cut
      {\outp{\beta}\cut{\beta}[\ann{B}{\strat{S}}]{y}}
      [\ann{B}{\strat{S}}]
      {\inp{y'} c_C}
    }
    \\
    &\eq[\etamu\etatmu]
    \cut
    {\outp{\beta'} c_B'}
    [\ann{B}{\strat{S}}]
    {\inp{y'} c_C}
    \\
    &\eq[\<iso>]
    \cut{z}[\ann{C}{\strat{S}}]{\gamma}
  \end{align*}

  We can then finish by induction on a general $k$.  Given that
  $k = l_1 \to \dots \to l_n \to \strat{S}$, then we can apply both sides of
  each type isomorphism to the fresh type variables $X_1:l_1 \dots X_n:l_n$, to
  get an equivalent isomorphism at kind $\strat{S}$.
\end{proof}

\subsection{Parametric type isomorphism is a congruence relation}

But transitivity isn't the only form of compositionality; we also want the
ability to apply type isomorphisms inside any context of a larger type.  In
other words, we want to be sure that isomorphisms are a \emph{congruence
  relation} between types, meaning that it is compatible with all the
type-forming operations.  That way, we would know that for any context $C$ that
could surround both $A$ and $B$, if $A \iso B$ then $C[A] \iso C[B]$.

There are several cases of compatibility to consider based on the grammar of
types.  The most fundamental case of compatibility is with the basic connectives
that form types.  For these, we can show individually that each of the
finitely-many dual core $\DUAL$ connectives are compatible with isomorphism;
this is enough since every other connective can be encoded into these.  This
would ordinarily be difficult to show for the quantifiers ($\forall$ and
$\exists$), because there might be different mappings to the isomorphisms for
each specialization of the abstract type.  However, the notion of parametric
type isomorphism is stronger, and stipulates that the \emph{same} mapping is
used for every such specialization, which gives a \emph{uniform} witness to
isomorphisms found inside quantifiers.
\begin{theorem}[Compatibility with $\DUAL$ Connectives]
\label{thm:core-compatibility}
Parametric type isomorphism is compatible with the core connectives of $\DUAL$,
\ie the following isomorphisms hold for any $\conns{G}$ extending $\DUAL$:
\begin{itemize}
\item For any $\semtypety[\conns{G}]{\Theta}{A_1 \iso A_2}{\CBV}$ and
  $\semtypety[\conns{G}]{\Theta}{B_1 \iso B_2}{\CBV}$ we have
  \begin{align*}
    \semtypety[\conns{G}]{\Theta}{A_1 \oplus B_1 \iso A_2 \oplus B_2}{\CBV}
    &&
    \semtypety[\conns{G}]{\Theta}{\involn A_1 \iso \involn A_2}{\CBN}
    \\
    \semtypety[\conns{G}]{\Theta}{A_1 \otimes B_1 \iso A_2 \otimes B_2}{\CBV}
    &&
    \semtypety[\conns{G}]{\Theta}
    {\FromPos[\strat{S}]A_1 \iso \FromPos[\strat{S}]A_2}{\strat{S}}
  \end{align*}
\item For any $\semtypety[\conns{G}]{\Theta}{A_1 \iso A_2}{\CBN}$ and
  $\semtypety[\conns{G}]{\Theta}{B_1 \iso B_2}{\CBN}$ we have
  \begin{align*}
    \semtypety[\conns{G}]{\Theta}{A_1 \with B_1 \iso A_2 \with B_2}{\CBN}
    &&
    \semtypety[\conns{G}]{\Theta}{\involp A_1 \iso \involp A_1}{\CBV}
    \\
    \semtypety[\conns{G}]{\Theta}{A_1 \parr B_1 \iso A_2 \parr B_2}{\CBN}
    &&
    \semtypety[\conns{G}]{\Theta}
    {\FromNeg[\strat{S}]A_1 \iso \FromNeg[\strat{S}]A_1}{\strat{S}}
  \end{align*}
\item For any $\semtypety[\conns{G}]{\Theta}{A_1 \iso A_2}{\strat{S}}$ we have
  \begin{align*}
    \semtypety[\conns{G}]{\Theta}
    {\ToPos[\strat{S}]A_1 \iso \ToPos[\strat{S}]A_2}{\CBV}
    &&
    \semtypety[\conns{G}]{\Theta}
    {\ToNeg[\strat{S}]A_1 \iso \ToNeg[\strat{S}]A_2}{\CBN}
  \end{align*}
\item For any $\semtypety[\conns{G}]{\Theta}{F_1 \iso F_2}{\strat{S} \to \CBV}$
  and $\semtypety[\conns{G}]{\Theta}{G_1 \iso G_2}{\strat{S} \to \CBN}$ we have
  \begin{align*}
    \semtypety[\conns{G}]{\Theta}
    {\exists_{\strat{S}} F_1 \iso \exists_{\strat{S}} F_2}{\CBV}
    &&
    \semtypety[\conns{G}]{\Theta}
    {\forall_{\strat{S}} G_1 \iso \forall_{\strat{S}} G_2}{\CBN}
  \end{align*}
\end{itemize}
\end{theorem}
\begin{proof}
  The majority of these isomorphisms have been proven previously in
  \cite{Downen2017PhD} (as a corollary of Theorem 8.8) for closed type
  isomorphisms, and they continue to hold exactly as before when generalized
  with the possibility of free type variables.  Here, we show the key additional
  cases for the quantifiers $\forall$ and $\exists$, which make crucial use of
  the added parametricity of isomorphisms between open types.  Since
  parametricity has been taken into account by definition, the isomorphisms
  themselves follow the standard form of unpacking structures built by the
  constructor of $\exists$ or $\forall$ in order to apply the given underlying
  isomorphism.

  Consider the case for the existential quantifier $\exists_{\strat{S}}$ (the
  case for $\forall_{\strat{S}}$ is exactly dual), and assume that the given
  isomorphism $\semtypety[\conns{G}]{\Theta}{F_1 \iso F_1}{\strat{S} \to \CBV}$
  is witnessed by
  \begin{align*}
    \typecmd[\conns{G}][\Theta,X:\strat{S}]{x_2:F_2~X}{\alpha_1:F_1~X}{c_1}
    &&
    \typecmd[\conns{G}][\Theta,X:\strat{S}]{x_1:F_1~X}{\alpha_2:F_2~X}{c_2}
  \end{align*}
  We then have the following pair of commands $c_1'$ and $c_2'$ as witnesses to
  the isomorphism
  $\semtypety[\conns{G}]{\Theta}{\exists_{\strat{S}} F_1 \iso
    \exists_{\strat{S}} F_2}{\CBV}$ (where we omit the $\CBV$ annotation on all
  commands):
  \begin{align*}
    c_1'
    &=
    \cut
    {y_2}
    [\exists_{\strat{S}} F_2]
    {
      \sn{(\pack{X}{x_2})}
      \cut
      {\outp{\alpha_1}c_1}
      [F_1X]
      {\inp{x_1}\cut{\pack{X}{x_1}}[\exists_{\strat{S}}F_1]{\beta_1}}
    }
    :
    (y_2:\exists_{\strat{S}} F_2 \entails_{\conns{G}}^\Theta \beta_1:\exists_{\strat{S}} F_1)
    \\
    c_2'
    &=
    \cut
    {y_1}
    [\exists_{\strat{S}} F_1]
    {
      \sn{(\pack{X}{x_1})}
      \cut
      {\outp{\alpha_2}c_2}
      [F_2X]
      {\inp{x_2}\cut{\pack{X}{x_2}}[\exists_{\strat{S}}F_2]{\beta_2}}
    }
    :
    (y_1:\exists_{\strat{S}} F_1 \entails_{\conns{G}}^\Theta \beta_2:\exists_{\strat{S}} F_2)
  \end{align*}
  Notice how the private type $X$ never escapes: it is existentially quantified
  in both the input ($y_i$) and output ($\beta_j$) (\co)variables, and is only
  used by the process of the underlying isomorphism ($c_1$ or $c_2$).  The
  compositions of $c_1'$ and $c_2'$ are an identity command as follows:
  \begin{align*}
    &
    \cut
    {\outp{\beta_1}c_1'}
    [\exists_{\strat{S}}F_1]
    {\inp{y_1}c_2'}
    \\
    &\eq[\betamu]
    \cut
    {y_2}
    [\exists_{\strat{S}} F_2]
    {
      \sn{(\pack{X}{x_2})}
      \cut
      {\outp{\alpha_1}c_1}
      [F_1X]
      {\inp{x_1}\cut{\pack{X}{x_1}}[\exists_{\strat{S}}F_1]{\inp{y_1}c_2'}}
    }
    \\
    &\eq[{\betatmu\betar[p]}]
    \cut
    {y_2}
    [\exists_{\strat{S}} F_2]
    {
      \sn{(\pack{X}{x_2})}
      \cut
      {\outp{\alpha_1}c_1}
      [F_1X]
      {
        \inp{x_1}
        \cut{\outp{\alpha_2}c_2}
        [F_2X]
        {\inp{x_2}\cut{\pack{X}{x_2}}[\exists_{\strat{S}}F_2]{\beta_2}}
      }
    }
    \\
    &\eq[{\swapr[\CBV]}]
    \cut
    {y_2}
    [\exists_{\strat{S}} F_2]
    {
      \sn{(\pack{X}{x_2})}
      \cut
      {
        \outp{\alpha_2}
        \cut{\outp{\alpha_1}c_1}[F_1X]{\inp{x_1}c_2}
      }
      [F_2X]
      {\inp{x_2}\cut{\pack{X}{x_2}}[\exists_{\strat{S}}F_2]{\beta_2}}
    }
    \\
    &\eq[\<iso>]
    \cut
    {y_2}
    [\exists_{\strat{S}} F_2]
    {
      \sn{(\pack{X}{x_2})}
      \cut
      {
        \outp{\alpha_2}
        \cut{x_2}[F_2X]{\alpha_2}
      }
      [F_2X]
      {\inp{x_2}\cut{\pack{X}{x_2}}[\exists_{\strat{S}}F_2]{\beta_2}}
    }
    \\
    &\eq[\betamu\betatmu]
    \cut
    {y_2}
    [\exists_{\strat{S}} F_2]
    {
      \sn{(\pack{X}{x_2})}
      \cut{\pack{X}{x_2}}[\exists_{\strat{S}}F_2]{\beta_2}
    }
    \\
    &\eq[{\etar[p]}]
    \cut{y_2}[\exists_{\strat{S}} F_2]{\beta_2}
  \end{align*}
  where the other direction is the same, up to swapping the indexes $1$ and $2$.
\end{proof}

The goal now is to lift the compatibility of the core $\DUAL$ connectives to any
other context.  This can be challenging with higher-order type variables.  For
example, suppose that $\semtypety[\DUAL]{}{A \iso B}{\strat{S}}$.  Why, then, is
it the case that
$\semtypety[\DUAL]{X{:}\strat{S}\to\strat{R}}{X~A \iso X~B}{\strat{R}}$?  The
reason must have something to do with the higher-order variable $X$, but we know
nothing about the type it will return.  Therefore, we will rule out such cases,
by requiring that types be concrete.  This restriction is still more liberal
than requiring that all types be closed, and is also limiting enough to prove
compatibility under the binders of the quantifiers.
\begin{definition}[Concrete Type]
  A type $A$ is \emph{concrete} if $A$ contains no subformula of the form
  $X ~ A_1 \dots A_n$ where $n \geq 1$.
\end{definition}
% Therefore, we will rule out such cases, by limiting the type environment
% $\Theta$ to only first-order typing variables.  This restriction is still more
% liberal than requiring that all types be closed, and is also enough to prove
% compatibility under the binders of the quantifiers.
% \begin{definition}[First Order Parametricity]
% \label{def:first-order-param}

% We say that the typing context $\Theta$ is \emph{first order} when every type
% variable $X$ in $\Theta$ is assigned a kind of the form $\strat{S}$.
% \end{definition}

\begin{lemma}[Isomorphism Substitution]
\label{thm:type-iso-subst}

For any concrete type
$\typety[\DUAL]{\Theta,X:k}{A}{\many{\strat{T}}\to\strat{S}}$, if
$\semtypety[\DUAL]{\Theta}{B \iso C}{k}$ then
\begin{math}
  \semtypety[\DUAL]{\Theta}
  {A\subst{X}{B} \iso A\subst{X}{C}}{\many{\strat{T}}\to\strat{S}}
  .  
\end{math}
\end{lemma}
\begin{proof}
  By induction on the derivation of
  $\typety[\DUAL]{\Theta,X:k}{A}{\many{\strat{T}}\to\strat{S}}$ (where we assume
  that $A$ is normalized without loss of generality):
  \begin{itemize}
  \item $A = X$ is immediate, by the definition of substitution.
  \item $A = Y$ (where $Y \neq X$) follows by reflexivity
    (\cref{thm:type-iso-equiv}).
  \item $A = \fn{Y{:}\strat{T}_1} A'$ follows from the inductive hypothesis on
    the sub-derivation of the body
    $\typety[\DUAL]{\Theta,Y:\strat{T}_1,X:k}{A'}{\many{\strat{T}_2}\to\strat{S}}$
    via the definition of higher-kinded type isomorphisms
    (\cref{def:parametric-type-isomorphism}) and the fact that
    $(\fn{Y{:}\strat{T}_1}A')~Y \eq[\beta] A'$.
  \item $A = \mk{F} A_1 \dots A_n$: follows from the inductive hypothesis on the
    sub-derivations of $A_1 \dots A_n$ via \cref{thm:core-compatibility}, since
    every core $\DUAL$ connective has arguments of the appropriate kind
    ($\strat{S}$ or $\strat{T} \to \strat{S}$ for some $\strat{S}$ and
    $\strat{T}$).
  \end{itemize}
  Note that it is not possible to have a case where $A = X ~ A_1 \dots A_n$ (for
  a non-zero $n$), due to the assumption that $A$ is concrete.
\end{proof}

\begin{theorem}[Compatibility with Abstraction and Application]
\label{thm:type-iso-compat}

Parametric type isomorphism is compatible with type application and first-order
type abstraction, \ie for any concrete types $A, A'$ and
$k = \many{\strat{S}} \to \strat{R}$:
\begin{lemmaenum}
\item If
  \begin{math}
    \semtypety[\DUAL]{\Theta, X:\strat{T}}
    {A \iso A'}
    {k}
  \end{math}
  then
  \begin{math}
    \semtypety[\DUAL]{\Theta}
    {\fn{X{:}\strat{T}}A \iso \fn{X{:}\strat{T}}A'}
    {\strat{T} \to k}
    .
  \end{math}
\item If
  \begin{math}
    \semtypety[\DUAL]{\Theta}
    {A \iso A'}
    {\strat{T} \to k}
  \end{math}
  and $\semtypety[\DUAL]{\Theta}{B \iso B'}{\strat{T}}$ then
  \begin{math}
    \semtypety[\DUAL]{\Theta}
    {A~B \iso A'~B'}
    {k}
    .
  \end{math}
\end{lemmaenum}
\end{theorem}
\begin{proof}
  Part (a), which applies type isomorphisms underneath a type-level
  $\lambda$-abstraction, follows from the higher-order case of
  \cref{def:parametric-type-isomorphism}.  The conclusion is defined to be the
  same as
  \begin{math}
    \semtypety[\DUAL]{\Theta,X{:}\strat{T}}
    {(\fn{X{:}\strat{T}}A)~X \iso (\fn{X{:}\strat{T}}A')~X}
    {\strat{S}}
  \end{math}
  which follows from the assumption because
  $(\fn{X{:}\strat{T}}A)~X \eq[\beta] A$ and
  $(\fn{X{:}\strat{T}}A')~X \eq[\beta] A'$.

  We will show part (b), which applies type isomorphisms onto the operation and
  operand of an application, in two parts.  First, we focus on the operand.
  Note that the assumption that
  \begin{math}
    \semtypety[\DUAL]{\Theta}
    {A \iso A'}
    {\strat{T} \to \many{\strat{S}} \to \strat{R}}
  \end{math}
  is defined to be
  \begin{math}
    \semtypety[\DUAL]{\Theta,X{:}\strat{T},\many{Y{:}\strat{S}}}
    {A~X~\many{Y} \iso A'~X~\many{Y}}
    {\strat{R}}
  \end{math}
  which comes with witnesses of the form
  \begin{align*}
    \typecmd[\DUAL][\Theta,X{:}\strat{T},\many{Y{:}\strat{S}}]
    {x':A'X\many{Y}}{\alpha:AX\many{Y}}{c}
    &&
    \typecmd[\DUAL][\Theta,X{:}\strat{T},\many{Y{:}\strat{S}}]
    {x:AX\many{Y}}{\alpha':A'X\many{Y}}{c'}
  \end{align*}
  We also have the standard substitution property for typing derivations---if
  $\typecmd[\DUAL][\Theta,X:k]{\Gamma}{\Delta}{c}$ and
  $\typety[\DUAL]{\Theta}{A}{k}$ are derivable, then so is
  \begin{math}
    \typecmd[\DUAL][\Theta]
    {\Gamma\subst{X}{A}}{\Delta\subst{X}{A}}
    {c\subst{X}{A}}
  \end{math}%
  ---which follows by induction on the typing derivation of the command.
  Therefore, by substitution on the witnesses to the isomorphism, we get
  \begin{align*}
    \typecmd[\DUAL][\Theta,\many{Y{:}\strat{S}}]
    {x':A'B\many{Y}}{\alpha:AB\many{Y}}{c\subst{X}{B}}
    &&
    \typecmd[\DUAL][\Theta,\many{Y{:}\strat{S}}]
    {x:AB\many{Y}}{\alpha':A'B\many{Y}}{c'\subst{X}{B}}
  \end{align*}
  so 
  \begin{math}
    \semtypety[\DUAL]{\Theta,\many{Y{:}\strat{S}}}
    {A~B~\many{Y} \iso A'~B~\many{Y}}
    {\strat{R}}
  \end{math}%
  , which is
  \begin{math}
    \semtypety[\DUAL]{\Theta}
    {A~B \iso A'~B}
    {k}
  \end{math}%
  , holds.

  Next, note that by the previous \cref{thm:type-iso-subst}, we can also perform
  the following substitution of
  $\semtypety[\DUAL]{\Theta}{B \iso B'}{\strat{T}}$ into
  $\typety[\DUAL]{\Theta,X{:}\strat{T}}{A'~X}{k}$ to get
  $\semtypety[\DUAL]{\Theta}{A'~B \iso A'~B'}{k}$. Therefore by transitivity
  (\cref{thm:type-iso-equiv}) we get that
  $\semtypety[\DUAL]{\Theta}{A~B \iso A'~B'}{k}$.
\end{proof}

\subsection{Type isomorphisms are correct encodings}

We are now ready to put together the full isomorphism relating every type to its
encoding.  This is done incrementally by composing each of the base cases (that
encoding a user-declared connective in terms of the core $\DUAL$ connective)
using the fact that isomorphism is a congruence relation.
\begin{theorem}[Encoding Isomorphism]
  For all non-cyclic $\conns{G}$ extending $\DUAL$ and concrete types
  $\typety[\conns{G}]{\Theta}{A}{k}$,
  $\semtypety[\conns{G}]{\Theta}{A \iso \den[\DUAL]{A}[\conns{G}]}{k}$.
\end{theorem}
\begin{proof}
  By induction on the derivation of $\typety[\conns{G}]{\Theta}{A}{k}$, where we
  assume that $A$ is normalized without loss of generality.  Note that the
  encoded type has the same kind, but only using $\DUAL$ connectives, \ie
  $\typety[\DUAL]{\Theta}{\den[\DUAL]{A}[\conns{G}]}{k}$.  Therefore, we can
  apply compatibility (\cref{thm:type-iso-compat}) in the case of an abstraction
  or application or reflexivity (\cref{thm:type-iso-equiv}) in the case of a
  type variable.  The remaining cases are for the connectives $\mk{F}$ declared
  as data or {\co}data types, which are given in \cref{sec:isomorphism-proofs}.
\end{proof}

But why is it useful to have the isomorphism $A \iso \den[\DUAL]{A}[\conns{G}]$?
Because we can use it to automatically generate a local transformation between
the two types!  This is done by using the mappings of the isomorphism (which are
represented internally as commands of the $\XDUAL$ calculus) to selectively
convert terms and {\co}terms from one type to the other.  And these
transformations aren't arbitrary; they form an \emph{equational correspondence}
\cite{SabryFelleisen1993RAPCPS} which says that the exact same equalities hold
between (\co)terms of the two types.  In other words, any optimization that was
valid on the original type $A$ is still valid after encoding it as
$\den[\DUAL]{A}[\conns{G}]$ and vice versa.
\begin{theorem}[Isomorphism Correspondence]
\label{thm:iso-eq-correspondence}
  If $\semtypety[\conns{G}]{\Theta}{A \iso B}{\strat{S}}$, then (\co)terms of
  type $A$ are in equational correspondence with (\co)terms of type $B$,
  respectively.
\end{theorem}
\begin{proof}
  To establish an equational correspondence, we need to find maps between terms
  of the isomorphic types $A$ and $B$.  Since the same language serves as both
  the source and the target, we will use a pair of appropriate contexts based on
  the types $A$ and $B$, that are derived from the given witnesses to the
  isomorphism, which can convert terms and {\co}terms between types $A$ and $B$.
  We can just focus on the case for converting terms, as {\co}term conversion is
  exactly dual.
  
  Suppose that $\typecmd[\conns{G}][\Theta]{x:A}{\beta:B}{c'}$ and
  $\typecmd[\conns{G}][\Theta]{y:B}{\alpha:A}{c}$ witnesses the isomorphism
  $\typety[\conns{G}]{\Theta}{A \iso B}{\strat{S}}$.  The desired contexts are
  then defined as $C = \outp\beta\cut{\hole}[\ann{A}{\strat{S}}]{\inp x c'}$ for
  converting a term from type $A$ to $B$, and
  $C' = \outp\alpha\cut{\hole}[\ann{B}{\strat{S}}]{\inp y c}$ for converting
  back.  To show that these context mappings give an equational correspondence,
  we must show that
  \begin{enumerate*}[(1)]
  \item the mappings preserve equality, and
  \item both compositions of the mappings are an identity.
  \end{enumerate*}
  The first fact follows immediately from compatibility; if two terms are equal,
  then they are equal in any context.  The second fact is derived from the
  definition of type isomorphisms at base kinds, which makes crucial use of
  $\swapr[\strat{S}]$ equality (similar to \cref{thm:type-iso-equiv}).  The
  composition for any term $\typetm[\conns{G}][\Theta]{\Gamma}{\Delta}{v}{A}$
  (assuming $x,y \notin \Gamma$ and $\alpha,\beta \notin \Delta$) is:
  \begin{align*}
    C'[C[v]]
    &\eq
    \outp\alpha
    \cut
    {\outp\beta\cutsmall{v}[\ann{A}{\strat{S}}]{\inp x c'}}
    [\ann{B}{\strat{S}}]
    {\inp y c}
    \\
    &\eq[{\swapr[\strat{S}]}]
    \outp\alpha
    \cut
    {v}
    [\ann{A}{\strat{S}}]
    {\inp x \cutsmall{\outp\beta c'}[\ann{B}{\strat{S}}]{\inp y c}}
    \\
    &\eq[\<iso>]
    \outp\alpha
    \cut{v}[\ann{A}{\strat{S}}]{\inp x \cut{x}[\ann{A}{\strat{S}}]{\alpha}}
    \\
    &\eq[{\etatmu\etamu}]
    v
  \end{align*}
  The reverse composition is also equal to the identity for any
  $\typetm[\conns{G}][\Theta]{\Gamma}{\Delta}{v'}{B}$, where we have
  \begin{math}
    \typetm[\conns{G}][\Theta]{\Gamma}{\Delta}{C[C'[v]] \eq v'}{B}
  \end{math}.
\end{proof}

\begin{corollary}[$\DUAL$ Encoding Correspondence]
\label{thm:encoding-eq-correspondence}
  For all non-cyclic $\conns{G}$ extending $\DUAL$ and closed types
  $\typety[\conns{G}]{}{A}{\strat{S}}$, the (\co)terms of type $A$ are in
  equational correspondence with the (\co)terms of type
  $\den[\DUAL]{A}[\conns{G}]$, respectively.
\end{corollary}

This transformation acts as a local encoding that can be applied to just one
sub-term or sub-{\co}term of the entire program.  In the field of compilers,
this corresponds to an optimization referred to as ``worker/wrapper.''  And the
fact that it is an equational correspondence means that the rest cannot possibly
tell the difference between the original expression and the encoded one, \ie
this encoding faithfully reflects the extensional semantics in the programmer's
definition.

\section{Related Work}
\label{sec:related-work}

There have been several polarized languages
\citep{Levy2001PhD,Zeilberger2009PhD,MunchMaccagnoni2013PhD}, each with subtly
different and incompatible restrictions on which programs are allowed to be
written.  The most common such restriction corresponds to \emph{focusing} in
logic \cite{Andreoli1992LPFPLL}; in the terms used here, focusing means that the
parameters to constructors and observers \emph{must} be values.  Rather than
impose a static focusing restriction on the syntax of programs, we instead imply
a dynamic focusing behavior---which evaluates the parameters of constructors and
observers before (\co)pattern matching---during execution.  Both static and
dynamic notions of focusing are two sides of the same coin, and amount to the
same end \citep{DownenAriola2018CCLSC}.

The other restrictions vary between different frameworks, however.  First, we
might ask where computation can happen.  In Levy's call-by-push-value
\cite{Levy2001PhD}, which is based on the $\lambda$-calculus, value types
(corresponding to positive types) can only be ascribed to values and computation
can only occur at computation types (corresponding to negative types).  In
Zeilberger's calculus of unity \cite{Zeilberger2008OUD}, which is based on the
classical sequent calculus, isolates computation in a separate syntactic
category of \emph{statements} which do not have a return type, but is
essentially the same as call-by-push-value in this regard as both frameworks
only deal with \emph{substitutable} entities, to the exclusion of named
computations which may not be duplicated or deleted.  Second, we might ask what
are the allowable types for variables and, when applicable, {\co}variables.  In
call-by-push-value, variables always have positive types, but in the calculus of
unity variables have negative types or positive \emph{atomic} types (and dually
{\co}variables have positive types or negative atomic types).  These
restrictions explain these two frameworks use a different conversion between the
two different kinds of types: $\FromPos$ introduces a positive variable and
$\ToPos$ introduces a negative one, and in the setting of the sequent calculus
$\FromNeg$ introduces a negative {\co}variable and $\ToNeg$ introduces a
positive one.  They also explain the calculus of unity's pattern matching: if
there cannot be positive variables, then pattern matching \emph{must} continue
until it reaches something non-decomposable like a $\lambda$-abstraction.

In contrast, System L allows for computation to occur at \emph{any} type, and
has no restrictions on the types of variables and {\co}variables.  In both of
these ways, the System $\DUAL$ is spiritually closest to System L.  A motivation
for this choice is that call-by-need forces more generality into the system: if
there is no computation and no variables of call-by-need types, then the entire
point of sharing work is missed.  However, the call-by-value and -name
sub-language can still be reduced down to the more restrictive style of
call-by-push-value and the calculus of unity.  We showed here that the two
styles of positive and negative shifts are isomorphic, so the brunt of the work
is to reduce to the appropriate normal form.  Additionally, negative variables
$x:A:\CBN$ can be eliminated by substituting $\obs{y}{\Force}$ for $x$ where
$y : \FromNeg A : \CBV$, which satisfies the call-by-push-value restriction on
variables.  Alternatively, the calculus of unity restriction on (\co)variables
can be satisfied by type-directed $\eta$-expansion into nested (\co)patterns.

Our notion of data and {\co}data extends the ``jumbo'' connectives of Levy's
jumbo $\lambda$-calculus \cite{Levy2006JLC} to include a treatment of
call-by-need as well the move from mono-discipline to multi-discipline.  Our
notion of (\co)data is also similar to Zeilberger's \cite{Zeilberger2009PhD}
definition of types via (\co)patterns, which is fully dual, extended with
sharing.  This work also extends the work on the operational and equational
theory of call-by-push-value \cite{FSSS2019CBPVIC} to incorporate not only
``jumbo'' connectives, but also call-by-need evaluation and its dual.

There has also been other recent work on extending call-by-push value with
call-by-need evaluation \cite{MM2019ECBPV}, with the goal of studying the impact
of computational effects on the differences between evaluation strategies.  The
main result of this work is to identify some effects which do not distinguish
certain evaluation strategies.  Similar to the fact that call-by-name and
call-by-need evaluation give the same result when the only effect is
non-termination \cite{AMOFW1995CBNLC,AriolaFelleisen1997CBNLC}, McDermott and
Mycroft showed how call-by-value and call-by-need evaluation give the same
result with non-determinism as the only effect.  As a methodology,
\cite{MM2019ECBPV} uses a call-by-push-value framework extended with a memoizing
let construct similar to the one used here, but with the key difference that the
types of call-by-need computations are the same as call-by-name ones.  The
tradeoff of this design difference is that there are fewer kinds of types (only
value and computation types, as in call-by-push-value), but in turn types
provide fewer guarantees (for instance, the type system does not ensure that
call-by-need computations must be shared).  As a consequence, it is not clear
whether or not the type isomorphisms studied here---which rely heavily on the
extensional properties which come from the connection between discipline and
types---extend to this looser type system.

The challenges we faced in establishing a robust local encoding for user-defined
(\co)data types is similar to those encountered in the area of compositional
compiler correctness
\cite{PA2014VOCMLS,BSDA2014VCSMC,NHKMDV2015Pilsner,SBSA2015CCC,GuEtAl2015DSCAL,WWS2019ASVCC,PAC2019FASC,DA2019N700CCT}.
Traditional compiler correctness only states that the compiled code exhibits the
same behavior as the source program.  As such, a compiler may only use a small
fragment of the target language with a simpler semantics, and may use invariants
between client and implementation code that are essential to its correctness.
This rules out the possibility of linking with code in the target language that
does not originate from that compiler; doing so may violate the compiler's
invariants or go outside of its fragment of the target language.  In contrast,
compositional compiler correctness gives a much stronger guarantee that allows
for safely linking with more code, whether it comes from a different compiler or
is originally written in the target language.  The difference between global
compilation and local encodings illustrated in
\cref{sec:dual-compile,sec:encoding-type-iso} is related to this difference
between traditional and compositional compiler correctness, except that the
encodings are performed \emph{within} one language (here $\XDUAL$) as opposed to
\emph{across} a source and target language of a compiler.

\section{Conclusion}
\label{sec:conclusion}

We have showed here how logical duality can help with compiling programs.  On
the one hand, the idea of polarity can be extended with other calling
conventions like call-by-need and its dual, which opens up its applicability to
the implementation of practical lazy functional languages.  On the other hand, a
handful of classical connectives can robustly compile and encode a wide variety
of user-defined types, even for languages with effects, using the notion of
parametric type isomorphisms.  Parametricity lets us apply these encodings to
common forms of type abstractions in programs, specifically parametric
polymorphism and existentially-quantified abstract data types.  Because of the
duality of $\DUAL$'s type abstractions $\forall$ and $\exists$, we can even
encode some inductive \emph{and} {\co}inductive types and computations
% (see \cref{rm:recursive-encodings})
following the style of B\"ohm-Berarducci encodings in System F.  Unfortunately,
these encodings require too much indirection to be isomorphisms in the same
sense as \cref{sec:encoding-type-iso}.  To remedy this situation, the core
$\DUAL$ could also be extended with a form of type recursion capable of robustly
encoding inductive and {\co}inductive types, which we leave to future work.

As a next step, we intend to extend GHC's intermediate language, which already
mixes call-by-need and call-by-value types, with the missing call-by-name types.
%%% Functional calculus
% (see \cref{sec:fun-core} for the functional version of the $\DUAL$ calculus).
Since it already has unboxed types \cite{PeytonJonesLaunchbury1991UVFCCNSFL}
corresponding to positive types, what remains are the fully extensional negative
types.  Crucially, we believe that negative function types would lift the idea
of \emph{call arity}---the number of arguments a function takes before ``work''
is done---from the level of terms to the level of types.  Call arity is used to
optimize curried function calls, since passing multiple arguments at once is
more efficient that computing intermediate closures as each argument is passed
one at a time.  No work is done in a negative type until receiving an $\Unwrap$
request or unpacking a $\Delay$ box, so polarized types compositionally specify
multi-argument calling conventions.

For example, a binary function on integers would have the type
$\mk{Int} \to \mk{Int} \to \ToNeg \mk{Int}$, which only computes when both
arguments are given, versus the type
$\mk{Int} \to \ToNeg \FromNeg(\mk{Int} \to \ToNeg \mk{Int})$
which specifies work is done after the first argument, breaking the call into
two steps since a closure must be evaluated and followed.  This generalizes the
existing treatment of function closures in call-by-push-value to call-by-need
closures.  The advantage of lifting this information into types is so that call
arity can be taken advantage of in higher order functions.  For example, the
$\<zipWith>$ function takes a binary function to combine two lists, pointwise,
and has the type
\begin{math}
  \forall\ann{X}{\CBNeed}. \forall\ann{Y}{\CBNeed}. \forall\ann{Z}{\CBNeed}.
  (X \to Y \to Z)
  \to
  [X]
  \to
  [Y]
  \to
  [Z]
  .
\end{math}
The body of $\<zipWith>$ does not know the call arity of the function it's
given, but in the polarized type built with negative functions:
\begin{math}
  \forall\ann{X}{\CBNeed}. \forall\ann{Y}{\CBNeed}. \forall\ann{Z}{\CBNeed}.
  \FromNeg(\ToPos X \to \ToPos Y \to \ToNeg Z)
  \to
  \ToPos{[X]}
  \to
  \ToPos{[Y]}
  \to
  \ToNeg{[Z]}
\end{math}
the interface in the type spells out that the higher-order function uses the
faster two-argument calling convention.

\section*{Acknowledgments}
\noindent
This work is supported by the National Science Foundation under grants
CCF-1719158 and CCF-1423617.

%% in general the use of bibtex is encouraged

\bibliographystyle{alpha}
\bibliography{sequent,downen}

\clearpage
\appendix

\allowdisplaybreaks[4]

\section{Isomorphisms for \texorpdfstring{$\DUAL$}{D} Quantifiers and Encodings}
\label{sec:isomorphism-proofs}

\subsection{Isomorphism laws for type quantifiers}

Let's consider how to prove some of the specific isomorphism laws from
\cref{fig:dual-core-laws}, namely the quantifier laws.  Each law involving one
of the quantifiers ($\forall$ or $\exists$) can be proved in terms of the notion
of parametric type isomorphism, which means that there is \emph{uniform}
evidence that the type abstractions on both sides can be inter-converted without
loss of information.  The most important first step for establishing these
isomorphisms between polarized connectives is to set up an isomorphism between
their patterns or {\co}patterns; the rest tends to follow.  For our purposes
here, we will focus on just the positive existential quantifier $\exists$, and
note that the isomorphisms involving the negative universal quantifier $\forall$
are exactly dual to the following discussion.

The first isomorphism,
\begin{math}
  \exists X{:}\strat{S}. \exists Y{:}\strat{T}. A
  \iso
  \exists Y{:}\strat{T}. \exists X{:}\strat{S}. A
\end{math}
swaps the type components between the patterns $\pack{X}{\pack{Y}{z}}$ and
$\pack{Y}{\pack{X}{z}}$, and is effectively corresponds the commutation between
the tuples $(x,y)$ and $(y,x)$.  As such, when types are ignored the witnesses
to the isomorphism is effectively the same (up to renaming bound type variables)
in both direction, as only the type-level components are being swapped as
follows:
\begin{small}
\begin{align*}
  c
  &=
  \cut
  {y}
  [\CBV]
  {
    \sn{(\pack{Y}{x})}
    \cut
    {x}
    [\CBV]
    {\sn{(\pack{X}{z})}\cut{\pack{X}{\pack{Y}{z}}}[\CBV]{\alpha}}
  }
  :
  (
    y : \exists Y{:}\strat{T}. \exists X{:}\strat{S}. A
    \entails_\DUAL
    \alpha : \exists X{:}\strat{S}. \exists Y{:}\strat{T}. A
  )
  \\
  c'
  &=
  \cut
  {x}
  [\CBV]
  {
    \sn{(\pack{X}{y})}
    \cut
    {y}
    [\CBV]
    {\sn{(\pack{Y}{z})}\cut{\pack{Y}{\pack{X}{z}}}[\CBV]{\beta}}
  }
  :
  (
    x : \exists X{:}\strat{S}. \exists Y{:}\strat{T}. A
    \entails_\DUAL
    \beta : \exists Y{:}\strat{T}. \exists X{:}\strat{S}. A
  )
\end{align*}
\end{small}%
Notice that both compositions of these two commands are equal to the identities
$\cut{x}[\CBV]{\alpha}$ and $\cut{y}[\CBV]{\beta}$, respectively, according to
the $\betamu$, $\betatmu$, $\betar[p]$, and finally $\etar[p]$ rules, in exact
same manner that the commutation $A \otimes B \iso B \otimes A$ is derived.  And
since the above witnesses are uniform for any choice of type $A$, we can
abstract away the specific $A$ in the witnesses as
\begin{align*}
  c
  &:
  (
    y : \exists Y{:}\strat{T}. \exists X{:}\strat{S}. (F~X~Y)
    \entails_\DUAL^{F:\strat{S}\to\strat{T}\to\CBV}
    \alpha : \exists X{:}\strat{S}. \exists Y{:}\strat{T}. (F~X~Y)
  )
  \\
  c'
  &:
  (
    x : \exists X{:}\strat{S}. \exists Y{:}\strat{T}. (F~X~Y)
    \entails_\DUAL^{F:\strat{S}\to\strat{T}\to\CBV}
    \beta : \exists Y{:}\strat{T}. \exists X{:}\strat{S}. (F~X~Y)
  )
\end{align*}
to get the stronger parametric isomorphism:
\begin{align*}
  \semtypety[\DUAL]{F:\strat{S}\to\strat{T}\to\CBV}
  {
  \exists X{:}\strat{S}. \exists Y{:}\strat{T}. (F ~ X ~ Y)
  \iso
  \exists Y{:}\strat{T}. \exists X{:}\strat{S}. (F ~ X ~ Y)
  }
  {\CBV}
\end{align*}

Next, we have the isomorphism
\begin{math}
  \exists X{:}\strat{S}. B \iso B
\end{math}
where $X$ is not free in $B$, just deletes the redundant abstraction which is
never used.  In the one direction we can insert an unneeded $\pack{X}{y}$ around
a value $x : B$, and in the other direction we can extract and return that $B$
component (which is well-typed, because the existentially quantified variable
$X$ never escapes because $B$ does not reference it), like so:
\begin{align*}
  c' &= \cut{x}[\CBV]{\sn{\pack{X}{y}}\cut{y}[\CBV]{\beta}}
  &&:
  (x : \exists X{:}\strat{S}. B \entails_\DUAL \beta:B)
  \\
  c &= \cut{\pack{X}{y}}[\CBV]{\alpha}
  &&:
  (y : B \entails_\DUAL \alpha : \exists X{:}\strat{S}. B)
\end{align*}
One direction, which composes on the existential type
$\exists X{:}\strat{S}. B$, is just a straightforward calculation
\begin{align*}
  \cut{\outp \alpha c}[\CBV]{\inp x c'}
  &\eq[\etatmu]
  \cut{\outp \alpha c}[\CBV]{\sn{\pack{X}{y}}\cut{y}[\CBV]{\beta}}
  \eq[\betamu]
  \cut{\pack{X}{y}}[\CBV]{\sn{\pack{X}{y}}\cut{y}[\CBV]{\beta}}
  \eq[{\betar[p]}]
  \cut{y}[\CBV]{\beta}
\end{align*}
But the other direction, which composes on the arbitrary type $B:\CBV$, is equal
to an identity due to the fact that the existential type $B$ is call-by-value,
which gives us the strongest possible $\betamu$ rule in the following
derivation:
\begin{align*}
  \cut{\outp \beta c'}[\CBV]{\inp y c}
  &=
  \cut{\outp \beta c'}[\CBV]{\inp y \cut{\pack{X}{y}}[\CBV]{\alpha}}
  \\
  &\eq[\betamu]
  \cut
  {x}
  [\CBV]
  {\sn{\pack{X}{y}}\cut{y}[\CBV]{\inp y \cut{\pack{X}{y}}[\CBV]{\alpha}}}
  \\
  &\eq[\betatmu]
  \cut
  {x}
  [\CBV]
  {\sn{\pack{X}{y}}\cut{\pack{X}{y}}[\CBV]{\alpha}}
  \\
  &\eq[{\etar[p]}]
  \cut{x}[\CBV]{\alpha}
\end{align*}
And again, since the witness to the isomorphism is the same for any type
$B:\CBV$, we can abstract it out to get
\begin{align*}
  c'
  &:
  (x : \exists X{:}\strat{S}. Y \entails_\DUAL^{Y:\CBV} \beta:Y)
  &
  c
  &:
  (y : Y \entails_\DUAL^{Y:\CBV} \alpha : \exists X{:}\strat{S}. Y)
\end{align*}
which guarantees the stronger parametric type isomorphism:
\begin{align*}
  \semtypety[\DUAL]{Y:\CBV}{\exists X{:}\strat{S}. Y \iso Y}{\CBV}
\end{align*}

Finally, we have several isomorphisms for distributing other connectives over
quantifiers.  The most interesting ones are the ones involving negation
\begin{align*}
  \involp (\forall X{:}\strat{S}. C) &\iso \exists X{:}\strat{S}. (\involp C)
  &
  \involn (\exists X{:}\strat{S}. A) &\iso \forall X{:}\strat{S}. (\involn A)
\end{align*}
since the second one holds in intuitionistic logic, but the first one does not.
These isomorphisms relate the patterns
$\cont{(\spec{X}{\beta})} \iso \pack{X}{(\cont{\beta})}$ for the first
isomorphism and the {\co}patterns
$\throw{[\pack{X}{y}]} \iso \spec{X}{[\throw{y}]}$ for the second isomorphism.
The second isomorphism can be translated to intuitionistic logic by encoding the
negation $\involn A$ as $A \to \bot$, but the first isomorphism does not
correspond to anything in intuitionistic logic.  More formally, the witnesses to
the first isomorphism are
\begin{align*}
  c
  &=
  \cut
  {y}
  [\CBV]
  {
    \sn{\cont\delta}
    \cut
    {\fn{[\spec{X}\gamma]}\cut{\pack{X}{(\cont{\gamma})}}[\CBV]{\alpha}}
    [\CBN]
    {\delta}
  }
  :
  (
    y:\involp(\forall X{:}\strat{S}.C)
    \entails_\DUAL
    \alpha:\exists X{:}\strat{S}.(\involp C)
  )
  \\
  c'
  &=
  \cut
  {x}
  [\CBV]
  {
    \sn{(\pack{X}{z})}
    \cut
    {z}
    [\CBV]
    {\sn{\cont{\gamma}}\cut{\cont{(\spec{X}{\gamma})}}[\CBV]{\beta}}
  }
  :
  (
    x:\exists X{:}\strat{S}.(\involp C)
    \entails_\DUAL
    \beta:\involp(\forall X{:}\strat{S}.C)
  )
\end{align*}
Both compositions of these mappings are equal to identity commands as follows,
where $\etar[q][\CBN]$ and $\etar[p][\CBV]$ denote the full $\eta$ laws (as
derived in \cref{sec:dual-core-equation}) that apply to any negative ($\CBN$)
term or positive ($\CBV$) {\co}term, respectively:
\begin{align*}
  &
  \cut{\outp\alpha c}[\CBV]{\outp x c'}
  \\
  &\eq[\betamu]
  \cut
  {y}
  [\CBV]
  {
    \sn{\cont\delta}
    \cut
    {
      \fn{[\spec{X}\gamma]}
      \cut
      {\pack{X}{(\cont{\gamma})}}
      [\CBV]
      {\outp x c'}
    }
    [\CBN]
    {\delta}
  }
  \\
  &\eq[\betatmu]
  \cut
  {y}
  [\CBV]
  {
    \sn{\cont\delta}
    \cut
    {
      \fn{[\spec{X}\gamma]}
      \cut
      {\pack{X}{(\cont{\gamma})}}
      [\CBV]
      {
        \sn{(\pack{X}{z})}
        \cut
        {z}
        [\CBV]
        {\sn{\cont{\gamma}}\cut{\cont{(\spec{X}{\gamma})}}[\CBV]{\beta}}
      }
    }
    [\CBN]
    {\delta}
  }
  \\
  &\eq[{\betar[p]}]
  \cut
  {y}
  [\CBV]
  {
    \sn{\cont\delta}
    \cut
    {
      \fn{[\spec{X}\gamma]}
      \cut
      {\cont{\gamma}}
      [\CBV]
      {\sn{\cont{\gamma}}\cut{\cont{(\spec{X}{\gamma})}}[\CBV]{\beta}}
    }
    [\CBN]
    {\delta}
  }
  \\
  &\eq[{\betar[p]}]
  \cut
  {y}
  [\CBV]
  {
    \sn{\cont\delta}
    \cut
    {
      \fn{[\spec{X}\gamma]}
      \cut{\cont{(\spec{X}{\gamma})}}[\CBV]{\beta}
    }
    [\CBN]
    {\delta}
  }
  \\
  &\eq[{\betatmu}]
  \cut
  {y}
  [\CBV]
  {
    \sn{\cont\delta}
    \cut
    {
      \fn{[\spec{X}\gamma]}
      \cut
      {\outp{\alpha} \cut{\cont{\alpha}}[\CBV]{\beta}}
      [\CBN]
      {\spec{X}{\gamma}}
    }
    [\CBN]
    {\delta}
  }
  \\
  &\eq[{\etar[q][\CBN]}]
  \cut
  {y}
  [\CBV]
  {
    \sn{\cont\delta}
    \cut
    {\outp{\alpha} \cut{\cont{\alpha}}[\CBV]{\beta}}
    [\CBN]
    {\delta}
  }
  \\
  &\eq[{\betamu}]
  \cut
  {y}
  [\CBV]
  {
    \sn{\cont\delta}
    \cut{\cont{\delta}}[\CBV]{\beta}
  }
  \\
  &\eq[{\etar[p]}]
  \cut{y}[\CBV]{\beta}
\end{align*}
\begin{align*}
  &
  \cut{\outp\beta c'}[\CBV]{\inp y c}
  \\
  &\eq[\betamu]
  \cut
  {x}
  [\CBV]
  {
    \sn{(\pack{X}{z})}
    \cut
    {z}
    [\CBV]
    {
      \sn{\cont{\gamma}}
      \cut
      {\cont{(\spec{X}{\gamma})}}
      [\CBV]
      {\inp y c}
    }
  }
  \\
  &\eq[\betatmu]
  \cut
  {x}
  [\CBV]
  {
    \sn{(\pack{X}{z})}
    \cut
    {z}
    [\CBV]
    {
      \sn{\cont{\gamma}}
      \cut
      {\cont{(\spec{X}{\gamma})}}
      [\CBV]
      {
        \sn{\cont\delta}
        \cut
        {\fn{[\spec{X}\gamma]}\cut{\pack{X}{(\cont{\gamma})}}[\CBV]{\alpha}}
        [\CBN]
        {\delta}
      }
    }
  }
  \\
  &\eq[{\betar[p]}]
  \cut
  {x}
  [\CBV]
  {
    \sn{(\pack{X}{z})}
    \cut
    {z}
    [\CBV]
    {
      \sn{\cont{\gamma}}
      \cut
      {\fn{[\spec{X}\gamma]}\cut{\pack{X}{(\cont{\gamma})}}[\CBV]{\alpha}}
      [\CBN]
      {\spec{X}{\gamma}}
    }
  }
  \\
  &\eq[{\betar[q]}]
  \cut
  {x}
  [\CBV]
  {
    \sn{(\pack{X}{z})}
    \cut
    {z}
    [\CBV]
    {
      \sn{\cont{\gamma}}
      \cut{\pack{X}{(\cont{\gamma})}}[\CBV]{\alpha}
    }
  }
  \\
  &\eq[{\betatmu}]
  \cut
  {x}
  [\CBV]
  {
    \sn{(\pack{X}{z})}
    \cut
    {z}
    [\CBV]
    {
      \sn{\cont{\gamma}}
      \cut
      {\cont{\gamma}}
      [\CBV]
      {\inp y \cut{\pack{X}{y}}[\CBV]{\alpha}}
    }
  }
  \\
  &\eq[{\etar[p][\CBV]}]
  \cut
  {x}
  [\CBV]
  {
    \sn{(\pack{X}{z})}
    \cut
    {z}
    [\CBV]
    {\inp y \cut{\pack{X}{y}}[\CBV]{\alpha}}
  }
  \\
  &\eq[{\betatmu}]
  \cut
  {x}
  [\CBV]
  {
    \sn{(\pack{X}{z})}
    \cut{\pack{X}{y}}[\CBV]{\alpha}
  }
  \\
  &\eq[{\etar[p]}]
  \cut{x}[\CBV]{\alpha}
\end{align*}
Notice how the interplay between both call-by-value and call-by-name semantics
as used in the two quantifiers and negation connectives, and its impact on
ensuring the full $\eta$ laws, are essential for deriving the above isomorphism.

The other isomorphisms for distributing multiplicative connectives over
quantifiers given in \cref{fig:dual-core-laws} can be witnessed in the same way
by the following relationship between (\co)patterns that reassociate the two
different forms of grouping operations:
\begin{align*}
  ((\pack{X}{x}), y) &\iso \pack{X}{(x,y)}
  &
  [[\spec{X}\alpha], \beta] &\iso \spec{X}{[\alpha, \beta]}
\end{align*}
However, notice the conspicuously missing laws for distributing the additive
connectives over their matching quantifiers:
\begin{align*}
  (\forall X{:}\strat{S}. C) \with D &\iso \forall X{:}\strat{S}. (C \with D)
  &
  (\exists X{:}\strat{S}. A) \oplus B &\iso \exists X{:}\strat{S}. (A \oplus B)
\end{align*}
That's because there is some missing information (the instantiation choice for
the quantified type $X$) that is unavailable when the other constructor or
destructor (containing $B$ or $D$, respectively) is taken.  For example, if we
were to try to set up a correspondence between the patterns of the second
isomorphism, the best we can do is as follows:
\begin{align*}
  \inl{(\pack{X}{x})} &\gives \pack{X}{(\inl{x})}
  &
  \pack{X}{(\inl{x})} &\gives \inl{(\pack{X}{x})}
  \\
  \inr{y} &\gives \pack{(\FromPos[\strat{S}]0)}{(\inr{y})}
  &
  \pack{X}{(\inr{y})} &\gives \inr{y}
\end{align*}
Notice how there is no quantified type given in the pattern $\inr{y}$, so we can
just use an arbitrary type like $\FromPos[\strat{S}]0$ since it is not
referenced by $y$, anyway.  However, the consequence of picking an arbitrary
type is that a round-trip sends the pattern $\pack{X}{(\inr{y})}$ to
$\pack{(\FromPos[\strat{S}]0)}{(\inr{y})}$.  The same thing happens if we try to
set up a correspondence between the {\co}patterns of the first isomorphism like
so
\begin{align*}
  \fst{[\spec{X}{\alpha}]} &\gives \spec{X}{[\fst{\alpha}]}
  &
  \spec{X}{[\fst{\alpha}]} &\gives \fst{[\spec{X}{\alpha}]}
  \\
  \snd{\beta} &\gives \spec{[\FromNeg[\strat{S}]\top]}{[\snd{\beta}]}
  &
  \spec{X}{[\snd{\beta}]} &\gives \snd{\beta}
\end{align*}
where a round-trip sends the {\co}pattern $\spec{X}{[\snd{\beta}]}$ to
$\spec{[\FromNeg[\strat{S}]\top]}{[\snd{\beta}]}$.

But this problem is somewhat artificial: the semantics we give in
\cref{fig:dual-core-equality,fig:dual-core-operation} ignores these type
annotations, so they play no part in the dynamic interpretation of programs.
This fact means that we could enhance the $\eta$ laws for quantifiers to
actually ignore the chosen types as follows:
\begin{align*}
  \fn{[\spec{X}{\beta}]}\cut{x}[\ann{A}{\CBN}]{\spec{B}{\beta}}
  &=
  x
  &
  \sn{(\pack{X}{y})}\cut{\pack{B}{y}}[\ann{A}{\CBV}]{\alpha}
  &=
  \alpha
\end{align*}
These equations would be sound as long as we fully commit to a parametric and
type-irrelevant semantics for the quantifiers.  This can be done formally by
means of type erasure, a parametric reducibility candidate semantics, or just
ignoring type annotations in the results of programs.  And the enhanced $\eta$
laws above are enough to prove that the additive connectives ($\oplus$ and
$\with$) distribute over the quantifiers ($\exists$ and $\forall$,
respectively).  Intuitively, this fact states that the choice of the quantified
type (chosen by the producer for $\exists$ and chosen by the consumer for
$\forall$) cannot impact other branches of computation where it is not used.
For example, a producer of type $\forall X{:}\strat{S}. (C \with D)$, where $X$
is not referenced in $D$, must produce the exact same result when asked for its
second component no matter what type the consumer chooses for $X$.

\subsection{Equational reasoning for macro (\co)patterns}

The encodings of declared data and {\co}data types involves a macro-expansion of
flat (\co)patterns (of the form $\mk{K}\many{X}\many{\alpha}\many{y}$ and
$\mk{O}\many{X}\many{y}\many{\alpha}$) into nested (\co)patterns built from the
constructors and destructors of the core $\DUAL$ types.  Thankfully, these
encodings follow a very particular form, where the \emph{additive} (\co)patterns
(those with multiple options like $\Inj{i}$ of the $\oplus$ connective and
$\Fst{i}$ of the $\with$ connective) are never found inside the
\emph{multiplicative} (\co)patterns (those with multiple sub-patterns like
$(p_1,p_2)$ of the $\otimes$ connective and $[q_1,q_2]$ of the $\parr$
connective).  This special property makes macro-expanding patterns and
{\co}patterns rather straightforward, which can be done with the following
equations:
\begin{align*}
  \cocaseof{\many{\send{\FromNeg[\strat{S}]q}{c}}}
  &=
  \cocaseof{
    \recv
    {\FromNeg[\strat{S}]\alpha}
    {\cut{\cocaseof{\many{\send{q}{c}}}}[\CBN]{\alpha}}
  }
  &  
  \caseof{\many{\recv{\FromPos[\strat{S}]p}{c}}}
  &=
  \caseof{
    \recv
    {\FromPos[\strat{S}]x}
    {\cut{x}[\CBV]{\caseof{\many{\recv{p}{c}}}}}
  }
  \\
  \cocaseof
  {
  \begin{aligned}
    &\many{\send{\fst{q_1}}{c_1}} \\
    &\many{\send{\snd{q_2}}{c_2}}
  \end{aligned}
  }
  &=
  \cocaseof
  {
  \begin{aligned}
    &\send{\fst\alpha}{\cut{\cocaseof{\many{\send{q_1}{c_1}}}}[\CBN]{\alpha}} \\
    &\send{\snd\beta}{\cut{\cocaseof{\many{\send{q_2}{c_2}}}}[\CBN]{\beta}}
  \end{aligned}
  }
  &\!\!\!
  \caseof
  {
  \begin{aligned}
    &\many{\recv{\inl{p_1}}{c_1}} \\
    &\many{\recv{\inr{p_2}}{c_2}}
  \end{aligned}
  }
  &=
  \caseof
  {
  \begin{aligned}
    &\recv{\inl x}{\cut{x}[\CBV]{\caseof{\many{\send{p_1}{c_1}}}}} \\
    &\recv{\inr y}{\cut{y}[\CBV]{\caseof{\many{\send{p_2}{c_2}}}}}
  \end{aligned}
  }
  \\
  \fn{[\spec{X}{q}]}c
  &=
  \fn{[\spec{X}{\alpha}]}\cut{\fn{q}c}[\CBN]{\alpha}
  &
  \sn{(\pack{X}{p})}c
  &=
  \sn{(\pack{X}{y})}\cut{y}[\CBV]{\sn{p}c}
  \\
  \fn{[q_1,q_2]}{c}
  &=
  \fn{[\alpha,\beta]}
    \cut
    {\fn{q_1}\cut{\fn{q_2}{c}}[\CBN]{\beta}}
    [\CBN]
    {\alpha}
  \!\!\!\!&
  \sn{(p_1,p_2)}c
  &=
  \sn{(x,y)}
    \cut
    {x}
    [\CBV]
    {\sn{p_1}\cut{y}[\CBV]{\sn{p_2}c}}
  \\
  \fn{[\throw{p}]}c
  &=
  \fn{[\throw{x}]}{\cut{x}[\CBV]{\sn{p}c}}
  &
  \sn{(\cont{q})}c
  &=
  \sn{(\cont{\alpha})}{\cut{\sn{p}c}[\CBN]{\alpha}}
  \\
  \fn{\alpha}c
  &=
  \outp\alpha c
  &
  \sn{x}c
  &=
  \inp x c
\end{align*}
Using this macro expansion, we can verify that the correspondingly expanded
pattern-matching equalities can be derived from the simpler rules of the core
$\DUAL$ equational theory.
\begin{lemma}[Macro (\co)pattern matching]
\label{thm:macro-beta-eta}

Each of the following operational steps and equations, given by the
macro-expansion $\den[\DUAL]{-}[\conns{G}]$:
\begin{align*}
  \rewriterule{\den[\DUAL]{\betar[q]}[\conns{G}]}
  {
    \cut
    {\cocaseofsmall{\dots \mid \send{\den[\DUAL]{q}[\conns{G}]}{c} \mid \dots}}
    [\strat{S}]
    {\den[\DUAL]{q}[\conns{G}]\subs{\rho}}
    &\sreds
    c\subs{\rho}
  }
  \\
  \rewriterule{\den[\DUAL]{\betar[p]}[\conns{G}]}
  {
    \cut
    {\den[\DUAL]{p}[\conns{G}]\subs{\rho}}
    [\strat{S}]
    {\caseofsmall{\dots \mid \recv{\den[\DUAL]{p}[\conns{G}]}{c} \mid \dots}}
    &\sreds
    c\subs{\rho}
  }
  \\
  \rewriterule{\den[\DUAL]{\etar[q]}[\conns{G}]}
  {
    \cocaseof
    {
      \many
      {
        \send
        {\den[\DUAL]{q}[\conns{G}]}
        {\cutsmall{x}[\strat{S}]{\den[\DUAL]{q}[\conns{G}]}}
      }
    }
    &\eq
    x
  }
  \\
  \rewriterule{\den[\DUAL]{\etar[p]}[\conns{G}]}
  {
    \caseof
    {
      \many
      {
        \recv
        {\den[\DUAL]{p}[\conns{G}]}
        {\cutsmall{\den[\DUAL]{p}[\conns{G}]}[\strat{S}]{\alpha}}
      }
    }
    &\eq
    \alpha
  }
\end{align*}
are derivable in the core $\DUAL$-calculus, where
$\den[\DUAL]{\betar[p]}[\conns{G}]$ and $\den[\DUAL]{\betar[p]}[\conns{G}]$ take
\emph{one or more} steps.
\end{lemma}
\begin{proof}
  The macro $\den[\DUAL]{\betar[p]}[\conns{G}]$ and
  $\den[\DUAL]{\betar[q]}[\conns{G}]$ reductions can be calculated from the
  small-step $\betar[p]$ and $\betar[q]$ rules given in
  \cref{fig:dual-core-equality}.  Likewise, the macro
  $\den[\DUAL]{\etar[p]}[\conns{G}]$ and $\den[\DUAL]{\etar[q]}[\conns{G}]$
  reductions can be calculated from the small-step $\etar[p]$ and $\etar[q]$
  rules with the help of $\betamu\betatmu\etamu\etatmu$ to apply to (\co)values
  (and not just (\co)variables) as was shown in \cref{sec:dual-core}.  The
  calculations follow the same structure as in \cite{Downen2017PhD} Theorem 8.1,
  but with the addition of type-quantifying patterns for $\forall$ and $\exists$
  that are analogous to simpler cases for $\otimes$ and $\parr$ (\co)patterns.
\end{proof}

\subsection{Isomorphisms between declared types and their encoding}

\begin{lemma}[Data Encoding Isomorphism]
\label{thm:data-encoding-iso}

Given any $\conns{G}$ extending $\DUAL$, and any data type declaration
\begin{align*}
  \begin{data}{\mk{F}\many{(X:k)} : \strat{S}}
    &
    \typecon[\many{Y_1:\strat{S}_1}]
    {\many{A_1:\strat{T}_1}}{\many{B_1:\strat{R}_1}}
    {\mk{K}_1}{\mk{F}\many{X}}{\strat{S}}
    \\
    &
    \typecon[\many{Y_n:\strat{S}_n}]
    {\many{A_n:\strat{T}_1}}{\many{B_n:\strat{R}_n}}
    {\mk{K}_n}{\mk{F}\many{X}}{\strat{S}}
  \end{data}
\end{align*}
in $\conns{G}$, it follows that
\begin{math}
  \semtypety
  {}
  {
    \mk{F}
    \iso
    \den[\DUAL]{\mk{F}}[\conns{G}]
  }
  {\many{k} \to \strat{S}}
\end{math}.
\end{lemma}
\begin{proof}
  Given the patterns $p_i=\mk{K}_i\many{Y_i}\many{\alpha_i}\many{x_i}$ matching
  the above declaration for $i$ ranging from $1$ to $n$, the witness to the
  isomorphism is given by the commands:
  \begin{align*}
    c
    &=
    \cut
    {x'}
    [\strat{S}]
    {
      \caseof
      {
        \many
        {
          \recv
          {\den[\DUAL]{p_i}[\conns{G}]}
          {\cutsmall{p_i}[\strat{S}]{\alpha}}
        }
      }
    }
    :
    (
      x':\den[\DUAL]{\mk{F}}[\conns{G}]\many{X}
      \entails_{\conns{G}}^{\many{X:k}}
      \alpha:\mk{F}\many{X}
    )
    \\
    c'
    &=
    \cut
    {x}
    [\strat{S}]
    {
      \caseof
      {
        \many
        {
          \recv
          {p_i}
          {\cutsmall{\den[\DUAL]{p_i}[\conns{G}]}[\strat{S}]{\alpha'}}
        }
      }
    }
    :
    (
      x:\mk{F}\many{X}
      \entails_{\conns{G}}^{\many{X:k}}
      \alpha':\den[\DUAL]{\mk{F}}[\conns{G}]\many{X}
    )
  \end{align*}
  Both compositions are equal to the identity mapping.  For the identity on
  $\mk{F}$, we have
  \begin{align*}
    \cut
    {\outp{\alpha'}c'}
    [\strat{S}]
    {\inp{x'}c}
    &\eq[\etatmu]
    \cut
    {\outp{\alpha'}c'}
    {
      \caseof
      {
        \many
        {
          \recv
          {p_i}
          {\cutsmall{\den[\DUAL]{p_i}[\conns{G}]}[\strat{S}]{\alpha'}}
        }
      }
    }
    \\
    &\eq[\betamu]
    \cut
    {x}
    [\strat{S}]
    {
      \caseof
      {
        \many
        {
          \recv
          {p_i}
          {
            \cut
            {\den[\DUAL]{p_i}[\conns{G}]}
            [\strat{S}]
            {
              \caseofsmall
              {
                \many
                {
                  \recv
                  {\den[\DUAL]{p_i}[\conns{G}]}
                  {\cutsmall{p_i}[\strat{S}]{\alpha}}
                }
              }
            }
          }
        }
      }
    }
    \\
    &\eq[{\den[\DUAL]{\betar[p]}[\conns{G}]}]
    \cut
    {x}
    [\strat{S}]
    {
      \caseof
      {
        \many
        {
          \recv
          {p_i}
          {\cutsmall{\den[\DUAL]{p_i}[\conns{G}]}[\strat{S}]{\alpha'}}
        }
      }
    }
    \\
    &\eq[{\etar[p]}]
    \cut{x}[\strat{S}]{\alpha}
  \end{align*}
  And for the identity on the encoding $\den[\DUAL]{\mk{F}}[\conns{G}]$, we have
  \begin{align*}
    \cut{\outp\alpha c}[\strat{S}]{\inp x c'}
    &\eq[{\etatmu}]
    \cut
    {\outp\alpha c}
    [\strat{S}]
    {
      \caseof
      {
        \many
        {
          \recv
          {p_i}
          {\cutsmall{\den[\DUAL]{p_i}[\conns{G}]}[\strat{S}]{\alpha'}}
        }
      }
    }
    \\
    &\eq[{\betamu}]
    \cut
    {x'}
    [\strat{S}]
    {
      \caseof
      {
        \many
        {
          \recv
          {\den[\DUAL]{p_i}[\conns{G}]}
          {
            \cut
            {p_i}
            [\strat{S}]
            {
              \caseofsmall
              {
                \many
                {
                  \recv
                  {p_i}
                  {\cutsmall{\den[\DUAL]{p_i}[\conns{G}]}[\strat{S}]{\alpha'}}
                }
              }
            }
          }
        }
      }
    }
    \\
    &\eq[{\betar[p]}]
    \cut
    {x'}
    [\strat{S}]
    {
      \caseof
      {
        \many
        {
          \recv
          {\den[\DUAL]{p_i}[\conns{G}]}
          {\cutsmall{\den[\DUAL]{p_i}[\conns{G}]}[\strat{S}]{\alpha'}}
        }
      }
    }
    \\
    &\eq[{\den[\DUAL]{\etar[p]}[\conns{G}]}]
    \cut{x'}[\strat{S}]{\alpha'}
    &\qedhere
  \end{align*}
\end{proof}

\begin{lemma}[{\Co}data Encoding Isomorphism]
\label{thm:codata-encoding-iso}

Given any $\conns{G}$ extending $\DUAL$, and any {\co}data type declaration
\begin{align*}
  \begin{codata}{\mk{G}\many{(X:k)} : \strat{S}}
    &
    \typeobs[\many{Y_1:\strat{S}_1}]
    {\many{A_1:\strat{T}_1}}{\many{B_1:\strat{R}_1}}
    {\mk{O}_1}{\mk{G}\many{X}}{\strat{S}}
    \\
    &
    \typeobs[\many{Y_n:\strat{S}_n}]
    {\many{A_n:\strat{T}_1}}{\many{B_n:\strat{R}_n}}
    {\mk{O}_n}{\mk{G}\many{X}}{\strat{S}}
  \end{codata}
\end{align*}
in $\conns{G}$, it follows that
\begin{math}
  \semtypety
  {}
  {
    \mk{G}
    \iso
    \den[\DUAL]{\mk{G}}[\conns{G}]
  }
  {\many{k} \to \strat{S}}
\end{math}.
\end{lemma}
\begin{proof}
  Dual to the proof of \cref{thm:data-encoding-iso}.  
\end{proof}

\section{Type Safety}
\label{sec:type-safety}

A side effect of disciplines linking the static and dynamic semantics is that
the operational semantics is defined over \emph{typed} (open) commands.
However, there is a convenient way to generalize the operational semantics to
work with untyped terms as well.  In particular, the types in variable binding
annotations are irrelevant for computation; only the discipline symbol
$\strat{S}$ has any impact on evaluation.  An easy way to remove the restriction
of typing is to collapse types, which we can do by adding the following rule to
the kind system:
\begin{equation*}
  \infer[\mathit{Untype}]
  {\typety[\conns{G}]{\Theta}{A = B}{\strat{S}}}
  {
    \typety[\conns{G}]{\Theta}{A}{\strat{S}}
    &
    \typety[\conns{G}]{\Theta}{B}{\strat{S}}
  }
\end{equation*}
so that for any two types $A$ and $B$ of kind $\strat{S}$ and term $v$, we have
$v : A$ if and only if $v : B$, and similarly for {\co}terms.  Therefore, there
is no longer any reason to keep track of types, only disciplines, so we may as
well write commands as $\cut{v}[\strat{S}]{e}$ instead of
$\cut{v}[\ann{A}{\strat{S}}]{e}$.
\begin{definition}[Well-disciplined]
\label{def:well-disciplined}
The \emph{discipline system} of $\DUAL$ is the type system of $\DUAL$ extended
with the $\mathit{Untype}$, and we say that commands, terms, and {\co}terms
are \emph{well-disciplined} when they have some derivation in the discipline
system.
\end{definition}
The discipline system accepts strictly more expressions than the type system,
since it expresses general recursion in pure $\DUAL$ which is not well-typed.
For example, an analogue to the usual non-terminating $\lambda$-calculus term
$(\fn x x x) ~ (\fn x x x)$ can be written as the following well-disciplined
$\DUAL$ command using negation and the shift $\ToPos$ for converting the
negative value $x$ to the positive value $\wrap{x}$:
\begin{equation*}
  \cut
  {\cocaseof{\recv{\throw{\wrap{x}}}{\cut{x}[\CBN]{\throw{\wrap{x}}}}}}
  [\CBN]
  {\inp y \cut{y}[\CBN]{\throw{\wrap{y}}}}
\end{equation*}
Similarly, a fixed-point combinator corresponding to
$\fn f (\fn x f~(x~x))~(\fn x f~(x~x))$ can be expressed in the well-disciplined
calculus as follows, where we write $\app{W}{F}$ as shorthand for
$[\throw{W},F]$:
\begin{equation*}
  \fn{[\app{\wrap{f}}{\alpha}]}
  \cut
  {
    \fn{[\app{\wrap{x}}{\beta}]}
    \cut
    {f}
    [\CBN]
    {\app{\wrap{\outp\gamma\cut{x}[\CBN]{\app{\wrap{x}}{\gamma}}}}{\beta}}
  }
  [\CBN]
  {\inp g \cut{g}[\CBN]{\app{\wrap{g}}{\alpha}}}
\end{equation*}
This approach to untyped multi-discipline semantics, via relaxing a traditional
type system, follows from \cite{Zeilberger2009PhD}.  The discipline system can
also be captured directly via the grammar of terms itself in the style of
\cite{Levy2001PhD} and \cite{MunchMaccagnoni2013PhD}, which demonstrates that
disciplines can be easily inferred from the syntax of terms.  So in this sense,
the operational semantics in \cref{fig:dual-core-operation} can also be used as
a semantics for untyped, but well-disciplined, commands.

As a notational convenience, we say that $c \not\sred$ if there is no $c'$ such
that $c \not\sred c'$ (and dually $c \sred$ if there is a $c'$ such that
$c \sred c'$).  The first basic property of call-by-(\co)need evaluation is that
(\co)values do not have any standard reduction, and symmetrically, any $\mu$- or
$\tmu$-abstraction whose body has a standard reduction cannot be a (\co)value.
\begin{lemma}[(\Co)Need Normal Forms]
\label{thm:need-normal}

\begin{lemmaenum}
\item If $\inp x c$ is a $\CBNeed$ {\co}value then $c \not\sred$.
\item If $\outp \alpha c$ is a $\CoNeed$ value then $c \not\sred$.
\item If $H$ does not bind $x$ then
  $H[\cut{x}[\ann{A}\CBNeed]{E_\CBNeed}] \not\sred$.
\item If $H$ does not bind $\alpha$ then
  $H[\cut{V_\CoNeed}[\ann{A}\CoNeed]{\alpha}] \not\sred$.
\end{lemmaenum}
\end{lemma}
\begin{proof}
  (a) and (c) both follow simultaneously by induction on the definition of
  $\CBNeed$ {\co}values and heaps.  The interesting case is when we have
  $\cut{x}[\ann{A}\CBNeed]{\inp y H[\cut{y}[\ann{B}\CBNeed]{E_\CBNeed}]}$ where
  $H$ does not bind $y$: the $\betatmu[\CBVN]$ rule doesn't apply to the
  top-level command (because $\CBNeed \notin \{\CBV, \CBN, \CoNeed\}$), neither
  does the $\betatmuneed$ rule (because $x$ is a variable), and furthermore
  $H[\cut{y}[\ann{B}\CBNeed]{E_\CBNeed}] \not\sred$ by the inductive hypothesis.
  Likewise, (b) and (d) both follow simultaneously by induction on the
  definition of $\CoNeed$ values and heaps.
\end{proof}
\begin{corollary}[(\Co)Need Non-Normal Forms]
\label{thm:need-nonnormal}

If $c \sred c'$ then
\begin{lemmaenum}
\item $\inp x c$ is not a $\CBNeed$ {\co}value, and
\item $\outp \alpha c$ is not a $\CoNeed$ value.
\end{lemmaenum}
\end{corollary}
\begin{corollary}[Needed $\mu\tmu$-abstractions]
\label{thm:needed-abstractions}

\begin{lemmaenum}
\item If $\inp x c$ is a $\CBNeed$ {\co}value, then $x \in \neededvar(c)$.
\item If $\outp\alpha c$ is a $\CoNeed$ value, then $\alpha \in \neededvar(c)$.
\end{lemmaenum}
\end{corollary}
Note that there may be some $\mu$- or $\tmu$-abstractions that are not
(\co)values, but also cannot reduce.  For example, consider
$\inp x \cut{y}[\ann{A}\CBNeed]{\alpha}$: this {\co}term is not a {\co}value
because it ignores its input, but it needs to know the value of $y$ in order to
progress.  However, the above two properties are enough to guarantee that
standard reduction is always deterministic, even with call-by-(\co)need
evaluation steps.
\begin{lemma}[Determinism]
\label{thm:standard-determinism}

If $c_1 \unsred c \sred c_2$ then $c_1$ and $c_2$ are the same command.
\end{lemma}
\begin{proof}
  First, we show the base case when a reduction is applied to the top-level of a
  command, by induction on the syntax of $c$.  Note that for commands in which
  the $\betamu[\CBVN]$, $\betatmu[\CBVN]$, $\betar[p]$, or $\betar[q]$ rules
  apply directly, determinism is immediate because no other rules apply to the
  top of the command, and the only context $H$ such that $c = H[c']$ is the
  empty context.  The only remaining cases are for the $\betatmuneed$ and
  $\betamuconeed$ rules:
  \begin{itemize}
  \item ($\betatmuneed$)
    \begin{math}
      \cut{V_\CBNeed}[\ann{A}\CBNeed]{\inp x c}
      \sred
      c\subst{\anns{x}{A}{\CBNeed}}{V_\CBNeed}
    \end{math}
    where $\inp x c$ is a $\CBNeed$ {\co}value.  Note that $c \not\sred$ due to
    \cref{thm:need-normal}, so this $\betatmuneed$ reduction is the only one
    that applies.
  \item ($\betamuconeed$) follows dually to the above case.
  \end{itemize}

  Next, we must consider when a reduction is applied inside a heap, so that
  $c = H[c']$ and $c_1' \unsred c' \sred c_2'$, which follows by induction on
  $H$.  The base case is shown above, and the two remaining cases are for
  delayed bindings:
  \begin{itemize}
  \item Given $c = \cut{v}[\ann{A}\CBNeed]{\inp x H'[c']}$, then we know that
    $\inp x H'[c']$ is not a $\CBNeed$ {\co}value (\cref{thm:need-nonnormal}),
    so $\betatmuneed$ cannot fire to the top of this command, and the result
    follows by the inductive hypothesis.
  \item The dual case where $c = \cut{\outp \alpha H'[c']}{\ann{A}\CoNeed}{e}$
    follows dually to the above.
    \qedhere
  \end{itemize}
\end{proof}

\newcommandx{\typesubst}[8][1,2,6]
{(#3 \entails_{#1}^{#2} #4) \Rightarrow #5 : (#7 \entails_{#1}^{#6} #8)}

To demonstrate the closure of standard reduction under substitution, we need a
more formal way to determine that a substitution is well-disciplined or
well-typed.  The definition of the typing rules for substitutions, of the form
$\typesubst[\conns{G}][\Theta]{\Gamma}{\Delta}{\rho}[\Theta']{\Gamma'}{\Delta'}$
where $\rho$ transforms the sequent on the left in order the one on the right,
is:
\begin{gather*}
  \axiom
  {
    \typesubst[\conns{G}]
    [\Theta]{\Gamma}{\Delta}
    {\varepsilon}
    [\Theta]{\Gamma}{\Delta}
  }
  \\[1ex]
  \infer
  {
    \typesubst[\conns{G}]
    [\Theta]{\Gamma,x:A}{\Delta}
    {\rho\subst{\anns{x}{A}{\strat{S}}}{V_{\strat{S}}}}
    [\Theta']{\Gamma'}{\Delta'}
  }
  {
    \typesubst[\conns{G}]
    [\Theta]{\Gamma}{\Delta}
    {\rho}
    [\Theta']{\Gamma'}{\Delta'}
    &
    \typety[\conns{G}]{\Theta'}{A\subs{\rho}}{\strat{S}}
    &
    \typetmfoc[\conns{G}][\Theta']{\Gamma'}{\Delta'}{V_{\strat{S}}}{A\subs{\rho}}
  }
  \\[1ex]
  \infer
  {
    \typesubst[\conns{G}]
    [\Theta]{\Gamma}{\alpha:A,\Delta}
    {\rho\subst{\anns{\alpha}{A}{\strat{S}}}{E_{\strat{S}}}}
    [\Theta']{\Gamma'}{\Delta'}
  }
  {
    \typesubst[\conns{G}]
    [\Theta]{\Gamma}{\Delta}
    {\rho}
    [\Theta']{\Gamma'}{\Delta'}
    &
    \typety[\conns{G}]{\Theta'}{A\subs{\rho}}{\strat{S}}
    &
    \typecotmfoc[\conns{G}][\Theta']{\Gamma'}{\Delta'}{E_{\strat{S}}}{A\subs{\rho}}
  }
  \\[1ex]
  \infer
  {
    \typesubst[\conns{G}]
    [\Theta,X:k]{\Gamma}{\Delta}
    {\rho\subst{\ann{X}{k}}{A}}
    [\Theta']{\Gamma'}{\Delta'}
  }
  {
    \typety[\conns{G}]{\Theta'}{A}{k}
    &
    \typesubst[\conns{G}]
    [\Theta]{\Gamma\subst{\ann{X}{k}}{A}}{\Delta\subst{\ann{X}{k}}{A}}
    {\rho}
    [\Theta']{\Gamma'}{\Delta'}
  }
\end{gather*}
As with typing expressions (commands, terms, and {\co}terms), we can weaken the
above rules to only checking that a substitution is well-disciplined.  A
substitution $\rho$ is well-typed is well-disciplined if it has a typing
derivation possibly using the $\mathit{Untype}$ rule from
\cref{def:well-disciplined}.  With this typing and discipline system for
substitutions, we can see that
\begin{enumerate*}[(1)]
\item typing derivations,
\item the syntax of (\co)values, and
\item the reductions of the operational semantics
\end{enumerate*}
are all closed under a matching substitution.
\begin{lemma}[Typing Substitution Closure]
\label{thm:typing-closed-under-subst}

Let
$\typesubst[\conns{G}][\Theta]{\Gamma}{\Delta}{\rho}[\Theta']{\Gamma'}{\Delta'}$.
\begin{lemmaenum}
\item If $\typety[\conns{G}]{\Theta}{A}{k}$ then
  $\typety[\conns{G}]{\Theta'}{A\subs{\rho}}{k}$.
\item If $\typecmd[\conns{G}][\Theta]{\Gamma}{\Delta}{c}$ then
  $\typecmd[\conns{G}][\Theta']{\Gamma'}{\Delta'}{c\subs{\rho}}$.
\item If $\typetm[\conns{G}][\Theta]{\Gamma}{\Delta}{v}{A}$ then
  $\typetm[\conns{G}][\Theta']{\Gamma'}{\Delta'}{v\subs{\rho}}{A\subs{\rho}}$.
\item If $\typetmfoc[\conns{G}][\Theta]{\Gamma}{\Delta}{v}{A}$ then
  $\typetmfoc[\conns{G}][\Theta']{\Gamma'}{\Delta'}{v\subs{\rho}}{A\subs{\rho}}$.
\item If $\typecotm[\conns{G}][\Theta]{\Gamma}{\Delta}{e}{A}$ then
  $\typecotm[\conns{G}][\Theta']{\Gamma'}{\Delta'}{e\subs{\rho}}{A\subs{\rho}}$.
\item If $\typecotmfoc[\conns{G}][\Theta]{\Gamma}{\Delta}{e}{A}$ then
  $\typecotmfoc[\conns{G}][\Theta']{\Gamma'}{\Delta'}{e\subs{\rho}}{A\subs{\rho}}$.
\end{lemmaenum}
\end{lemma}
\begin{proof}
  By mutual induction on the typing derivation given in each part.  The
  interesting case is the base cases where $\rho$ replaces variable with a value
  or dually replaces a {\co}variable with a {\co}value.  In both of these cases,
  we have $x\subs{\rho} = V_{\strat{S}}$ or $\alpha\subs{\rho} = E_{\strat{S}}$,
  where the necessary typing derivation for $V_{\strat{S}}$ or $E_{\strat{S}}$
  can be found by induction on the typing derivation of the substitution $\rho$.
\end{proof}

\begin{corollary}[(\Co)Value Substitution Closure]
\label{thm:value-closed-under-subst}

Let
$\typesubst[\conns{G}][\Theta]{\Gamma}{\Delta}{\rho}[\Theta']{\Gamma'}{\Delta'}$
be a well-disciplined substitution (\ie possibly using the $\mathit{Untype}$
rule) and $\typety[\conns{G}]{\Theta}{A}{\strat{S}}$.
\begin{lemmaenum}
\item If $\typetmfoc[\conns{G}][\Theta]{\Gamma}{\Delta}{V_{\strat{S}}}{A}$ is a
  well-disciplined $\strat{S}$-value then so is
  $\typetmfoc[\conns{G}][\Theta']{\Gamma'}{\Delta'}{V_{\strat{S}}\subs{\rho}}{A\subs{\rho}}$.
\item If $\typecotmfoc[\conns{G}][\Theta]{\Gamma}{\Delta}{E_{\strat{S}}}{A}$ is
  a well-disciplined $\strat{S}$-{\co}value then so is
  $\typecotmfoc[\conns{G}][\Theta']{\Gamma'}{\Delta'}{E_{\strat{S}}\subs{\rho}}{A\subs{\rho}}$.
\end{lemmaenum}
\end{corollary}

\begin{lemma}[Step Substitution Closure]
\label{thm:sred-closed-under-subst}

Let $\typecmd[\conns{G}][\Theta]{\Gamma}{\Delta}{c}$ be a well-disciplined
command and
$\typesubst[\conns{G}][\Theta]{\Gamma}{\Delta}{\rho}[\Theta']{\Gamma'}{\Delta}'$
a well-disciplined substitution.  If $c \sred c'$ then
$c\subs{\rho} \sred c'\subs{\rho}$.
\end{lemma}
\begin{proof}
  By induction on the surrounding heap context surrounding the reduction of $c$
  and cases on the possible reduction rules applied.  Note that in each
  reduction rule which assumes that certain sub-(\co)terms are (\co)values,
  those (\co)values are closed under substitution
  (\cref{thm:value-closed-under-subst}), so the side-condition of the rule
  still applies.
\end{proof}

In \cref{sec:dual-core}, we defined the status of a command in terms of the
notion of its \emph{needed (\co)variables} (written as $\neededvar(c)$).
Importantly, determining a command's status is a decidable judgment with a
unique answer for every command.
\begin{property}[Unique status]
  \label{thm:unique-status}
  Every command $c$ is exactly one of
  \begin{enumerate*}[(1)]
  \item finished,
  \item stuck, or
  \item in progress (\ie $c \sred c'$ for some $c'$).
  \end{enumerate*}
\end{property}

\thmtypesafety*
\begin{proof}
  Preservation follows from \cref{thm:typing-closed-under-subst} by cases on the
  reduction rule applied and induction on the surrounding heap context.
  Progress follows by induction on the given typing derivation of
  $c = \cut{v}[\ann{A}{\strat{S}}]{e}$, which must end with a $\CutRule$ of the
  form
  \begin{gather*}
    \infer[\CutRule]
    {\typecmd[\conns{G}][\Theta]{\Gamma}{\Delta}{\cut{v}[\strat{S}]{e}}}
    {
      \typetm[\conns{G}][\Theta]{\Gamma}{\Delta}{v}{A}
      &
      \typety[\conns{G}]{\Theta}{A}{\strat{S}}
      &
      \typecotm[\conns{G}][\Theta]{\Gamma}{\Delta}{e}{A}
    }
  \end{gather*}
  and depends crucially on cases for the discipline symbol $\strat{S}$:
  \begin{itemize}
  \item $\strat{S} = \CBV$: First consider the term $v$.
    \begin{itemize}
    \item If $v = \outp\alpha c'$, then $c \sred[\betamu]$ for any $e$.
    \item Otherwise $v = W$, and we must next consider both $W$ and the
      {\co}term $e$.
      \begin{itemize}
      \item If $e = \inp x c'$, then $c \sred[\betatmu]$ for any $W$.
      \item If $e = \alpha$, then $\alpha \in \neededvar(c)$, so $c$ is finished
        for any $W$.
      \item If $e = \mk{O}\many{C}\many{V_{\strat{T}}}\many{E_{\strat{R}}}$,
        then by inversion on the type $A$, it must be that either $W$ is a
        $\lambda$-abstraction with a matching branch for $\mk{O}$ (in which case
        $c \sred[{\betar[q]}]$) or a variable $x$ (in which case
        $x \in \neededvar(c)$ so $c$ is finished).
      \item Otherwise $e = \caseof{\many[i]{\recv{p_i}{c'_i} \mid}}$.  Similarly
        by inversion on the type $A$, it must be that either $W$ is a
        constructor application matching one of $p_i$ (in which case
        $c \sred[{\betar[p]}]$) or a variable $x$ (in which case
        $x \in \neededvar(c)$ so $c$ is finished).
      \end{itemize}
    \end{itemize}
  \item $\strat{S} = \CBNeed$: First consider the {\co}term $e$.
    \begin{itemize}
    \item If $e = \inp x c'$ and not a {\co}value, then we must apply the
      inductive hypothesis to the sub-derivation of $c'$.
      \begin{itemize}
      \item If $c' \sred$ then $c \sred$ by closure under heap contexts.
      \item Otherwise $c'$ is finished with $x \notin \neededvar(c')$ (because
        then $\inp x c'$ would be a {\co}value by
        \cref{thm:needed-abstractions}, contradicting the assumption).  $c$ must
        also finished because $c \not\sred$ (since $\inp x c'$ is not a
        {\co}value, preventing $\betamu$ and $\betatmuneed$ reduction) and
        $\neededvar(c) = \neededvar(c')$ is not empty.
      \end{itemize}
    \item Otherwise $e = E_\CBNeed$, and we must next consider both $E_\CBNeed$
      and the term $v$.
      \begin{itemize}
      \item If $v = \outp\alpha c'$ then $c \sred[\betamu]$ for any $E_\CBNeed$.
      \item If $v = x$, then $x \in \neededvar(c)$, so $c$ is finished for any
        $E_\CBNeed$ (even a $\tmu$-{\co}value).
      \item If $v = \mk{K}\many{C}\many{E_{\strat{R}}}\many{V_{\strat{T}}}$,
        then by inversion on the type $A$, it must be that either
        $E_\CBNeed$ is a $\rlam$-abstraction with a matching branch for
        $\mk{K}$ (in which case $c \sred[{\betar[p]}]$), a $\tmu$-{\co}value (in
        which case $c \sred[{\betatmuneed}]$), or a {\co}variable $\alpha$ (in
        which case $\alpha \in \neededvar(c)$ so $c$ is finished).
      \item Otherwise $v = \cocaseof{\many[i]{\send{q_i}{c_i} \mid}}$.  By
        inversion on the type $A$, it must be that either $E_\CBNeed$ is a
        destructor application matching one of $q_i$ (in which case
        $c \sred[{\betar[q]}]$), a $\tmu$-{\co}value (in which case
        $c \sred[\betatmuneed]$), or a {\co}variable $\alpha$ (in which case
        $\alpha \in \neededvar(c)$ so $c$ is finished).
      \end{itemize}
    \end{itemize}
  \item $\strat{S} = \CBN$ or $\CoNeed$: follows dually to the above two cases.
    \qedhere
  \end{itemize}
\end{proof}

% The monolithic definition of needed (\co)variables, which was stated in terms of
% decomposing a command as a heap and the sub-command which attempts to access
% those (\co)variables, can be broken down into a more incremental, inductive
% definition.
% \begin{lemma}%[Need decomposition]
% \label{thm:need-decomposition}
% The set of needed (\co)variables of an expression is equivalent to the following
% inductive definition:
%   \begin{gather*}
%     \begin{aligned}
%       \neededvar\cut{v}[\ann{A}{\strat{S}}]{e}
%       &=
%       \neededvar[\strat{S}](v) \cup \neededvar[\strat{S}](e)
%       &
%       (\cut{v}[\ann{A}{\strat{S}}]{e} \not\sred)
%       \\
%       \neededvar\cut{v}[\ann{A}{\strat{S}}]{e}
%       &=
%       \{\}
%       &
%       (\cut{v}[\ann{A}{\strat{S}}]{e} \sred)
%     \end{aligned}
%     \\
%     \begin{aligned}
%       \neededvar[\strat{S}](x) &= \{x\}
%       &&&
%       \neededvar[\strat{S}](\alpha) &= \{\alpha\}
%       \\
%       \neededvar[\CoNeed](\outp\alpha c)
%       &=
%       \neededvar(c) - \{\alpha\}
%       &&&
%       \neededvar[\CBNeed](\inp x c)
%       &=
%       \neededvar(c) - \{x\}
%       \\
%       \neededvar[\strat{S}](v)
%       &= \{\}
%       &(\text{otherwise}&)
%       &
%       \neededvar[\strat{S}](e)
%       &= \{\}
%       &(\text{otherwise}&)
%     \end{aligned}
%   \end{gather*}
% \end{lemma}
% \begin{proof}
%   By induction on the syntax of expressions and the definition of
%   $\neededvar(c)$.
% \end{proof}

%%% Local Variables:
%%% mode: latex
%%% TeX-master: "journal"
%%% End:

\end{document}